\documentclass[a4paper,12pt]{article} 
\usepackage{amsmath,amsthm,amssymb}
\newtheorem{theorem}{Theorem}[section]

\newtheorem{proposition}[theorem]{Proposition}
\newtheorem{lemma}[theorem]{Lemma}
\newtheorem{corollary}[theorem]{Corollary}

\newcommand{\bm}[1]{\mbox{\boldmath$#1$}}
\newcommand{\be}{\begin{equation}}
\newcommand{\ee}{\end{equation}}
\newcommand{\bea}{\begin{eqnarray}}
\newcommand{\eea}{\end{eqnarray}}
\newcommand{\non}{\nonumber}
\newcommand{\ra}{\rangle}
\newcommand{\la}{\langle}

\begin{document} 

\title{
Correlation Functions 
of the integrable higher-spin 
XXX and XXZ spin chains through the fusion method}

\author{Tetsuo Deguchi$^{1}$\footnote{e-mail deguchi@phys.ocha.ac.jp} 
and Chihiro Matsui$^{2, 3}$
\footnote{e-mail matsui@spin.phys.s.u-tokyo.ac.jp}}
\date{ }
\maketitle
\begin{center}  $^{1}$ 
Department of Physics, Graduate School of Humanities and Sciences, 
Ochanomizu University \\
2-1-1 Ohtsuka, Bunkyo-ku, Tokyo 112-8610, Japan 
\end{center} 
\begin{center} $^{2}$ 
Department of Physics, Graduate School of Science, the University of Tokyo \\ 
7-3-1 Hongo, Bunkyo-ku, Tokyo 113-0033, Japan \\ 
$^{3}$ CREST, JST, 4-1-8 Honcho Kawaguchi, Saitama, 332-0012, Japan
\end{center}

\begin{abstract}
 For the integrable higher-spin XXX and XXZ spin chains 
we present multiple-integral representations for  
the correlation function of an arbitrary product of 
Hermitian elementary matrices in the massless ground state. 
We give a formula expressing it by a single term of multiple integrals. 
In particular, we explicitly derive 
the emptiness formation probability (EFP). 
We assume $2s$-strings for the ground-state solution 
of the Bethe ansatz equations for the spin-$s$ XXZ chain, 
and solve the integral equations for the spin-$s$ Gaudin matrix. 
 In terms of the XXZ coupling $\Delta$ we define $\zeta$ 
by  $\Delta=\cos \zeta$, 
and put it in a region $0 \le \zeta < \pi/2s$ 
of the gapless regime: $-1 < \Delta \le 1$ ($0 \le \zeta < \pi$),  
where $\Delta=1$ ($\zeta=0$) corresponds to the antiferromagnetic point.   
We calculate the zero-temperature correlation functions 
by the algebraic Bethe ansatz, introducing the Hermitian elementary 
matrices in the massless regime, and  
taking advantage of the fusion construction of 
the $R$-matrix of the higher-spin representations of the 
affine quantum group. 
\end{abstract}

\newpage
%
%
\setcounter{equation}{0} 
 \renewcommand{\theequation}{1.\arabic{equation}}
\section{Introduction}

The correlation functions of the spin-1/2 XXZ spin chain have 
 been studied extensively through the algebraic Bethe-ansatz 
during the last decade \cite{Korepin,Slavnov,MS2000,KMT1999,KMT2000,Review}. 
The multiple-integral representations 
of the correlation functions for the infinite lattice at zero temperature 
first derived through the affine quantum-group symmetry 
\cite{Miki,Jimbo-Miwa} and 
also by solving the $q$-KZ equations \cite{Jimbo-Miwa-qKZ,Takeyama} 
have been rederived and then generalized into those 
for the finite-size lattice under non-zero magnetic field. 
They are also extended into those at finite temperatures \cite{Goehmann-CF}. 
Furthermore, the asymptotic expansion of a correlation function 
has been systematically discussed \cite{KMTK2008}. 
Thus, the exact study of the correlation functions of the XXZ spin chain 
should be not only very fruitful but also quite fundamental  
in the mathematical physics of integrable models.

Recently, the correlation functions and form factors of the 
integrable higher-spin XXX spin chains and 
the form factors of the integrable higher-spin XXZ spin chains 
have been derived by the algebraic Bethe-ansatz method 
\cite{Kitanine2001,Castro-Alvaredo,DM1}.  
In the spin-1/2 XXZ chain the Hamiltonian 
under the periodic boundary conditions is given by  
\be 
{\cal H}_{\rm XXZ} =  
{\frac 1 2} \sum_{j=1}^{L} \left(\sigma_j^X \sigma_{j+1}^X +
 \sigma_j^Y \sigma_{j+1}^Y + \Delta \sigma_j^Z \sigma_{j+1}^Z  \right) \, . 
\label{hxxz}
\ee
Here $\sigma_j^{a}$ ($a=X, Y, Z$) are the Pauli matrices defined 
on the $j$th site and $\Delta$ denotes the XXZ coupling.   
We define  parameter $q$  by  
\be 
\Delta= (q+q^{-1})/2 \, . 
\ee
We define $\eta$ and $\zeta$ by $q=\exp \eta$ and $\eta= i \zeta$, 
respectively. We thus have $\Delta= \cos \zeta$. 
In the massless regime: $-1 < \Delta \le 1$, 
we have $0 \le  \zeta < \pi$ for the spin-1/2 XXZ spin chain (\ref{hxxz}). 
At $\Delta=1$ (i.e. $q=1$), the Hamiltonian (\ref{hxxz}) 
corresponds to the antiferromagnetic Heisenberg (XXX) chain. 
The solvable higher-spin generalizations of the XXX and XXZ spin chains 
have been studied by the fusion method in several references 
\cite{FusionXXX,BabujanPLA,Babujan,Fateev-Zamolodchikov,SAA,Babujan-Tsvelick,KR,V-DWA}. The spin-$s$ XXZ Hamiltonian is derived from the spin-$s$ fusion 
transfer matrix (see also section 2.6). 
For instance,  the Hamiltonian of the integrable 
spin-$1$ XXX spin chain is given by 
\be 
{\cal H}^{(2)}_{\rm XXX} = {\frac 1 2} \sum_{j=1}^{N_s} 
\left( {\vec S}_j \cdot {\vec S}_{j+1} 
- ({\vec S}_j \cdot {\vec S}_{j+1})^2 \right) \, . 
\label{eq:spin-1-XXX}
\ee
Here ${\vec S}_j$ denotes the spin-1 spin-angular momentum 
operator acting on the $j$th site among the $N_s$ lattice 
sites of the spin-$s$ chain. 
For the general spin-$s$ case, the integrable 
spin-$s$ XXX and XXZ Hamiltonians denoted 
${\cal H}^{(2s)}_{\rm XXX}$ and ${\cal H}^{(2s)}_{\rm XXZ}$, respectively,  
 can also be derived systematically.    

The correlation functions of integrable higher-spin XXX and XXZ spin chains 
are associated with various topics of mathematical physics. 
For the integrable spin-1 XXZ spin chain 
correlation functions have been derived 
by the method of $q$-vertex operators 
through some novel results of 
the representation theory of the quantum algebras
\cite{Idzumi0,Idzumi1,Idzumi2,Bougourzi,Konno}. 
They should  be closely related to the higher-spin solutions 
of the quantum Knizhnik-Zamolodchikov equations \cite{Takeyama}.  
For the  fusion eight-vertex models, 
correlation functions have  been discussed 
by an algebraic method \cite{Kojima}. 
Moreover, the  partition function of the six-vertex model 
under domain wall boundary conditions have been extended 
into the higher-spin case \cite{Foda}.  

In a massless region $0 \le \zeta < \pi/2s$,  
the low-lying excitation spectrum at zero temperature 
of the integrable spin-$s$ XXZ chain  
should correspond to the level-$k$ $SU(2)$ WZWN model with $k=2s$. 
By assuming the string hypothesis 
it is conjectured that the ground state 
of the integrable spin-$s$ XXX Hamiltonian  
is given by $N_s/2$ sets of $2s$-strings \cite{Takhtajan}.   
It has also been extended into the XXZ case \cite{Sogo}. 
The ground-state solution of $2s$-strings 
is derived for the spin-$s$ XXX chain  
through the zero-temperature limit of 
the thermal Bethe ansatz \cite{Babujan}.
The low-lying excitation spectrum  
is discussed in terms of spinons for the spin-$s$ XXX and XXZ 
spin chains \cite{Takhtajan,Sogo}. 
Numerically It was shown that 
the finite-size corrections to the ground-state energy 
of the integrable spin-$s$ XXX chain are consistent 
with the conformal field theory (CFT) with $c=3s/(s+1)$ 
\cite{Alcaraz-Martins:XXX,Affleck,Dorfel,Avdeev}. 
Here $c$ denotes the  central charge of the CFT. 
It is also the case with the integrable spin-$s$ XXZ chain 
in the region $0 \le \zeta < \pi/2s$ \cite{Alcaraz-Martins:XXZ,Fowler,Frahm}. 
The results are consistent with the conjecture 
that the ground state of the integrable spin-$s$ 
XXZ chain with $0 \le \zeta < \pi/2s$ 
is given by $N_s/2$ sets of $2s$-strings 
\cite{KR,Sogo,Alcaraz-Martins:XXZ,Fowler,Frahm,deVega-Woynarovich,KB,KBP}.  
Furthermore, it was shown analytically that the low-lying 
excitation spectrum of the integrable spin-$s$ XXZ chain in the 
region $0 \le \zeta < \pi/2s$ 
is consistent with the CFT of $c=3s/(s+1)$ \cite{KB,KBP}.  
In fact, the low-lying excitation spectrum of spinons 
for the spin-$s$ XXX chain  
is described in terms of the level-$k$ $SU(2)$ WZWN model 
with $k=2s$ \cite{JSuzuki}.

In the paper we calculate zero-temperature correlation functions   
for the integrable higher-spin XXZ spin chains 
by the algebraic Bethe-ansatz method. 
For a given product of elementary matrices 
we present the multiple-integral representations of the 
correlation function in the region $0 \le \zeta < \pi/2s$ 
of the massless regime near the antiferromagnetic point ($\zeta=0$).  
For an illustration, we derive the multiple-integral representations 
of the emptiness formation probability (EFP) of the 
spin-$s$ XXZ spin chain, explicitly.  
Here the spin $s$ is given by an arbitrary positive integer or half-integer. 
Assuming the conjecture that 
the ground-state solution of the Bethe ansatz equations 
is given by $2s$-strings for the regime of $\zeta$,  
we derive the spin-$s$ EFP for a finite chain 
and then take the thermodynamic limit. We solve the integral 
equations associated with the spin-$s$ Gaudin matrix 
for $0 \le \zeta < \pi/2s$,  
and express the diagonal elements  
in terms of the density of strings. 
Here we remark that the integral 
equations associated with the spin-$s$ Gaudin matrix have not been 
 explicitly solved, yet, even for the case of the integrable 
higher-spin XXX spin chains \cite{Kitanine2001}. 
We also calculate the spin-$s$ EFP for the homogeneous chain where 
all inhomogeneous parameters $\xi_p$ are given by zero. 
Here we shall introduce inhomogeneous parameters $\xi_p$ 
for $p=1, 2, \ldots, N_s$, in \S 2.4.  
Furthermore, we take advantage of the fusion construction 
of the spin-$s$ $R$-matrix in the algebraic Bethe-ansatz 
derivation of the correlation functions \cite{DM1}.

Given the spin-$s$ XXZ spin chain on the $N_s$ lattice sites, 
we define $L$ by $L=2s N_s$ and consider the spin-1/2 XXZ spin chain 
on the $L$ sites with inhomogeneous parameters $w_j$ 
for $j=1, 2, \ldots, L$. In the fusion method 
we express any given spin-$s$ local operator 
as a sum of products of operator-valued elements 
of the spin-1/2 monodromy matrix in the limit of 
sending inhomogeneous parameters $w_j$ to 
 sets of complete $2s$-strings as shown in Ref. \cite{DM1}. 
Here, we apply the spin-1/2 formula 
of the quantum inverse scattering problem  \cite{KMT1999},   
which is valid at least for generic inhomogeneous parameters.   
Therefore, sending inhomogeneous parameters $w_j$  
into complete $2s$-strings,   
we can evaluate the vacuum expectation values 
or the form factors of spin-$s$ local operators which are expressed 
in terms of the spin-1/2 monodromy matrix elements 
with generic inhomogeneous parameters $w_j$. 
Here, the rapidities of the ground state satisfy 
the Bethe ansatz equations with inhomogeneous parameters $w_j$.  
We assume in the paper  that the Bethe roots are 
continuous with respect to inhomogeneous parameters $w_j$, 
in particular,  in the limit of sending $w_j$ to complete $2s$-strings.

We can construct higher-spin transfer matrices 
by the fusion method \cite{KR,V-DWA}. 
Here we recall that the spin-1/2 XXZ Hamiltonian (\ref{hxxz}) 
is derived from the logarithmic derivative 
of the row-to-row transfer matrix of the six-vertex model. 
We call it the spin-1/2 transfer matrix and denote it 
by $t^{(1,1)}(\lambda)$. 
Let us express by $V^{(\ell)}$ an $(\ell+1)$-dimensional vector space.  
We denote by  $T^{(\ell, \, 2s)}(\lambda)$ 
the spin-$\ell/2$ monodromy matrix acting on 
the tensor product of the auxiliary space 
$V^{(\ell)}$ and the $N_s$th tensor product of the quantum spaces, 
 $(V^{(2s)})^{\otimes N_s}$. We call it of type 
$(\ell, (2s)^{\otimes N_s})$, which we express ($\ell, 2s$) in the 
superscript. Taking the trace of 
the spin-$\ell/2$ monodromy matrix $T^{(\ell, \, 2s)}(\lambda)$ 
over the auxiliary space  
$V^{(\ell)}$, we define the spin-$\ell/2$ transfer matrix, 
$t^{(\ell, \, 2s)}(\lambda)$. 
For $\ell=2s$, we have the spin-$s$ transfer matrix 
$t^{(2s, \, 2s)}(\lambda)$, 
and we derive the spin-$s$ Hamiltonian from its logarithmic derivative. 

We construct the ground state $| \psi_g^{(2s)} \rangle$ of the spin-$s$ 
XXZ Hamiltonian  by the $B$ operators of the $2$-by-$2$ monodromy matrix 
$T^{(1, 2s)}(\lambda)$. 
As shown by Babujan, the spin-$s$ transfer matrix $t^{(2s, 2s)}(\lambda)$ 
commutes with the spin-1/2 transfer matrix $t^{(1, 2s)}(\lambda)$ 
due to the Yang-Baxter relations, 
and hence they have eigenvectors in common \cite{Babujan}. 
The ground state $| \psi_g^{(2s)} \rangle$ of the spin-$s$ XXZ spin chain 
is originally an eigenvector of the spin-$s$ transfer matrix  
$t^{(2s, 2s)}(\lambda)$, and consequently it is also  
an eigenvector of the spin-$1/2$ transfer matrix $t^{(1, 2s)}(\lambda)$. 
Therefore, the ground state $| \psi^{(2s)}_g \rangle$ 
of the spin-$s$ XXZ spin chain can be constructed 
by applying the $B$ operators of the $2$-by-$2$ monodromy matrix 
$T^{(1, 2s)}(\lambda)$ to the vacuum.  

We can show that the fusion $R$-matrix corresponds 
to the $R$-matrix of the affine quantum group $U_q(\widehat{sl_2})$. 
We recall that by the fusion method, we can construct the $R$-matrix acting on 
the tensor product $V^{(\ell)} \otimes V^{(2s)}$ 
\cite{FusionXXX,BabujanPLA,Babujan,Fateev-Zamolodchikov,SAA,Babujan-Tsvelick,KR,V-DWA}. We denote it by $R^{(\ell, 2s)}$. 
In the affine quantum group, the $R$-matrix 
is defined as the intertwiner of 
the tensor product of two representations $V$ and $W$ 
\cite{Drinfeld,Jimbo-QG,Jimbo-Hecke}.  
Due to the conditions of the intertwiner 
the $R$-matrix of the affine quantum group is determined uniquely 
up to a scalar factor \cite{Jimbo-review},  
which we denote by $R_{V, W}$. 
Therefore, showing that the fusion $R$-matrix satisfies 
all the conditions of the intertwiner, we prove that the fusion $R$-matrix 
coincides with the $R$-matrix of the quantum group, $R_{V, W}$.   
Consequently, for $\ell=2s$ the fusion $R$-matrix, $R^{(2s,2s)}(\lambda)$, 
becomes the permutation operator when spectral parameter $\lambda$ 
is given by zero.  This property of the $R$-matrix  
plays a central role in the derivation of the integrable 
spin-$s$ Hamiltonian. It is also fundamental 
in the inverse scattering problem in the spin-$s$ case \cite{MT2000}.

There are several relevant and interesting studies 
of the integrable spin-$s$ XXZ spin chains. 
The expression of eigenvalues of the spin-$s$ XXX transfer matrix 
$t^{(2s, 2s)}(\lambda)$ was derived by Babujan 
\cite{BabujanPLA,Babujan,Babujan-Tsvelick} 
through the algebraic Bethe-ansatz method.  
It was also derived by solving the series 
of functional relations  
among the spin-$s$ transfer matrices \cite{KR}. 
The functional relations are systematically generalized to  
the $T$ systems \cite{Kuniba-Nakanishi-Suzuki}. 
Recently, the algebraic Bethe ansatz 
for the spin-$s$ XXZ transfer matrix 
has been thoroughly reviewed and reconstructed 
from the viewpoint of the algebraic Bethe ansatz 
of the $U(1)$-invariant integrable model 
\cite{Melo-Martins,Martins-Melo}. 
Quite interestingly, it has also been applied to 
construct the invariant subspaces associated with the 
Ising-like spectra of the superintegrable chiral Potts 
model \cite{Nishino}.

The content of the paper consists of the following. In section 2, 
we introduce the $R$-matrix for the spin-1/2 XXZ spin chain. 
We then introduce conjugate basis vectors in order to formulate 
Hermitian elementary matrices $\widetilde{E}^{m, \, n}$ 
in the massless regime where $|q|=1$. 
We define the massless higher-spin monodromy matrices 
$\widetilde{T}^{(\ell, \, 2s)}(\lambda)$ 
in terms of the conjugate vectors, 
 after reviewing the fusion construction of the massive 
higher-spin monodromy matrices 
$T^{(\ell, \, 2s)}(\lambda)$    
and higher-spin XXZ transfer matrices, 
$t^{(\ell, \, 2s)}(\lambda)$, for $\ell=1, 2, \ldots$, as follows.  
We express the matrix elements of 
$T^{(\ell, \, 2s)}_{0, 1 2 \cdots N_s}(\lambda)$    
in terms of those of the spin-1/2 monodromy matrix 
$T^{(1, \, 1)}_{0, 1 2 \cdots L}(\lambda)$.  
Here  $T^{(1, \, 1)}_{0, 1 2 \cdots L}(\lambda)$ is defined on 
the tensor product of the two-dimensional auxiliary space  
$V_0^{(1)}$ and the $L$th tensor product  
of the 2-dimensional quantum space, $(V^{(1)})^{\otimes L}$.   
Here we recall $L=2s \times N_s$.  
In the fusion construction \cite{DM1}, 
monodromy matrix  $T^{(1, \, 2s)}_{0, 1 2 \cdots N_s}(\lambda)$ 
acting on the $N_s$ lattice sites is 
derived from  monodromy matrix 
$T^{(1, \, 1)}_{0, 1 2 \cdots L}(\lambda)$ acting 
on the $2s N_s$ lattice sites by setting inhomogeneous parameters $w_j$ 
to $N_s$ sets of complete $2s$-strings and by multiplying it by 
the $N_s$th tensor product of projection operators which project 
$(V^{(1)})^{\otimes 2s}$ to $V^{(2s)}$. 
In section 3, we explain the method for 
calculating the expectation values of given 
products of spin-$s$ local operators. 
We express the local operators in terms of global operators with   
inhomogeneous parameters $w_j^{(2s; \, \epsilon)}$, which are defined to be  
close to complete $2s$-strings with small deviations of $O(\epsilon)$, 
and evaluate the scalar products and the expectation values 
for the Bethe state with the same 
inhomogeneous parameters $w_j^{(2s; \, \epsilon)}$.   
Then, we obtain the expectation values, sending $\epsilon$ to 0. 
Here we note that the projection operators 
introduced in the fusion construction 
commute with the matrix elements of monodromy 
matrix $T^{(1, \, 1)}$  of inhomogeneous parameters 
$w_j^{(2s; \, \epsilon)}$ with $O(\epsilon)$ corrections. 
In section 4 we calculate 
the emptiness formation probability (EFP) for the spin-$s$ XXZ spin chain 
for a large but finite chain, and then evaluate the 
matrix $S$ which is introduced for expressing  
the EFP of an infinite spin-$s$ XXZ chain 
in the massless regime with $\zeta < \pi/2s$. 
Here we solve explicitly the integral equations for the spin-$s$ 
Gaudin matrix, 
expressing the $2s$-strings of the ground-state solution 
systematically in terms of the string centers. 
In section 5, we present explicitly  
the multiple-integral representation of the spin-$s$ EFP. 
We also derive it for the inhomogeneous chain where 
all the inhomogeneous parameters $\xi_p$ are given by 0.  
For an illustration, we calculate the multiple-integral representation 
of spin-1 EFP for $m=1$, $\la \widetilde{E}^{2, \, 2} \ra$, explicitly. 
In the XXX limit, the value of 
$\la \widetilde{E}^{2, \, 2} \ra$ approaches 1/3, 
which is consistent 
with the XXX result of Ref. \cite{Kitanine2001}. 
In section 6, we present the multiple-integral representations of the 
integrable spin-$s$ XXZ correlation functions.  
 We express the correlation function 
of an arbitrary product of elementary matrices by 
a single term of multiple integrals.  
For instance, we calculate the multiple-integral representation 
of the spin-1 ground-state expectation value, 
$\la \widetilde{E}^{1, \, 1} \ra$, explicitly, 
and show that it is consistent with 
the value of spin-1 EFP in section 5, i.e. we show 
 $\la \widetilde{E}^{1, \, 1} \ra + 2 \la \widetilde{E}^{2, \, 2} \ra=1$.   
Finally in section 7, we give concluding remarks.

%
%
\setcounter{section}{1} 
\setcounter{equation}{0} 
 \renewcommand{\theequation}{2.\arabic{equation}}
\section{Fusion transfer matrices}

\subsection{$R$-matrix and the monodromy matrix of type $(1,1^{\otimes L})$}
Let us introduce the $R$-matrix of the XXZ spin chain 
\cite{Korepin,MS2000,KMT1999,KMT2000}. 
We denote by $e^{a, \, b}$ a unit matrix that has only one nonzero element 
equal to 1 at entry $(a, b)$ where $a, b= 0, 1$. 
Let $V_1$ and $V_2$ be two-dimensional vector spaces. 
The $R$-matrix acting on $V_1 \otimes V_2$ is given by  
\be 
{R}^{+}(\lambda_1-\lambda_2) = \sum_{a,b,c,d=0,1} 
R^{+}(u)^{a \, b}_{c \, d} \, 
\, e^{a, \, c} \otimes e^{b, \, d} = 
\left( 
\begin{array}{cccc}
1 & 0 & 0 & 0 \\ 
0 & b(u) & c^{-}(u) & 0 \\
0 & c^{+}(u) & b(u) & 0 \\
0 & 0 & 0 & 1 \\
\end{array} 
\right) \, \label{eq:R+},  
\ee
where $u=\lambda_1-\lambda_2$, 
$b(u) = \sinh u/\sinh(u + \eta)$ and 
$c^{\pm}(u) = \exp( \pm u) \sinh \eta/\sinh(u + \eta)$. 
In the massless regime, we set $\eta= i \zeta$ by a real number  
$\zeta$,   
and we have $\Delta= \cos \zeta$. In the paper we mainly consider 
the region $0 \le \zeta < \pi/2s$.  
In the massive regime, we assign $\eta$ a real nonzero number 
and we have $\Delta = \cosh \eta > 1$.  
 Here we remark that the $R^{+}(\lambda_1-\lambda_2)$ 
is compatible with the homogeneous grading 
of $U_q(\widehat{sl}_2)$, which is explained in Appendix A 
\cite{DM1}.

We denote by $R^{(p)}(u)$ or simply by $R(u)$ the symmetric $R$-matrix 
where $c^{\pm}(u)$ of (\ref{eq:R+}) are replaced by 
$c(u)= \sinh \eta/\sinh(u+\eta)$ \cite{DM1}. The symmetric $R$-matrix 
is compatible with the affine quantum group $U_q(\widehat{sl}_2)$ 
of the principal grading \cite{DM1}.

Let $s$ be an integer or a half-integer.  We shall mainly 
consider the tensor product 
$V_1^{(2s)} \otimes \cdots \otimes V_{N_s}^{(2s)}$ 
of $(2s+1)$-dimensional vector spaces $V^{(2s)}_j$ 
with $L= 2s N_s$. 
In general, we consider the tensor product 
$V_0^{(2s_0)} \otimes V_1^{(2s_1)} \otimes \cdots 
\otimes V_{r}^{(2s_{r})}$ with $2s_1 + \cdots + 2s_r = L$,  
where $V_{j}^{(2s_j)}$ have spectral parameters 
$\lambda_j$ for $j=1, 2, \ldots, r$.    
For a given set of matrix elements 
$A^{a, \, \alpha}_{b, \, \beta}$  for $a,b=0, 1, \ldots, 2s_j$ and  
$\alpha, \beta=0, 1, \ldots, 2s_k$,  
we define operator $A_{j, k}$ by 
\bea 
A_{j, k} & = & \sum_{a,b=1}^{\ell} \sum_{\alpha, \beta} 
A^{a, \, \alpha}_{b, \, \beta} I_0^{(2s_0)} \otimes
I_1^{(2s_1)} \otimes \cdots \otimes I_{j-1}^{(2s_{j-1})}  \non \\ 
& & \quad \otimes e^{a, b}_j \otimes 
I_{j+1}^{(2s_{j+1})}  \otimes \cdots \otimes I_{k-1}^{(2s_{k-1})}  
\otimes e_k^{\alpha, \beta}  \otimes I_{k+1}^{(2s_{k+1})} 
 \otimes \cdots \otimes I_{r}^{(2s_r)} \, .   
\label{defAjk}
\eea

We now consider the $(L+1)$th tensor product of spin-1/2 representations, 
which consists of the tensor product of auxiliary space $V_0^{(1)}$ and 
the $L$th tensor product of quantum spaces $V_j^{(1)}$ 
for $j=1, 2, \ldots, L$, i.e.    
$V_0^{(1)} \otimes \left( V_1^{(1)} \otimes \cdots \otimes V_L^{(1)} \right)$. 
We call it the tensor product of type $(1,1^{\otimes L})$ 
and denote it by the following symbol:   
\be 
(1,1^{\otimes L}) = (1, \overbrace{1, 1, \ldots, 1}^{L}) \, . 
\ee

Applying definition (\ref{defAjk}) 
for matrix elements $R(u)^{ab}_{cd}$ of a given $R$-matrix, 
we define  $R$-matrices 
$R_{j k}(\lambda_j, \lambda_k)=R_{j k}(\lambda_j-\lambda_k)$ 
 for integers $j$ and $k$ with $0\le j < k \le L$. 
For integers $j, k$ and $\ell$ with $0 \le j < k < \ell \le L$, 
the $R$-matrices satisfy the Yang-Baxter equations 
\be  
 R_{j k}(\lambda_j-\lambda_k) 
R_{j \ell}(\lambda_j-\lambda_{\ell})
R_{k \ell}(\lambda_k-\lambda_{\ell})
=
 R_{k \ell}(\lambda_k-\lambda_{\ell}) 
R_{j \ell}(\lambda_j-\lambda_{\ell})
R_{j k}(\lambda_j-\lambda_k) \, . \label{eq:YBE}
\ee

We define the monodromy matrix of type $(1,1^{\otimes L})$ 
associated with homogeneous grading by 
\be 
T^{(1, \, 1 \, +)}_{0, 1 2 \cdots L}(\lambda_0; w_1, w_2, \ldots, w_L) 
= R_{0L}^{+}(\lambda_0-w_L) 
\cdots R_{02}^{+}(\lambda_0-w_2) R_{01}^{+}(\lambda_0-w_1) \, .   
\ee
Here we have set $\lambda_j=w_j$ for $j=1, 2, \ldots, L$,  
where $w_j$ are arbitrary parameters.     
We call them inhomogeneous parameters. We have expressed the symbol of 
type $(1,1^{\otimes L})$ as $(1, \, 1)$ in superscript.  
The symbol $(1, \, 1 \, +)$ denotes that it is consistent with 
homogeneous grading. 
We express operator-valued matrix elements of 
the monodromy matrix as follows.  
\be
T^{(1, 1 \, +)}_{0, 1 2 \cdots L}(\lambda; \{ w_j \}_L ) = 
\left( 
\begin{array}{cc} 
A^{(1 +)}_{1 2 \cdots L}(\lambda; \{ w_j \}_L) & 
B^{(1 +)}_{1 2 \cdots L}(\lambda; \{ w_j \}_L) \\ 
C^{(1 +)}_{1 2 \cdots L}(\lambda; \{ w_j \}_L) & 
D^{(1 +)}_{1 2 \cdots L}(\lambda; \{ w_j \}_L)  
\end{array} 
\right) \, . 
\ee
Here $\{ w_j \}_L$ denotes the set of $L$ parameters, 
$w_1, w_2, \ldots, w_L$.  
We also denote the matrix elements of 
the monodromy matrix by   
$[T^{(1, 1 +)}_{0, 1 2 \cdots L}(\lambda; \{ w_j \}_L )]_{a,b}$ 
for $a,b=0,1$. 

We derive 
the monodromy matrix consistent with principal 
grading, $T^{(1, 1 \, p)}_{0, 1 2 \cdots L}(\lambda; \{ w_j \}_L )$,    
from that of homogeneous grading via similarity transformation 
$\chi_{0 1 \cdots L}$  as follows \cite{DM1}.  
 \bea
& & T^{(1, 1 \, +)}_{0, 1 2 \cdots L}(\lambda; \{ w_j \}_L )  =  
\chi_{0 1 2 \cdots L} 
T^{(1, 1 \, p)}_{0, 1 2 \cdots L}(\lambda; \{ w_j \}_L )
\chi_{0 1 2 \cdots L}^{-1}  
\non \\ 
&  & = 
\left( 
\begin{array}{cc} 
\chi_{1 2 \cdots L} 
A^{(1 \,p)}_{1 2 \cdots L}(\lambda; \{ w_j \}_L) 
\chi_{1 2 \cdots L}^{-1} & 
e^{- \lambda_0} \chi_{1 2 \cdots L} 
B^{(1 \, p)}_{1 2 \cdots L}(\lambda; \{ w_j \}_L) 
\chi_{1 2 \cdots L}^{-1}  \\ 
e^{\lambda_0}
\chi_{1 2 \cdots L} 
C^{(1 \, p)}_{1 2 \cdots L}(\lambda; \{ w_j \}_L) 
\chi_{1 2 \cdots L}^{-1}  & 
\chi_{1 2 \cdots L} D^{(1 \, p)}_{1 2 \cdots L}(\lambda; \{ w_j \}_L) 
\chi_{1 2 \cdots L}^{-1}   
\end{array} 
\right) \, . \label{eq:gauge-transform}
\eea
Here $\chi_{01 \cdots L}= \Phi_0 \Phi_1 \cdots\Phi_L$ and  
$\Phi_j$ are given by  diagonal two-by-two matrices 
$\Phi_j={\rm diag}(1, \exp(w_j))$ acting on $V_j^{(1)}$ 
for $j=0, 1, \ldots, L$, and we set $w_0= \lambda_0$.  
In Ref. \cite{DM1} operator 
$A^{(1 \, +)}(\lambda)$ has been written as  $A^{+}(\lambda)$.  
Hereafter we shall often abbreviate the symbols $p$ 
in superscripts which shows the principal grading, and denote 
$(2s \, p)$ simply by $(2s)$.   

Let us introduce useful notation for expressing  
products of $R$-matrices as follows. 
\bea 
R_{1, 23 \cdots n}^{(w)} & = & R_{1 n}^{(w)}  
\cdots R_{13}^{(w)} R_{12}^{(w)}  \, , \non \\
R_{12 \cdots n-1, n}^{(w)}  & = & 
R_{1 n}^{(w)}  R_{2 n}^{(w)}  \cdots R_{n-1 \, n}^{(w)}  \, . 
\label{eq:useful-notation}
\eea
Here $R_{a b}^{(w)} $ denote the $R$-matrix 
$R_{a b}^{(w)} =R_{a b}^{(w)} (\lambda_a-\lambda_b)$ 
for $a, b =1, 2, \ldots, n$, where $w=+$ and $w=p$ in superscripts 
show the homogeneous and the principal grading, respectively.   
Then,    the monodromy matrix of type $(1, 1^{\otimes L} \, w)$ 
is expressed as follows.   
\bea 
T_{0, \, 1 2 \cdots L}^{(1, \, 1 \, w)}(\lambda_0; \{w_j \}_L) & = & 
R_{0, \, 1 2 \cdots L}^{(\, w)}(\lambda_0; \{w_j \}_L) \non \\ 
& = & R_{0 L}^{(w)} R_{0 L-1}^{(w)} \cdots R_{0 1}^{(w)}  \, . 
\eea
For instance we have 
$B_{12 \cdots L}^{(1 \, w)}(\lambda_0; \{ w_j \}_L)
=[R_{0, 1 2 \cdots L}^{(1 \, w)}(\lambda_0; \{w_j \}_L)]_{0,1} \, . $

%
%
\subsection{Projection operators and the massive fusion $R$-matrices}

Let  $V_1$ and $V_2$ be $(2s+1)$-dimensional vector spaces.     
We define permutation operator $\Pi_{1, \, 2}$ by   
\be 
\Pi_{1, \, 2} \, v_1 \otimes v_2 = 
v_2 \otimes v_1 \, , \quad v_1 \in V_1 \, , \, v_2 \in V_2 \, .  
\label{eq:Pi-2s-2s}
\ee
In the case of spin-1/2 representations,  
we define operator ${\check R}_{12}^{+}(\lambda_1-\lambda_2)$ by 
\be 
{\check  R}_{12}^{+}(\lambda_1-\lambda_2)
= \Pi_{1, \, 2} \,  R_{12}^{+}(\lambda_1 - \lambda_2) \, . 
\label{eq:Rcheck-spin-1/2}
\ee

We now introduce projection operators $P_{12\cdots \ell}^{(\ell)}$ for 
$\ell \ge 2$. 
We define  $P_{12}^{(2)}$ by $P_{12}^{(2)} = {\check R}_{1, 2}^{+}(\eta)$. 
 For $\ell > 2$ we define projection operators  
 inductively with respect to $\ell$ as follows 
\cite{Jimbo-Hecke,V-DWA}. 
\be 
P_{1 2 \cdots \ell}^{(\ell)} = 
P_{1 2 \cdots \ell-1}^{(\ell-1)} {\check R}^{+}_{\ell-1, \, \ell}
((\ell-1)\eta) P_{12\cdots \ell-1}^{(\ell-1)} \, . 
\label{eq:def-projector}
\ee 
The projection operator $P_{12\cdots \ell}^{(\ell)}$ 
gives a $q$-analogue of the full symmetrizer 
of the Young operators for the Hecke algebra \cite{Jimbo-Hecke}.   
We shall show the idempotency: 
$\left( P_{12\cdots \ell}^{(\ell)} \right)^2=P_{12\cdots \ell}^{(\ell)}$ 
in Appendix B. 
Hereafter we denote $P_{1 2 \cdots \ell}^{(\ell)}$ also by 
$P_{1}^{(\ell)}$ for short.

Applying projection operator $P_{a_1 a_2 \cdots a_{\ell}}^{(\ell)}$ 
to  vectors in the tensor product $V_{a_1}^{(1)} \otimes V_{a_2}^{(1)} 
\otimes \cdots \otimes V_{a_{\ell}}^{(1)}$, 
we can construct 
the $(\ell+1)$-dimensional vector space 
$V_{a_1 a_2 \cdots a_{\ell}}^{(\ell)}$ associated with 
the spin-$\ell/2$ representation of $U_q(sl_2)$. 
For instance, we have   
$P_{a_1 a_2}^{(2)} |+ - \rangle_{a} = (q/[2]_q)||2, 1 \rangle_{a}$,  
where we have introduced 
$| + - \rangle_{a} = |0 \rangle_{a_1} \otimes |1 \rangle_{a_2}$.  
The symbols such as $q$-integers are defined in Appendix C. 
Moreover, the basis vectors $|| \ell, n \rangle $ ($n=0, 1, \ldots, \ell$) 
and their dual vectors  $\langle \ell, n ||$  
are given for arbitrary nonzero integers $\ell$ 
in Appendix C. 
We denote $V_{a_1 a_2 \cdots a_{\ell}}^{(\ell)}$ 
also by $V_{a}^{(\ell)}$ or $V_{0}^{(\ell)}$ for short. 

Since 
$P^{(\ell)}_{1 2 \cdots \ell}$ is consistent   
with the spin-$\ell/2$ representation of $U_q(sl(2))$ 
(see (\ref{eq:consistency-P})), 
we have    
\be 
P^{(\ell)}_{1 2 \cdots \ell} 
= \sum_{n=0}^{\ell}  || \ell, n \rangle \, \langle \ell, n || \, . 
\label{eq:Psum}
\ee 
Applying projection operator 
$P_{2s(j-1)+1 \cdots 2s(j-1)+2s}^{(2s)}$ to tensor product 
$V_{2s(j-1)+1}^{(1)} \otimes \cdots \otimes V_{2s(j-1)+2s}^{(1)}$,  
we construct the spin-$s$ representation 
$V_{2s(j-1)+1 \cdots 2s(j-1)+2s}^{(2s)}$. We denote it also by 
$V_j^{(2s)}$, briefly.

In the tensor product of quantum spaces 
$V^{(2s)}_1 \otimes \cdots \otimes V_{N_s}^{(2s)}$,  
 we define $P_{12 \cdots L}^{(2s)}$  
 by 
\be  
P_{12 \cdots L}^{(2s)} = \prod_{i=1}^{N_s} P^{(2s)}_{2s(i-1)+1}  \, . 
\label{eq:PL}
\ee
Here we recall $L = 2s N_s $. 
We have put $2s$ in place of $\ell$.

We now introduce the massive fusion $R$-matrix 
$R^{(\ell, \, 2s \, +)}_{0, \, j}$ 
on the tensor product 
$V_{0}^{(\ell)} \otimes V_{j}^{(2s)}$ ($j=1, 2, \ldots, N_s$). 
It is valid in the massive regime with $\Delta > 1$.  
We first set rapidities $\lambda_{a_j}$ 
of auxiliary spaces $V_{a_j}^{(1)}$ 
by $\lambda_{a_{k}} = \lambda_{a_1} - (k-1) \eta$ for 
$k=1, 2, \ldots, \ell-1$, and then    
rapidities $\lambda_{2s(j-1)+k}$  
of quantum spaces $V_{2s(j-1)+k}^{(1)}$ by 
$\lambda_{2s(j-1)+k} = \lambda_{2s(j-1)+1} - (k-1) \eta$ 
for $k= 1, 2, \ldots, 2s$ and $j=1, 2, \ldots, N_s$.  
We define the massive fusion $R$-matrix 
$R^{(\ell, \, 2s \, +)}_{0, \, j}$ 
as follows.  
\bea 
& & R_{0 \, j}^{(\ell, \, 2s \, +)}(\lambda_{a_1}-\lambda_{2s(j-1)+1})  =  
P_{a_1 \cdots a_{\ell}}^{(\ell)} P_{2s(j-1)+1}^{(2s)} \, 
R^{+}_{a_1 \cdots a_{\ell}, \, 2s(j-1)+1 \cdots 2s j}  \, 
P_{a_1 \cdots a_{\ell}}^{(\ell)} P_{2s(j-1)+1}^{(2s)} \non \\ 
 &  & \qquad = 
P_{a_1 \cdots a_{\ell}}^{(\ell)} P_{2s(j-1)+1}^{(2s)} \, 
R^{+}_{a_1 \cdots a_{\ell}, \, 2s j} \cdots 
R^{+}_{a_1 \cdots a_{\ell}, \, 2s(j-1)+2} 
R^{+}_{a_1 \cdots a_{\ell}, \, 2s(j-1)+1}  \, 
P_{a_1 \cdots a_{\ell}}^{(\ell)} P_{2s(j-1)+1}^{(2s)} \, . 
\non \\
\eea

%
%
\subsection{Conjugate vectors and the massless 
fusion $R$-matrices} 

In order to construct Hermitian elementary matrices 
in the massless regime where $|q|=1$, 
we now introduce vectors $\widetilde{|| \ell, n \rangle}$ 
which are Hermitian conjugate to $\langle \ell, n ||$ when 
$|q|=1$ 
for positive integers $\ell$ with $n=0, 1, \ldots, \ell$. 
Setting the norm of $\widetilde{|| \ell, n \rangle}$ 
such that 
$\langle \ell, n || \,  \widetilde{|| \ell, n \rangle}=1$, 
we have   
\be 
\widetilde{|| \ell, n \rangle} = 
\sum_{1 \le i_1 < \cdots < i_n \le \ell} \sigma_{i_1}^{-} 
\cdots \sigma_{i_n}^{-} | 0 \rangle 
q^{-(i_1 + \cdots + i_n) + n \ell - n(n-1)/2} 
\left[ 
\begin{array}{cc} 
\ell \\ 
n 
\end{array} 
 \right]_q \, 
q^{-n(\ell-n)} 
\left( 
\begin{array}{cc} 
\ell \\ 
n 
\end{array} 
 \right)^{-1} \, . 
\ee
Here we have denoted the binomial coefficients 
for integers $\ell$ and $n$ with $0 \le n \le \ell$ as follows. 
\be 
\left( 
\begin{array}{cc} 
\ell \\ 
n 
\end{array} 
 \right)
= {\frac {\ell !} {(\ell-n)! n!}} \, .  
\ee
The $q$-binomial coefficients are defined in Appendix C. 
Dual vectors $\widetilde{\langle \ell, n ||}$, which are 
conjugate to $|| \ell, n \rangle$, are defined in Appendix C, 
and we have 
\be
\widetilde{\langle \ell, n ||} \, \widetilde{|| \ell, n \rangle} 
= \left[ 
\begin{array}{cc} 
\ell \\ 
n 
\end{array} 
 \right]_q^2 \,  
\left( 
\begin{array}{cc} 
\ell \\ 
n 
\end{array} 
 \right)^{-2} \, .  
\ee
They are determined by 
the action of $X^{\pm}$ with opposite coproduct: 
 $\Delta^{op}= \tau \circ \Delta$.  For instance, we have 
$\widetilde{|| \ell, n \rangle} = {\rm const.} \,  \Delta^{op}(X^{-})^{n} || \ell, 0 \rangle$. Here $X^{\pm}$ and $\Delta^{op}$ are defined in Appendix A.

For an illustration, 
in the spin-1 case,  
the basis vectors $|| 2, n  \rangle$ ($n=0, 1, 2$)      
 are given by \cite{DM1} 
\bea 
||2, 0 \rangle & = & |+ + \rangle \, , \non \\ 
||2, 1 \rangle & = & |+ - \rangle + q^{-1} | - + \rangle \, ,   \non \\ 
||2, 2 \rangle & = & |- - \rangle \, .  
\eea 
Here $| + - \rangle $ denotes $|0 \rangle_1 \otimes | 1 \rangle_2$,  
briefly.   
The conjugate vectors $\widetilde{|| 2, n \rangle}$ ($n=0, 1, 2$) 
are given by 
\bea 
\widetilde{||2, 0 \rangle} & = & |+ + \rangle \, , \non \\ 
\widetilde{||2, 1 \rangle} & = & \left( |+ -  \rangle  
+ q | - + \rangle \right) {\frac {[2]_q} {2q}} \, ,   \non \\ 
\widetilde{||2, 2 \rangle} & = & |- - \rangle \, . 
\eea
In the massless regime, 
operator $\widetilde{||2, 1 \rangle} \langle 2, 1 ||$ is Hermitian  
while $||2, 1 \rangle \langle 2, 1 ||$ is not.

Let us now introduce another set of projection operators 
$\widetilde{P}_{1 \cdots \ell}^{(\ell)}$ 
as follows. 
\be
\widetilde{ P}_{1 \cdots \ell}^{(\ell)} = \sum_{n=0}^{\ell} 
\widetilde{ || \ell , \,  n \rangle} \langle \ell , \, n ||  \, . 
\label{eq:P'sum}
\ee
Projector $\widetilde{ P}_{1 \cdots \ell}^{(\ell)}$ is idempotent: 
$(\widetilde{ P}_{1 \cdots \ell}^{(\ell)})^2= 
\widetilde{ P}_{1 \cdots \ell}^{(\ell)}$.   
In the massless regime where $|q|=1$, 
it is Hermitian:  
$\left( \widetilde{ P}_{1 \cdots \ell}^{(\ell)} \right)^{\dagger}= 
\widetilde{ P}_{1 \cdots \ell}^{(\ell)}$.  
From (\ref{eq:Psum}) and (\ref{eq:P'sum}), 
we show the following properties: 
\bea 
P_{1 2 \cdots \ell}^{(\ell)} 
\widetilde{ P}_{1 \cdots \ell}^{(\ell)}
& = & P_{1 2 \cdots \ell}^{(\ell)} \, ,  
\label{eq:PP'=P} \\
\widetilde{ P}_{1 \cdots \ell}^{(\ell)}
P_{1 2 \cdots \ell}^{(\ell)} 
& = & \widetilde{ P}_{1 \cdots \ell}^{(\ell)} \, . 
\label{eq:P'P=P'}
\eea
In the tensor product of quantum spaces, 
$V^{(2s)}_1 \otimes \cdots \otimes V_{N_s}^{(2s)}$,   
 we define $\widetilde{P}_{12 \cdots L}^{(2s)}$ by 
\be  
\widetilde{P}_{12 \cdots L}^{(2s)} 
= \prod_{i=1}^{N_s} \widetilde{P}^{(2s)}_{2s(i-1)+1}  \, . 
\ee
Here we recall $L = 2s N_s $ such as for (\ref{eq:PL}).

We define the massless fusion $R$-matrix 
$\widetilde{R}^{(\ell, \, 2s \, +)}_{0, \, j}$, 
applying projection operators $\widetilde{P}$ 
consisting of conjugate vectors to 
the product of $R$-matrices, as follows. 
\bea 
& & \widetilde{R}_{0 \, j}^{(\ell, \, 2s \, +)}
(\lambda_{a_1}-w_{2s(j-1)+1})  =  
\widetilde{P}_{a_1 \cdots a_{\ell}}^{(\ell)} 
\widetilde{P}_{2s(j-1)+1}^{(2s)} \, 
R^{+}_{a_1 \cdots a_{\ell}, \, 2s(j-1)+1 \cdots 2s j}  \, 
\widetilde{P}_{a_1 \cdots a_{\ell}}^{(\ell)} 
\widetilde{P}_{2s(j-1)+1}^{(2s)} \non \\ 
 &  & \qquad = 
\widetilde{P}_{a_1 \cdots a_{\ell}}^{(\ell)} 
\widetilde{P}_{2s(j-1)+1}^{(2s)} \, 
R^{+}_{a_1 \cdots a_{\ell}, \, 2s j} \cdots 
R^{+}_{a_1 \cdots a_{\ell}, \, 2s(j-1)+2} 
R^{+}_{a_1 \cdots a_{\ell}, \, 2s(j-1)+1}  \, 
\widetilde{P}_{a_1 \cdots a_{\ell}}^{(\ell)} 
\widetilde{P}_{2s(j-1)+1}^{(2s)} \, . 
\non \\
\eea
We should remark that the massless fusion $R$-matrix 
$\widetilde{R}^{(\ell, \, 2s)}$  
and the massive fusion $R$-matrix ${R}^{(\ell, \, 2s)}$
 have the same matrix elements.  
Some examples are shown in Appendix D. 

%
%
\subsection{Higher-spin monodromy matrix of type 
$(\ell, \, (2s)^{\otimes N_s})$ }

We now set the inhomogeneous parameters $w_j$ for $j=1, 2, \ldots, L$,  
as $N_s$ sets of complete $2s$-strings \cite{DM1}. 
We define $w_{(b-1)\ell+ \beta}^{(2s)}$ for $\beta = 1, \ldots, 2s$, 
as follows.  
\be 
w_{2s(b-1)+ \beta}^{(2s)} = \xi_b - (\beta-1) \eta \, , \quad 
 \mbox{for} \quad b = 1, 2, \ldots, N_s . 
\label{eq:ell-strings}
\ee
We shall define the monodromy matrix of type $(1, (2s)^{\otimes N_s})$ 
associated with homogeneous grading.  
We first define the massless monodromy matrix by  
\bea 
\widetilde{T}^{(1, \, 2s \, +)}_{0, \, 1 2 \cdots N_s}
(\lambda_0; \{ \xi_b \}_{N_s} ) 
& = & \widetilde{P}_{12 \cdots L}^{(2s)} 
R_{0, \, 1 \ldots L}^{(1, \, 1 \, +)} 
(\lambda_0; \{ w_{j}^{(2s)} \}_L) 
\widetilde{P}_{12 \cdots L}^{(2s)}  \non \\  
& = & 
\left( 
\begin{array}{cc} 
\widetilde{A}^{(2s +)}(\lambda; \{ \xi_b \}_{N_s}) & 
\widetilde{B}^{(2s +)}(\lambda; \{ \xi_b \}_{N_s}) \\ 
\widetilde{C}^{(2s +)}(\lambda; \{ \xi_b \}_{N_s}) & 
\widetilde{D}^{(2s +)}(\lambda; \{ \xi_b \}_{N_s})  
\end{array} 
\right) \, . 
\eea
Here, the (0,0) element is given by 
$ \widetilde{A}^{(2s +)}(\lambda; \{ \xi_b \}_{N_s})= 
\widetilde{P}_{12 \cdots L}^{(2s)}
A^{(1 +)}(\lambda; \{ w_j^{(2s)} \}_{L})
\widetilde{P}_{12 \cdots L}^{(2s)}$. 
We then define the massive monodromy matrix by 
\bea 
T^{(1, \, 2s \, +)}_{0, \, 1 2 \cdots N_s}
(\lambda_0; \{ \xi_b \}_{N_s} ) 
& = & {P}_{12 \cdots L}^{(2s)} 
R_{0, \, 1 \ldots L}^{(1, \, 1 \, +)} 
(\lambda_0; \{ w_{j}^{(2s)} \}_L) 
{P}_{12 \cdots L}^{(2s)}  \non \\  
& = & 
\left( 
\begin{array}{cc} 
{A}^{(2s +)}(\lambda; \{ \xi_b \}_{N_s}) & 
{B}^{(2s +)}(\lambda; \{ \xi_b \}_{N_s}) \\ 
{C}^{(2s +)}(\lambda; \{ \xi_b \}_{N_s}) & 
{D}^{(2s +)}(\lambda; \{ \xi_b \}_{N_s})  
\end{array} 
\right) \, . 
\eea

Let us introduce a set of $2s$-strings with small deviations from 
the set of complete $2s$-strings.   
\be 
w_{2s(b-1)+ \beta}^{(2s; \, \epsilon)} = \xi_b - (\beta-1) \eta  
+ \epsilon r_b^{(\beta)} \, , \quad 
 \mbox{for} \quad b=1, 2, \cdots, N_s, \, \mbox{and} \quad 
\beta=1, 2, \ldots, 2s. 
\label{eq:ell-strings-epsilon}
\ee
Here $\epsilon$ is very small and 
$r_{b}^{(\beta)}$ are generic parameters.  
We express the elements of the monodromy matrix $T^{(1,1)}$ with 
inhomogeneous parameters given by $w_j^{(2s; \, \epsilon)}$ 
for $j=1, 2, \ldots, L$ as follows. 
\be
T^{(1, \, 1 \, +)}_{0, \, 1 2 \cdots L}
(\lambda; \{ w_j^{(2s; \epsilon)} \}_{L}) = 
\left( 
\begin{array}{cc} 
A^{(2s +; \, \epsilon)}_{1 2 \cdots L}(\lambda) & 
B^{(2s +; \, \epsilon)}_{1 2 \cdots L}(\lambda) \\ 
C^{(2s +; \, \epsilon)}_{1 2 \cdots L}(\lambda) & 
D^{(2s +; \, \epsilon)}_{1 2 \cdots L}(\lambda)  
\end{array} 
\right) \, . 
\ee
Here we recall that 
$A^{(2s + ;  \, \epsilon)}_{1 2 \cdots L}(\lambda)$ denotes 
$A^{(1+)}_{1 2 \cdots L}(\lambda; \{ w_j^{(2s;  \, \epsilon)} \}_L)$. 
We also remark the following.   
\be 
\widetilde{ A}^{(2s +)}_{1 2 \cdots N_s}(\lambda; \{\xi_p \}_{N_s}) = 
\lim_{\epsilon \rightarrow 0} 
\widetilde{P}_{1 2 \cdots L}^{(2s)}  
A^{(2s +; \, \epsilon)}_{1 2 \cdots L}
(\lambda; \{w_j^{(2s; \, \epsilon)} \}_{L}) 
\widetilde{P}_{1 2 \cdots L}^{(2s)}  \, . 
\ee

Let us express the tensor product  
$V_0^{(\ell)} \otimes \left( V_1^{(2s)} \otimes \cdots 
\otimes V_{N_s}^{(2s)} \right)$,  by the following symbol   
\be 
(\ell, \, (2s)^{\otimes N_s}) 
= (\ell, \, \overbrace{2s, 2s, \ldots, 2s}^{N_s}) \, . 
\ee
Here we recall that $V_0^{(\ell)}$ abbreviates 
$V_{a_1 a_2 \ldots a_{\ell}}^{(\ell)}$.  
In the case of auxiliary space $V_0^{(\ell)}$ 
we define the massless monodromy matrix of type 
$(\ell, \, (2s)^{\otimes N_s})$ by 
\be 
\widetilde{T}^{(\ell, \, 2s \, +)}_{0, \, 1 2 \cdots N_s} 
 =  \widetilde{P}^{(\ell)}_{a_1 a_2 \cdots a_{\ell}} \,   
\widetilde{T}_{a_1, \, 1 2 \cdots N_s}^{(1, \, 2s \, +)}(\lambda_{a_1}) 
\widetilde{T}_{a_2, \, 1 2 \cdots N_s}^{(1, \, 2s \, +)}(\lambda_{a_1}-\eta) 
\cdots 
\widetilde{T}_{a_{\ell}, \, 1 2 \cdots N_s}^{(1, \,  2s \, +)}
(\lambda_{a_1}-(\ell-1)\eta) \,  
\widetilde{P^{(\ell)}}_{a_1 a_2 \cdots a_{\ell}} \, , 
\ee 
and the massive monodromy matrix of type 
$(\ell, \, (2s)^{\otimes N_s})$ by 
\be 
T^{(\ell, \, 2s \, +)}_{0, \, 1 2 \cdots N_s} 
 =  P^{(\ell)}_{a_1 a_2 \cdots a_{\ell}} \,   
T_{a_1, \, 1 2 \cdots N_s}^{(1, \, 2s \, +)}(\lambda_{a_1}) 
T_{a_2, \, 1 2 \cdots N_s}^{(1, \, 2s \, +)}(\lambda_{a_1}-\eta) 
\cdots 
T_{a_{\ell}, \, 1 2 \cdots N_s}^{(1, \,  2s \, +)}
(\lambda_{a_1}-(\ell-1)\eta) \,  
P^{(\ell)}_{a_1 a_2 \cdots a_{\ell}} \, .
\ee

For instance, the (0, 1) element of 
the massive monodromy matrix $T^{(2, \, 2s \, +)}(\lambda)$ is given by  
\be 
\langle 2,0 || T^{(2, \, 2s \, +)}_{a_1 a_2, \, 1 2 \cdots N_s}(\lambda) 
|| 2, 1 \rangle 
= A_{a_1}^{(2s \, +)}(\lambda)B_{a_2}^{(2s \, +)}(\lambda-\eta) 
+ q^{-1} B_{a_1}^{(2s \, +)}(\lambda)A_{a_2}^{(2s \, +)}(\lambda-\eta) \, . 
\ee

%
%
\subsection{Series of commuting higher-spin transfer matrices}

Suppose that  $|\ell, m \rangle$ 
for $m=0, 1, \ldots, \ell$, are the orthonormal 
basis vectors of $V^{(\ell)}$, 
and their dual vectors are given by $\langle \ell, m |$ 
for $m=0, 1, \ldots, \ell$. 
We define the trace of operator $A$ over the space $V^{(\ell)}$ by 
\be 
{\rm tr}_{V^{(\ell)}} A = \sum_{m=0}^{\ell} 
\langle \ell, m | A | \ell, m \rangle \, . 
\label{eq:Tsum}
\ee
The trace of $A$ over $V^{(\ell)}$ is equivalent to the trace  
of $A$ over the $\ell$th tensor product of $V^{(1)}$,  
$(V^{(1)})^{\otimes \ell}$, multiplied by a projector $P^{(\ell)}$ 
(or $\widetilde{P}^{(\ell)}$) as follows. 
\bea 
{\rm tr}_{V^{(\ell)}} A & = & 
{\rm tr}_{(V^{(1)})^{\otimes \ell}} \left(P^{(\ell)} A \right) \non \\ 
& = & \sum_{a_1, \ldots,  a_{\ell}=0, 1}  
\left( P^{(\ell)} A 
\right)^{a_1 \cdots a_{\ell}}_{a_1 \cdots a_{\ell}} \, . 
\eea
It follows from (\ref{eq:Psum}) that the trace with respect to  
$V^{(\ell)}$ is given by (\ref{eq:Tsum}).

We define the massive transfer matrix of type 
$(\ell, (2s)^{\otimes N_s})$ by 
\bea 
& & t^{(\ell, \, 2s \, +)}_{1 2 \cdots N_s}(\lambda) 
= {\rm  tr}_{V^{(\ell)}}
\left(T^{(\ell, \, 2s \, +)}_{0, \, 1 2 \cdots N_s}
(\lambda) \right) \non \\ 
&  & \, =  
\sum_{n=0}^{\ell}  {}_a \langle \ell,  n || 
 T^{(1, \, 2s \, +)}_{a_1, \, 1 2 \cdots N_s}(\lambda) 
T^{(1, \, 2s \, +)}_{a_2, \, 1 2 \cdots N_s}(\lambda-\eta)  
\cdots T^{(1, \, 2s \, +)}_{a_{\ell}, \, 1 2 \cdots N_s}
(\lambda-(\ell-1) \eta) \, 
|| \ell, n \rangle_a  \, ,  
\eea
and the massless transfer matrix of type 
$(\ell, (2s)^{\otimes N_s})$ by 
\bea 
& & \widetilde{t}^{(\ell, \, 2s \, +)}_{1 2 \cdots N_s}(\lambda) 
= {\rm  tr}_{V^{(\ell)}}
\left( \widetilde{T}^{(\ell, \, 2s \, +)}_{0, \, 1 2 \cdots N_s}
(\lambda) \right) \non \\ 
& = &  
\sum_{n=0}^{\ell} {}_a \langle \ell,  n || \, 
 \widetilde{T}^{(1, \, 2s \, +)}_{a_1, \, 1 2 \cdots N_s}(\lambda) 
\widetilde{T}^{(1, \, 2s \, +)}_{a_2, \, 1 2 \cdots N_s}(\lambda-\eta)  
\cdots \widetilde{T}^{(1, \, 2s \, +)}_{a_{\ell}, \, 1 2 \cdots N_s}
(\lambda-(\ell-1) \eta) \, 
\widetilde{ || \ell, n \rangle}_a  \, .   
\eea

It follows from the Yang-Baxter equations that 
the higher-spin transfer matrices commute in the tensor product space 
$V_1^{(2s)} \otimes \cdots \otimes V_{N_s}^{(2s)}$, 
which is derived by applying 
projection operator $P^{(2s)}_{1 2 \cdots L}$ 
to $V^{(1)}_1 \otimes \cdots \otimes V_L^{(1)}$. 
For instance, for the massless 
transfer matrices, making use of (\ref{eq:PP'=P}) and (\ref{eq:P'P=P'})  
we show 
\be  
P^{(2s)}_{1 2 \cdots L} 
{[} \widetilde{t}^{(\ell, \, 2s \, +)}_{1 2 \cdots N_s}(\lambda), \, \, 
\widetilde{t}^{(m, \, 2s \, +)}_{1 2 \cdots N_s}(\mu) {]} = 0 \, , 
\quad \mbox{for} \, \, 
\ell, m \in {\bf Z}_{\ge 0} . 
\ee
Therefore, for the massless transfer matrices, 
the eigenvectors 
of $\widetilde{t}^{(1, \, 2s \, +)}_{1 2 \cdots N_s}(\lambda)$ 
constructed by applying  $\widetilde{B}^{(2s \, +)}(\lambda)$ 
 to the vacuum $| 0 \rangle$ 
also diagonalize the higher-spin transfer matrices, in particular, 
 $\widetilde{t}^{(2s, \, 2s \, +)}_{1 2 \cdots N_s}(\lambda)$. 
Thus, we construct the ground state 
of the higher-spin Hamiltonian in terms of operators 
$\widetilde{B}^{(2s \, +)}(\lambda)$, 
which are the (0, 1) element of the monodromy matrix 
$\widetilde{T}^{(1, \, 2s \, +)}$.

%
%
\subsection{The integrable higher-spin Hamiltonians}

We now discuss the integrable massless spin-$s$ XXZ Hamiltonian. 
For $(2s+1)$-dimensional vector spaces $V_1^{(2s)}$ and $V_2^{(2s)}$,    
we can show that the massive spin-$s$ fusion $R$-matrix 
$R^{(2s, \, 2s \, +)}_{1 \, 2}(u)$ at $u=0$ becomes 
 the permutation operator $\Pi_{1, \, 2}$ 
for $V_1^{(2s)} \otimes V_2^{(2s)}$. Furthermore, 
operator ${\check R}^{(2s, \, 2s \, +)}_{1 , \, 2}(u) 
= \Pi_{1, \, 2} {R}_{1 \, 2}^{(2s, \, 2s \, +)}(u)$  
has the following spectral decomposition: 
\be  
{\check R}^{(2s, \, 2s \, +)}_{1 , \, 2}(u) 
= \sum_{j=0}^{2s} \rho_{4s-2j}(u) \,  
\left( {P}_{4s-2j}^{2s, \, 2s} \right)_{1 , \, 2}  
 \, , 
\label{eq:SPD-spin-s-XXZ}
\ee
where operator 
$({P}_{4s-2j}^{2s, \, 2s})_{1 , \, 2}$  
projects $V_1^{(2s)} \otimes V_2^{(2s)}$ 
to spin-$(2s-j)$ representation for $j=0, 1, \ldots, 2s$.  
Functions $\rho_{4s-2j}(u)$ are given by \cite{Jimbo-QG}   
\be 
\rho_{4s-2j}(u) = 
\prod_{k=2s-j+1}^{2s} {\frac {\sinh(k \eta  -u)} {\sinh(k \eta  +u)}} \, . 
\ee
The massless spin-$s$ $R$-matrix is thus given by 
\be  
\widetilde{{\check R}}^{(2s, \, 2s \, +)}_{i , \, i+1}(u)   
= \sum_{j=0}^{2s} \rho_{4s-2j}(u) \, 
\widetilde{P}^{(2s)}_{2s(i-1)+1} \widetilde{P}^{(2s)}_{2 s i +1} 
\, \cdot \,  
\left( {P}_{4s-2j}^{2s, \, 2s} \right)_{i , \, i+1} 
\, .  
\ee

It is easy to show that the massless spin-$s$ $R$-matrix 
$\widetilde{R}_{1 \, 2}^{(2s, \, 2s \, +)}(u)$
becomes the permutation operator at $u=0$: 
$\widetilde{R}_{1 \, 2}^{(2s, \, 2s \, +)}(0) = \Pi_{1, \, 2}$.  
Therefore, 
putting inhomogeneous parameters $\xi_p=0$ for $p=1, 2, \ldots, N_s$, 
we show that
that the transfer matrix 
$\widetilde{t}^{(2s, \, 2s \, +)}_{1 2 \cdots N_s}(\lambda)$ 
becomes the shift operator at $\lambda=0$. 
We thus derive the massless spin-$s$ XXZ Hamiltonian 
by the logarithmic derivative of the massless spin-$s$ transfer matrix, 
similarly as for the massive case.   
\bea 
{\cal H}^{(2s)}_{\rm XXZ} & = & \left. {\frac d {d \lambda}} 
\log \widetilde{t}^{(2s,  \, 2s \, +)}_{1 2 \cdots N_s}(\lambda)
\right|_{\lambda=0 , \, \xi_j=0} 
= 
\sum_{i=1}^{N_s} 
\left. \frac d {du} 
\widetilde{\check R}_{i, i+1}^{(2s, 2s)}(u) \right|_{u=0} \non \\ 
& = & \sum_{i=1}^{N_s} \sum_{j=0}^{2s} 
{\frac {d \rho_{4s-2j}} {du}} (0)  
 \, 
\widetilde{P}^{(2s)}_{2s(i-1)+1} \widetilde{P}^{(2s)}_{2 s i +1} 
\, \cdot \,  
\left( {P}_{4s-2j}^{2s, \, 2s} \right)_{i , \, i+1} \, . 
\label{eq:deriv-XXZ-Hamiltonian}
\eea

%
%
%
\setcounter{section}{2} 
 \setcounter{equation}{0} 
 \renewcommand{\theequation}{3.\arabic{equation}}
%
\section{Higher-spin expectation values}


%
%
\subsection{Algebraic Bethe ansatz}

In terms of  the vacuum vector $| 0 \rangle$ 
 where all spins are up,   
we define functions $a(\lambda)$ and $d(\lambda)$  by 
\bea 
A^{(1 \, p)}(\lambda; \{ w_j \}_L) |0 \rangle & = & 
a(\lambda; \{ w_j \}_L) | 0 \rangle \, , \non \\  
D^{(1 \, p)}(\lambda; \{ w_j \}_L) |0 \rangle & = & 
d(\lambda; \{ w_j \}_L) | 0 \rangle  \,  . 
\eea
We have $a(\lambda; \{ w_j \}_L) = 1$ and 
\be 
d(\lambda; \{ w_j \}_L)  =  \prod^{L}_{j=1} b(\lambda, w_j) \, . 
\ee
Here $b(\lambda, \mu)= b(\lambda-\mu)$.  
For the homogeneous grading ($w=+$) 
 and the principal  grading ($w=p$), 
it is easy to show the following relations:   
\bea 
A^{(2s \, w)}(\lambda) |0 \rangle & = & 
\widetilde{A}^{(2s \, w)}(\lambda) |0 \rangle = 
a^{(2s)}(\lambda; \{ \xi_k \}) | 0 \rangle \, , \non \\  
D^{(2s \, w)}(\lambda) |0 \rangle & = & 
\widetilde{D}^{(2s \, w)}(\lambda) |0 \rangle = 
d^{(2s)}(\lambda; \{ \xi_k \}) | 0 \rangle  \,  , 
\eea
where $a^{(2s)}(\lambda; \{ \xi_k \})$ and 
$d^{(2s)}(\lambda; \{ \xi_k \})$ are given by 
\bea 
a^{(2s)}(\lambda; \{ \xi_k \}) & = & 
a(\lambda; \{ w_j^{(2s)} \})=1 \, , \non \\  
d^{(2s)}(\lambda; \{ \xi_k \}) & = & 
d(\lambda; \{ w_j^{(2s)} \}) 
 =  \prod^{N_s}_{p=1} b_{2s}(\lambda, \xi_p) \, . 
\eea
Here we have  defined $b_t(\lambda, \mu)$ by  
$b_t(\lambda, \mu) = {\sinh(\lambda-\mu)}/{\sinh(\lambda-\mu+ t \eta)}$ .  
Here we recall $b(u)= b_{1}(u)= \sinh u/\sinh(u+\eta)$.  

In the massless regime, we define the Bethe vectors 
$| \widetilde{\{\lambda_{\alpha} \}}_{M}^{(2s \, w)} \,  \rangle$ 
for $w=+$ and $p$, and their dual vectors 
$ \langle \widetilde{\{\lambda_{\alpha} \}}_{M}^{(2s \, w)} |$ 
for $w=+$ and $p$,  as follows.   
\bea 
| \widetilde{ \{\lambda_{\alpha} \}}_{M}^{(2s \, w)} \,  \rangle 
& = & \prod_{\alpha=1}^{M} 
\widetilde{B}^{(2s \, w)}(\lambda_{\alpha}) | 0 \rangle \, ,  
\label{eq:eigen} \\ 
 \langle \widetilde{ \{\lambda_{\alpha} \}}_{M}^{(2s \, w)} \, |  
& = & \langle  0 | \, 
\prod_{\alpha=1}^{M} 
\widetilde{C}^{(2s \, w)}(\lambda_{\alpha})  \, . 
\label{eq:dual-eigen}
\eea
Here we recall 
$\widetilde{B}^{(2s \, +)}(\lambda_{\alpha})= 
\widetilde{P}^{(2s)}_{1 \cdots L} 
B^{(1 \, +)}(\lambda_{\alpha}, \{ w_k \}_L) 
\widetilde{P}^{(2s)}_{1 \cdots L}$. 
The Bethe vector (\ref{eq:eigen})  
gives an eigenvector of the massless transfer matrix 
\be 
\widetilde{t}^{(1, \, 2s \, w)}(\mu; \{ \xi_p \}_{N_s})
=\widetilde{ A}^{(2s \, w)}(\mu; \{\xi_p \}_{N_s}) + 
\widetilde{D}^{(2s \, w)}(\mu; \{\xi_p \}_{N_s})
\ee
for $w=+$ and $w=p$  with the following eigenvalue: 
\be 
\Lambda^{(1, {2s} \, w)}(\mu)= 
\prod_{j=1}^{M} 
{\frac {\sinh(\lambda_j - \mu + \eta)} {\sinh(\lambda_j - \mu)}} 
+ 
\prod_{p=1}^{N_s} 
b_{2s}(\mu, \xi_p) \, \cdot \, 
\prod_{j=1}^{M} 
{\frac {\sinh(\mu - \lambda_j + \eta)} {\sinh(\mu - \lambda_j)}} \, , 
\ee
if rapidities $\{\lambda_j \}_M$ satisfy  the Bethe ansatz equations  
\be 
\prod_{p=1}^{N_s} b_{2s}^{-1}(\lambda_{j}, \xi_p) = 
\prod_{k \ne j} {\frac {b(\lambda_{k}, \lambda_{j})} 
                     {b(\lambda_{j}, \lambda_{k})}}  \quad 
(j= 1, \ldots, M) \, .  \label{eq:BAE} 
\ee

Let us denote by 
$| \{\lambda_{\alpha}(\epsilon) \}_M^{(2s \, w; \, \epsilon)} \rangle$ 
the Bethe vector of $M$ Bethe roots $\{ \lambda_j(\epsilon)  \}_M$ 
for $w=+, p$:  
\be 
| \{\lambda_{\alpha}(\epsilon) \}_M^{(2s \, w; \, \epsilon)} 
\rangle =  B^{(2s \, w; \epsilon)}(\lambda_1(\epsilon) ) 
\cdots B^{(2s \, w; \epsilon)}(\lambda_M(\epsilon) ) | 0 \rangle \, ,  
\ee
where rapidities $\{ \lambda_j(\epsilon) \}_M$ satisfy the Bethe ansatz 
equations with inhomogeneous parameters $w_j^{(2s; \epsilon)}$ as follows.  
\be 
\frac {a(\lambda_j(\epsilon); \{ w_{k}^{(2s; \, \epsilon)} \}_L )} 
{d(\lambda_j(\epsilon) ;  \{ w_{k}^{(2s; \, \epsilon)} \}_L) } = 
\prod_{k=1; k \ne j}^{M} 
{\frac {b(\lambda_k(\epsilon), \lambda_j(\epsilon))}  
{b(\lambda_j(\epsilon), \lambda_k(\epsilon))}} \, . 
\ee
It gives an eigenvector of the transfer matrix 
\be 
t^{(1, 1 \, w)}(\mu; \{ w_j^{(2s; \, \epsilon)}  \}_{L})
=A^{({2s \, w; \, \epsilon})}(\mu; \{w_j^{(2s; \, \epsilon)} \}_{L}) + 
D^{({2s \, w; \, \epsilon})}(\mu; \{w_j^{(2s; \, \epsilon)} \}_{L})
\ee
with the following eigenvalue: 
\be 
\Lambda^{(1, 1 \, w)}(\mu; \{w_j^{(2s; \, \epsilon)} \}_{L})= 
\prod_{j=1}^{M} 
{\frac {\sinh(\lambda_j(\epsilon) - \mu + \eta)} 
{\sinh(\lambda_j(\epsilon) - \mu)}} 
+ \prod_{j=1}^{L} 
b(\mu, w_j^{(2s; \, \epsilon)} ) \, \cdot \, 
\prod_{j=1}^{M} 
{\frac {\sinh(\mu - \lambda_j(\epsilon) + \eta)} 
{\sinh(\mu - \lambda_j(\epsilon))}} \, . 
\ee

Let us assume that in the limit of $\epsilon$ going to 0,  
the set of Bethe roots $\{ \lambda_j(\epsilon)  \}_M$ is given by 
$\{ \lambda_j  \}_M$. Then, we have  
\be 
 P_{12 \cdots L}^{(2s)} \,  
| \widetilde{ \{ \lambda_j \}}_M^{(2s \, +)} \rangle = 
\lim_{\epsilon \rightarrow 0} P_{12 \cdots L}^{(2s)} \, 
| \{ \lambda_j(\epsilon)  \}_M^{(2s \, +; \epsilon)} \rangle  \, . 
\label{eq:limit-Bethe-state}
\ee

%
%
\subsection{
Hermitian elementary matrices ${\widetilde{ E}}_i^{m \, , \, n \, (2s \, +)}$ 
in the massless regime}

We define massless elementary matrices 
$\widetilde{E}^{m, \, n \, (2s+)}$ for $m, n=0, 1, \ldots, 2s$,  
in the spin-$s$ representation of $U_q(sl_2)$ as follows. 
\be 
\widetilde{E}^{m, \, n \, (2s \, +)} = 
\widetilde{|| \ell, m \rangle} \langle \ell, n|| \, . 
\ee
In the tensor product space, 
$(V^{(2s)})^{\otimes N_s}$, 
we define $\widetilde{E}^{m, \, n \, (2s \, +)}_i$ 
for $i=1, 2, \ldots, N_s$ by  
\be 
\widetilde{E}^{m, \, n \, (2s \, +)}_i = (I^{(2s)})^{\otimes (i-1)} \otimes 
\widetilde{E}^{m, \, n \, (2s \, +)}  \otimes (I^{(2s)})^{\otimes (N_s-i)} \, . 
\ee
Elementary matrices $\widetilde{E}^{n, \, n \, (2s \, +)}$ 
for $n=0, 1, \ldots, 2s$,  
are Hermitian in the massless regime. 
In fact, when $|q|=1$, for $m, n= 0, 1, \ldots, 2s$,  we have 
\be 
\left( \widetilde{E}^{m, \, n \, (2s \, +)} \right)^{\dagger} 
= 
\left[ 
\begin{array}{c}  
2s \\ 
m 
\end{array} 
\right]_q^2 \, 
\left[ 
\begin{array}{c}  
2s \\ 
n 
\end{array} 
\right]_q^{-2} \, 
\left( 
\begin{array}{c}  
2s \\ 
m 
\end{array} 
\right)^{-1} \, \, 
\left( 
\begin{array}{c}  
2s \\ 
n 
\end{array} 
\right) \, 
\widetilde{E}^{n, \, m \, (2s \, +)} \, . 
\ee

We can express any given spin-$s$ local operator 
of the massless case 
in terms of the spin-1/2 global operators   
by a method similar to the massive case \cite{DM1}. 
For $m = n$,  we have 
\bea  
\widetilde{ E}_{i}^{n, \, n \, (2s+)} 
& = & \left(
\begin{array}{c} 
2s \\
n 
\end{array} 
 \right) \, 
\widetilde{P}^{(2s)}_{1 \cdots L} 
\,  
\prod_{\alpha=1}^{(i-1)2s} (A^{(1+)}+D^{(1+)})(w_{\alpha}) 
\prod_{k=1}^{n} D^{(1+)}(w_{(i-1)2s+k} ) 
\non \\ 
& & \quad \times \, \prod_{k=n+1}^{2s} A^{(1+)}(w_{(i-1)2s+k})  
\prod_{\alpha=i 2s +1}^{2s N_s} 
(A^{(1+)}+D^{(1+)})(w_{\alpha}) \, \,   
\widetilde{ P}^{(2s)}_{1 \cdots L}
\, .  \label{eq:Em=n}
\eea
Formulas expressing $\widetilde{E}^{m, \, n \, (2s+)}$ 
for $m > n$ or $m < n$ are given in Appendix E. 

When we evaluate expectation values, 
we want to remove the projection operators 
introduced in order to express the spin-$s$ local operator 
in terms of spin-1/2 global operators such as in (\ref{eq:Em=n}). 
Then, we shall make use of the following lemma.  

\begin{lemma} 
Projection operators $P^{(2s)}_{12 \cdots L}$ and 
$\widetilde{P}^{(2s)}_{12 \cdots L}$ 
commute with the matrix elements of the monodromy matrix 
$T^{(1,1 \, +)}_{0, 12 \cdots L}(\lambda; \{ w_j^{(2s; \epsilon)} \}_L)$ 
such as $A^{(2s \, +; \epsilon)}(\lambda)$ 
in the limit of $\epsilon$ going to 0. 
\bea  
 P_{12 \cdots L}^{(2s)}  
T^{(1 , \, 1 \, +)}_{0, 12 \cdots L}(\lambda; 
\{ w_j^{(2s; \, \epsilon)} \}_L) \, P_{12 \cdots L}^{(2s)}  
& = & 
P_{12 \cdots L}^{(2s)}
 \, T^{(1, \, 1 \, +)}_{0, 12 \cdots L}
(\lambda; \{ w_j^{(2s; \, \epsilon)} \}_L) 
+ O(\epsilon) \, , \label{eq:commute} \\ 
 P_{12 \cdots L}^{(2s)}  
T^{(1, \, 1 \, +)}_{0, 12 \cdots L}(\lambda; 
\{ w_j^{(2s; \, \epsilon)} \}_L) \,
 \widetilde{P}_{12 \cdots L}^{(2s)} 
& = & 
P_{12 \cdots L}^{(2s)} \, T^{(1,1)}_{0, 12 \cdots L}
(\lambda; \{ w_j^{(2s; \, \epsilon)} \}_L) 
+ O(\epsilon) \, . \label{eq:commute2} 
\eea
For instance we have 
$P_{12 \cdots L}^{(2s)} B^{(2s \, +; \, \epsilon)}(\lambda) 
P_{12 \cdots L}^{(2s)}
= P_{12 \cdots L}^{(2s)}  B^{(2s \, +; \, \epsilon)}(\lambda) + O(\epsilon). $ 
\end{lemma}
\begin{proof} 
Taking derivatives with respect to inhomogeneous parameters $w_j$, 
we can show 
\be 
T^{(1,1 \, +)}_{0, 12 \cdots L}(\lambda; \{ w_j^{(2s \, +; \, \epsilon)} \}_L)
= 
T^{(1,1 \, +)}_{0, 12 \cdots L}(\lambda; \{ w_j^{(2s)} \}_L) 
+ O(\epsilon) \, ,  
\ee
where $T^{(1, \, 1 \, +)}_{0, 12 \cdots L}(\lambda; \{ w_j^{(2s)} \}_L)$
commutes with the projection operator $P_{12 \cdots L}^{(2s)}$ 
as follows \cite{DM1}.  
\be 
P_{12 \cdots L}^{(2s)} 
T^{(1, \, 1 \, +)}_{0, 12 \cdots L}(\lambda; \{ w_j^{(2s)} \}_{L}) = 
P_{12 \cdots L}^{(2s)} 
T^{(1, \, 1 \, +)}_{0, 12 \cdots L}(\lambda; \{ w_j^{(2s)} \}_{L}) 
P_{12 \cdots L}^{(2s)} \, . 
\ee
We show (\ref{eq:commute2}) making use of (\ref{eq:PP'=P}). 
\end{proof}

%
%

\subsection{Expectation value of a local operator 
through the limit: $\epsilon \rightarrow 0$}

In the massless regime, 
we define the expectation value of product of operators 
$\prod_{k=1}^{m} \widetilde{E}_k^{i_k, \, j_k \, (2s \, +)}$ 
with respect to an eigenstate 
$| \widetilde{ \{\lambda_{\alpha} \}}_M^{(2s \, +)} \rangle$ by 
\be 
\langle 
\prod_{k=1}^{m}
\widetilde{E}_k^{i_k, \, j_k \, (2s \, +)} \rangle \left( 
\{ \lambda_{\alpha} \}_M^{(2s \, +)} \right)  
=
{\frac {\langle \widetilde{ \{\lambda_{\alpha} \}}_M^{(2s \, +)} | \, 
\prod_{k=1}^{m} \widetilde{E}_k^{i_k, \, j_k \, (2s \, +)} 
 | \widetilde{ \{\lambda_{\alpha} \}}_M^{(2s \, +)} \rangle} 
 {\langle \widetilde{ \{\lambda_{\alpha} \}}_M^{(2s \, +)} | 
 \widetilde{ \{\lambda_{\alpha} \}}_M^{(2s \, +)} \rangle}} \, .  
%
\ee

We evaluate the expectation value of a given spin-$s$ local operator 
for a Bethe-ansatz eigenstate $| \{\lambda_{\alpha} \}_M^{(2s)} \rangle$, 
as follows. 
We first assume that the Bethe roots 
$\{ \lambda_{\alpha}(\epsilon)\}_M$  
are continuous with respect to small parameter $\epsilon$.  
We express the spin-$s$ local operator 
in terms of spin-1/2 global operators 
such as formula (\ref{eq:Em=n}) with generic 
inhomogeneous parameters $w_j^{(2s; \epsilon)}$.  
Applying (\ref{eq:commute}) and (\ref{eq:commute})
we remove the projection operators 
out of the product of global operators.  
We next calculate the scalar product for the Bethe state 
$|\{ \lambda_k(\epsilon) \}_M^{(2s; \, \epsilon)} \rangle$ 
which has the same inhomogeneous parameters $w_j^{(2s; \epsilon)}$, 
making use of the formulas of the spin-1/2 case.       
Then we take the limit of sending $\epsilon$ to 0, and obtain 
the expectation value of the spin-$s$ local operator.

For an illustration,  
let us consider the expectation value of 
$\widetilde{E}_1^{n , \, n \, (2s \, +)}$.  
First, applying projection operator $P^{(2s)}_{1 2 \cdots L}$ to 
$| \widetilde{ \{ \lambda_{\alpha} \}}^{(2s \, +)}_{M} \rangle 
= \prod_{\alpha=1}^{M} 
\widetilde{B}^{(2s \, +)}(\lambda_{\alpha}) |0 \rangle$ 
we show  
\bea 
P^{(2s)}_{1 \cdots L} |  \widetilde{ 
\{ \lambda_{\alpha} \}}_M^{(2s \, +)} \rangle 
& = & P^{(2s)}_{1 \cdots L} \, \prod_{\alpha=1}^{M} 
B^{(2s \, +; \epsilon)}(\lambda_{\alpha}(\epsilon)) \,  |0 \rangle 
+ O(\epsilon) 
\non \\ 
& = & e^{- \sum_{\alpha=1}^{M} \lambda_{\alpha}(\epsilon)} \, 
 P^{(2s)}_{1 \cdots L} \, \chi_{1 2 \cdots L} \, 
\prod_{\alpha=1}^{M} 
B^{(2s; \, \epsilon)}(\lambda_{\alpha}(\epsilon)) \,  |0 \rangle 
+ O(\epsilon)
\, . 
\eea
Second, making use of the relation 
$\langle 0 | = \langle  0| P_{12 \cdots L}^{(2s)}$, 
we show 
\bea 
\langle  \widetilde{ \{ \lambda_{\alpha} \}}_M^{(2s+)} | 
& = & \langle 0 |  \, 
\prod_{\alpha=1}^{M} C^{(2s \, +; \epsilon)} 
(\lambda_{\alpha}(\epsilon)) \, P^{(2s)}_{1 \cdots L}
 + O(\epsilon) \non \\ 
& = & \langle 0 |  \, 
\prod_{\alpha=1}^{M} C^{(2s; \, \epsilon)} 
(\lambda_{\alpha}(\epsilon)) \, \chi_{1 2 \cdots L}^{-1}  
\, P^{(2s)}_{1 \cdots L} \, e^{\sum_{\alpha=1}^{M} \lambda_{\alpha}(\epsilon)} 
 + O(\epsilon)
\, .  
\eea
Making use of (\ref{eq:Em=n}) we have 
\bea 
& & 
\langle \widetilde{ \{ \lambda_{\alpha} \}}_M^{(2s \, +)} |  \,  
\widetilde{E}^{n \, n \, (2s \, +)}_{1}  \,   
| \widetilde{ \{ \lambda_{\alpha} \}}_M^{(2s \, +)} \rangle 
\non \\ 
& = & 
\left(
\begin{array}{c} 
2s \\
n 
\end{array} 
 \right) \, 
\langle 0 |  
 \prod_{\alpha=1}^{M} C^{(2s+; \, \epsilon)}
(\lambda_{\alpha}(\epsilon)) P^{(2s)}_{1 \cdots L}  
\, \underline{\widetilde{P}_{12 \cdots L}^{(2s)}} \,    
\prod_{k=1}^{n} D^{(2s+; \, \epsilon)}(w_{k}^{(2s; \, \epsilon)} )  
 \prod_{k=n+1}^{2s} {A}^{(2s+; \, \epsilon)}
(w_{k}^{(2s; \, \epsilon)})  
\non \\
& & \quad \times \,
\prod_{\alpha= 2s +1}^{2s N_s} 
(A^{(2s \, +; \, \epsilon)}+D^{(2s \, +; \, \epsilon)})
(w_{\alpha}^{(2s; \, \epsilon)}) \, \,   
\underline{ \widetilde{P}^{(2s)}_{1 \cdots L}} \, \cdot \, 
\prod_{\alpha=1}^{M} 
\widetilde{B}^{(2s \, +)}(\lambda_{\alpha}) 
 |0 \rangle \, + O(\epsilon) \, . 
 \label{eq:remove-projectors}
\eea 
Here we have $\prod_{j=1}^{2sN_s} 
(A^{(2s+; \, \epsilon)}+D^{(2s+ ; \, \epsilon)})(w_j^{(2s; \, \epsilon)})
=I^{\otimes L}$ 
for generic $\epsilon$. 
We apply projection operators $P^{(2s)}$ to $\widetilde{P}^{(2s)}$ 
from the left, 
which are underlined in (\ref{eq:remove-projectors}), 
and make use of (\ref{eq:PP'=P}). 
We then move 
the projection operators $P^{(2s)}$ in the leftward direction, 
 making use of (\ref{eq:commute}).  
Thus,  the right-hand side of (\ref{eq:remove-projectors}) 
is now given by the following: 
\bea
& = & 
\left(
\begin{array}{c} 
2s \\
n 
\end{array} 
 \right) \, 
\langle 0 | 
 \prod_{\alpha=1}^{M}  
C^{(2s +; \, \epsilon)}(\lambda_{\alpha}(\epsilon))  \, \,  
\prod_{k=1}^{n} D^{(2s +; \, \epsilon)}(w_{k}^{(2s; \, \epsilon)} ) 
\,  \prod_{k=n+1}^{2s} A^{(2s + ; \, \epsilon)}
(w_{k}^{(2s; \, \epsilon)})  
\non \\ 
& & \times 
\prod_{j= 2s +1}^{2s N_s} 
(A^{(2s+; \, \epsilon)}+D^{(2s+; \, \epsilon)})
(w_{j}^{(2s; \, \epsilon)}) \,  \,   
\prod_{\beta=1}^{M} 
B^{(2s+; \, \epsilon)}(\lambda_{\beta}(\epsilon)) |0 \rangle 
+ O(\epsilon) . \label{eq:product} 
\eea
After applying the gauge transformation 
$\chi_{1 \cdots L}^{-1}$ inverse to (\ref{eq:gauge-transform}) 
\cite{DM1}, we obtain 
\bea  
& & \langle \widetilde{ \{ \lambda_{\alpha} \}}_M^{(2s \, +)} | 
 \widetilde{ E}^{n \, n \, (2s \, +)}_{1}   
| \widetilde{ \{ \lambda_{\alpha} \}}_M^{(2s \, +)} \rangle  
\non \\ 
& = & 
\left(
\begin{array}{c} 
2s \\
n 
\end{array} 
 \right) \, 
\lim_{\epsilon \rightarrow 0} \, 
\langle 0 | 
 \prod_{\alpha=1}^{M}  
C^{(2s; \, \epsilon)}(\lambda_{\alpha}(\epsilon))  \, \,  
\prod_{k=1}^{n} D^{(2s; \, \epsilon)}(w_{k}^{(2s; \, \epsilon)} ) 
\,  \prod_{k=n+1}^{2s} A^{(2s; \, \epsilon)}
(w_{k}^{(2s; \, \epsilon)})  
\non \\ 
& & \quad \times \,
\prod_{j= 2s +1}^{2s N_s} 
(A^{(2s; \, \epsilon)}+D^{(2s; \, \epsilon)})
(w_{j}^{(2s; \, \epsilon)}) \, \,   
\prod_{\beta=1}^{M} 
B^{(2s; \, \epsilon)}(\lambda_{\beta}(\epsilon)) |0 \rangle \, .  
\label{eq:product2}
\eea
Here $A^{(2s; \, \epsilon)}$ and $D^{(2s; \, \epsilon)}$ denote 
matrix elements $A^{(2s \, p; \, \epsilon)}$ 
and $D^{(2s \, p; \, \epsilon)}$
of the monodromy matrix 
with principal grading, respectively.  
In the last line of (\ref{eq:product}), 
 we have evaluated the eigenvalue of transfer matrix 
$A^{(2s; \, \epsilon)}(w_{j}^{(2s; \, \epsilon)})+
D^{(2s; \, \epsilon)}(w_{j}^{(2s; \, \epsilon)})$ 
on the eigenstate 
$| \{ \lambda_{\beta}(\epsilon) \}_M^{(2s; \, \epsilon)} \rangle$ 
as follows.    
\bea 
& & \prod_{j= 2s +1}^{2s N_s}  
\left( A^{(2s; \, \epsilon)}(w_{j}^{(2s; \, \epsilon)})+
D^{(2s; \, \epsilon)}(w_{j}^{(2s; \, \epsilon)}) \right)
| \{ \lambda_{\beta}(\epsilon) \}_M^{(2s; \, \epsilon)} \rangle 
\non \\ 
& & = \left( 
\prod_{j=2s+1}^{2s N_s} \prod_{\alpha=1}^{M} 
b^{-1}(\lambda_{\alpha}(\epsilon) -  w_{j}^{(2s; \, \epsilon)}) 
\right) \, 
| \{ \lambda_{\beta}(\epsilon) \}_M^{(2s; \, \epsilon)} \rangle \, . 
\eea

Before sending $\epsilon$ to 0, 
we expand the products of $C$ operators multiplied by operators 
$A$ and $D$ by the commutation relations between $C$ and $A$ as well as 
$C$ and $D$, respectively. We then  
evaluate the scalar product of $B$ and $C$ operators 
with inhomogeneous parameters $w^{(2s; \epsilon)}_j$. 
Finally,  we  derive the expectation value in 
the limit of sending $\epsilon$ to 0.

Sending $\epsilon$ to 0, we calculate 
the expectation value of $A^{(2s)}(\lambda) +D^{(2s)}(\lambda)$ 
at $\lambda= w_2^{(2s)}$. 
For instance, we calculate 
$A^{(2s; \, \epsilon)}(w_2^{(2s; \, \epsilon)})+
D^{(2s; \, \epsilon)}(w_2^{(2s; \, \epsilon)})$  
on the vacuum $|0 \rangle$ as follows.  
\bea 
& & \lim_{\epsilon \rightarrow 0} 
\langle 0 | 
\left(A^{(2s; \epsilon)}(w_2^{(2s; \epsilon)}) 
+ D^{(2s; \epsilon)}(w_2^{(2s; \epsilon)}) \right) | 0 \rangle
\non \\ 
& = & \lim_{\epsilon \rightarrow 0} 
\langle 0 | \left( A(w_2^{(2s; \epsilon)}; 
\{ w_j^{(2s; \epsilon)} \}_L )
+D^{(2s)}(w_2^{(2s; \epsilon)}; 
\{ w_j^{(2s; \epsilon)} \}_L) \right) | 0 \rangle 
\non \\ 
& = & \lim_{\epsilon \rightarrow 0} \, \left( 1 + \prod_{j=1}^{\ell N_s} 
b(w_2^{(2s; \epsilon)} - w_j^{(2s; \epsilon)}) \right) 
\langle 0 | 0 \rangle 
\non \\ 
& = & (1 + 0 ) \, \langle 0 | 0 \rangle \, . 
\label{eq:limit}
\eea
If we put $\lambda=w_2^{(2s)}$ after sending  $\epsilon$ to 0,  
the result is different from (\ref{eq:limit}) as follows. 
\be 
\lim_{\lambda \rightarrow w^{(2s)}_2} 
 \langle 0 | \left( A^{(2s)}(\lambda; \{ w_j^{(2s)} \}_L)+
D^{(2s)}(\lambda; \{ w_j^{(2s)} \}_L) \right)) | 0 \rangle 
= \left( 1 + \prod_{p=1}^{N_s} b_{\ell}(w^{(2s)}_2 - \xi_p) 
\right) \langle 0 | 0 \rangle \, . 
\ee

%
%
%
%
\setcounter{section}{3}
 \setcounter{equation}{0} 
 \renewcommand{\theequation}{4.\arabic{equation}}
%
\section{Derivation of matrix $S$}

\subsection{The ground-state solution of $2s$-strings}

We shall introduce $\ell$-strings for an integer $\ell$.  
Let us shift rapidities $\lambda_j$ by $s \eta$ such as 
$\tilde{\lambda}_j=\lambda_j + s \eta$. Then,   
the Bethe ansatz equations (\ref{eq:BAE}) are given by   
\be 
\prod_{p=1}^{N_s} {\frac 
{\sinh(\tilde{\lambda}_{j} - {\xi}_p + s \eta)} 
{\sinh(\tilde{\lambda}_{j} - {\xi}_p - s \eta)} }  
= \prod_{\beta=1; \beta \ne \alpha}^{n}  
{\frac {\sinh(\tilde{\lambda}_{j} - \tilde{\lambda}_{\beta} + \eta)} 
       {\sinh(\tilde{\lambda}_{j} - \tilde{\lambda}_{\beta} - \eta)}} 
\, , \quad \mbox{for} \, \, j=1, 2, \ldots, n \,. 
\ee

We define an $\ell$-string by the following set of rapidities. 
\be 
{\tilde \lambda}_{a}^{(\alpha)} = \mu_a 
+ ({\ell} +1 - 2 \alpha) {\frac {\eta} 2} 
+ \epsilon_a^{(\alpha)} 
\quad  
\mbox{for} \quad \alpha =1, 2, \ldots, \ell.  
\label{eq:2s-string} 
\ee
We call $\mu_a$ the center of the $\ell$-string 
and $\epsilon_a^{(\alpha)}$ string deviations. 
We assume that $\epsilon_a^{(\alpha)}$ are very small for large $N_s$: 
\be 
\lim_{N_s \rightarrow \infty} 
\epsilon_a^{(\alpha)}=0 .  
\ee
If they are zero, then we call the set of rapidities of 
(\ref{eq:2s-string})  a complete $\ell$-string. 
The string center $\mu_a$ corresponds 
to the central position among the $\ell$ complex 
numbers:  ${\tilde \lambda}_{a}^{(1)}, {\tilde \lambda}_{a}^{(2)}, \ldots, 
{\tilde \lambda}_{a}^{(\ell)}$. 
Furthermore we assume that $\mu_a$ are real. 
If inhomogeneous parameters, $\xi_p$, are small enough, 
then the Bethe ansatz equations should have 
an $\ell$-strings  as a solution.  

In terms of rapidities ${\lambda}_j$ which are not shifted,  
an $\ell$-string is expressed in the following form:   
\be 
\lambda_{a}^{(\alpha)} = \mu_a - (\alpha- 1/2) \eta 
+ \epsilon_a^{(\alpha)} \quad 
\mbox{for} \quad \alpha = 1, 2, \ldots, \ell \, . 
\ee
We denote ${\lambda}_{a}^{(\alpha)}$ 
also by ${\lambda}_{(a, \alpha)}$.

Let us now introduce the conjecture that the ground state 
of the spin-$s$ case $| \psi_g^{(2s)} \rangle$ is 
given by $N_s/2$ sets of $2s$-strings:   
\be 
\lambda_{a}^{(\alpha)} 
= \mu_a - (\alpha- 1/2) \eta + \epsilon_a^{(\alpha)} \, , \quad  
\mbox{for} \, \, a=1, 2, \ldots, N_s/2 \, \,  
\mbox{and} \, \,  \alpha = 1, 2, \ldots, 2s .  
\ee
In terms of $\lambda_{a}^{(\alpha)}$s in the massless regime, 
for $w=+$ and $p$,  we have  
\be 
 | \psi_g^{(2s \, w)} \rangle = 
\prod_{a=1}^{N_s/2} \prod_{\alpha=1}^{2s} 
\widetilde{B}^{(2s \, w)}(\lambda_a^{(\alpha)}; \{\xi_p \}) | 0 \rangle . 
\ee
Hereafter we set $M= 2s N_s/2 = s N_s$.

According to analytic and numerical studies 
\cite{deVega-Woynarovich,KB,KBP},  
we may assume the following  
properties of string deviations $\epsilon_a^{(\alpha)}$s. 
When $N_s$ is very large, the deviations are given by 
\be 
\epsilon_a^{(\alpha)} = i  \, \delta_a^{(\alpha)} \, ,  
\ee
where $i$ denotes $\sqrt{-1}$, and 
$\delta_a^{(\alpha)}$ are real. Moreover, 
$\delta_{a}^{(\alpha)} - \delta_{a}^{(\alpha+1)} > 0$ 
for $\alpha=1, 2, \ldots, 2s-1$,  and 
$|\delta_{a}^{(\alpha)}| > |\delta_{a} ^{(\alpha+1)}|$ 
for $\alpha < s$, 
 while $|\delta_{a}^{(\alpha)}| < |\delta_{a} ^{(\alpha+1)}|$ 
for $\alpha \ge s$.

In the thermodynamic limit: $N_s \rightarrow \infty$,  
the Bethe ansatz equations for the ground state of the 
higher-spin XXZ chain 
become the integral equation for the string centers, 
as shown in Appendix F \cite{M.Takahashi}. 
The density of string centers, $\rho_{\rm tot}(\mu)$, is given by 
\be 
\rho_{\rm tot}(\mu)= {\frac 1 N_s} \sum_{p=1}^{N_s} 
{\frac 1 {2 \zeta \cosh(\pi (\mu- \xi_p)/\zeta)}} 
\ee
Thus, the sum over all the Bethe roots of the ground state 
is evaluated by integrals 
in the thermodynamic limit, $N_s \rightarrow \infty$, 
as follows. 
\bea
\frac 1 {N_s} \sum_{A=1}^{M} f(\lambda_A) 
& = & \frac 1 {N_s} \sum_{\alpha=1}^{2s} 
\sum_{a=1}^{N_s/2} f(\lambda_{(a, \alpha)}) \non \\ 
& = & \sum_{\alpha=1}^{2s}
 \int_{-\infty}^{\infty}  f(\mu_a - (\alpha-1/2) \eta 
+ \epsilon_{a}^{(\alpha)}) \, 
\rho_{\rm tot}(\mu_a) \, d \mu_{a} + O(1/N_s) \, . 
\eea

For the homogeneous chain where $\xi_p=0$ for $p=1, 2, \ldots, N_s$, 
we denote the density of string centers by $\rho(\lambda)$.  
\be 
\rho(\lambda)= {\frac 1 {2 \zeta \cosh(\pi \lambda/ \zeta)}} \, . 
\ee

Let us introduce useful notation of the suffix of rapidities. 
For rapidities $\lambda_{a}^{(\alpha)}=\lambda_{(a, \alpha)}$ 
we define integers $A$ by $A= 2s(a-1) + \alpha$ for 
$a=1, 2, \ldots, N_s/2$ and for $\alpha=1, 2, \ldots, 2s$.  
We thus denote $\lambda_{(a, \alpha)}$ also 
by $\lambda_A$ for $A=1, 2, \ldots, s N_s$, and 
put $\lambda_{(a, \alpha)}$ 
in increasing order with respect to $A=2s(a-1)+\alpha$ 
such as $\lambda_{(1,1)}=\lambda_1, \lambda_{(1,2)}=\lambda_2, 
\ldots, \lambda_{(N_s/2, 2s)}=\lambda_{s N_s}$.    

In the ground state 
 rapidities $\lambda_A$ for 
$A=1, 2, \ldots, M$, are now expressed by   
\be 
\lambda_{2s(a-1)+ \alpha} = \mu_{a} - (\alpha- 1/2) \eta 
+ \epsilon_a^{(\alpha)} \quad
 \mbox{for} \, \, a=1, 2, \ldots, N_s/2 \, \, 
\mbox{and} \, \,  \alpha= 1, 2, \ldots, 2s.     
\ee
For a given real number $x$, 
let us denote by $[x]$ the greatest integer less than or equal to $x$. 
When $A=2s(a-1)+ \alpha$ with $1 \le \alpha \le 2s$, 
integer $a$ is given by $a=[(A-1)/2s]+1$, and integer $\alpha$ 
is given by $\alpha=A - 2s [(A-1)/2s]$.

%
%
\subsection{Derivation of the spin-$s$ EFP for a finite chain}

We define the emptiness formation probability (EFP) 
for the spin-$s$ case by  
\be 
\tau^{(2s \, +)}(m) = {\frac {\langle \psi_g^{(2s \, +)} 
| \widetilde{E}_1^{2s, 2s \, (2s \, +)} 
\cdots \widetilde{E}_m^{2s, 2s \, (2s \, +)} |\psi_g^{(2s \, +)} \rangle} 
{\langle \psi_g^{(2s \, +)} | \psi_g^{(2s \, +)} \rangle}} \, . 
\label{eq:EFPspin-s}
\ee
We shall denote $\tau^{(2s \, +)}(m)$ by $\tau^{(2s)}(m)$. 

Let us assume that  
Bethe roots $\{ \lambda_{\alpha}(\epsilon) \}_M$ 
with inhomogeneous parameters $w_j^{(2s; \, \epsilon)}$ 
($j=1, 2, \ldots, L$; $L=2sN_s$) 
become the ground-state solution 
of the spin-$s$ XXZ spin chain, $\{ \lambda_{\alpha} \}_M$, 
in the limit of sending $\epsilon$ to 0.  
We denote the Bethe vector with Bethe roots 
$\{\lambda_{\alpha}(\epsilon) \}_M$ by 
\bea 
| \psi_g^{(2s \, +; \, \epsilon)}  \rangle  & = & 
\prod_{\alpha=1}^{M} B^{(2s; \, \epsilon)}(\lambda_{\alpha}(\epsilon)) 
| 0 \rangle  
= e^{-\sum_{\alpha=1}^{M} \lambda_{\alpha}(\epsilon)} \, 
\chi_{1 2 \cdots L} \, \cdot \, 
 \prod_{\alpha=1}^{M} B^{(2s \, p; \, \epsilon)}
(\lambda_{\alpha}(\epsilon))  
| 0 \rangle     \non \\   
& = & 
e^{-\sum_{\alpha=1}^{M} \lambda_{\alpha}(\epsilon)} \, 
\chi_{1 2 \cdots L} \, 
| \psi_g^{(2s ; \, \epsilon)} \rangle 
\, . 
\eea
Here we recall the transformation inverse to (\ref{eq:gauge-transform}).  
We now calculate the norm of the spin-$s$ ground state 
from that of the spin-1/2 case through 
the limit of sending $\epsilon$ to 0 as follows.   
\bea 
& & {\langle \psi_g^{(2s \, +)} | \psi_g^{(2s \, +)} \rangle} 
= \lim_{\epsilon \rightarrow 0} \, 
\langle  \psi_g^{(2s; \, \epsilon)}
| \psi_g^{(2s; \, \epsilon)}  \rangle
 \non \\ 
& = & \lim_{\epsilon \rightarrow 0} \,  
\langle 0 | \prod_{k=1}^{M} C^{(2s; \, \epsilon)}(\lambda_k) 
\prod_{j=1}^{M} B^{(2s; \, \epsilon)}(\lambda_{j}) | 0 \rangle \non \\ 
 & = & \lim_{\epsilon \rightarrow 0} \,  
\sinh^M \eta \prod_{j, k=1; j \ne k}^{M} 
b^{-1}(\lambda_j(\epsilon), \lambda_k(\epsilon)) \, \cdot \, 
{\rm det} \Phi^{(1)'}\left( \{ \lambda_k(\epsilon) \}_M; 
\{ w_j^{(2s; \, \epsilon)} \}_L \right)  \non \\ 
& = & 
\sinh^M \eta \prod_{j, k=1; j \ne k}^{M} 
b^{-1}(\lambda_j, \lambda_k) \, \cdot \, 
{\rm det} \Phi^{(2s) '} 
\left( \{ \lambda_k \}_M; \{ \xi_p \}_{N_s} \right) 
\eea
where matrix elements of the spin-$s$ Gaudin matrix 
for $j, k= 1, 2, \ldots, M$,  
are given by    
\bea 
& & \Phi^{(2s) \, '}_{j, k}(\{ \lambda_{l} \}_M; \{\xi_p \}) \non \\ 
& = & - {\frac {\partial} {\partial \lambda_k}} 
\log \left( {\frac {a^{(2s)}(\lambda_j)}{d^{(2s)}(\lambda_j)}} \, 
 \prod_{t \ne j} {\frac {\sinh(\lambda_t- \lambda_j +\eta)}  
{\sinh(\lambda_t- \lambda_j - \eta)}}     
\right) \non \\ 
& = & \delta_{j, k} \left( 
\sum_{p=1}^{N_s} \frac {\sinh(2s \eta)} 
{\sinh(\lambda_j - \xi_p) \sinh(\lambda_j- \xi_p+ 2s \eta)} 
- \sum_{C=1}^{M} {\frac {\sinh 2\eta}  
{\sinh(\lambda_j - \lambda_{C} + \eta) 
\sinh(\lambda_j - \lambda_{C} - \eta)}} 
 \right) \non \\ 
& & + {\frac {\sinh 2\eta }   
 {\sinh(\lambda_j- \lambda_k + \eta)\sinh(\lambda_j- \lambda_k - \eta)}} 
\, . 
\eea
 
By applying formula (\ref{eq:Em=n}) with $n=2s$, 
the numerator of (\ref{eq:EFPspin-s}) is given by   
\bea 
& & \langle \psi_g^{(2s \, +; \, \epsilon)} | 
\widetilde{E}_1^{2s, 2s \, (2s \, +)} 
\cdots \widetilde{E}_m^{2s, 2s \, (2s \, +)} 
|\psi_g^{(2s \, +; \, \epsilon)} \rangle 
 =  \lim_{\epsilon \rightarrow 0}  
\langle \psi_g^{(2s; \, \epsilon)} 
| \prod_{k=1}^{m} E_k^{2s, 2s \, (2s)} \, 
|\psi_g^{(2s; \, \epsilon)} \rangle 
\non \\ 
&  & \quad = \lim_{\epsilon \rightarrow 0} 
\langle \psi_g^{(2s; \, \epsilon)} |  
P^{(2s)}_{1 2 \cdots L} \prod_{i=1}^{m} 
\Bigg( \prod_{\alpha=1}^{2s (i-1)} \left(A^{(2s; \, \epsilon)}
+D^{(2s; \, \epsilon)} \right)
(w_{\alpha}^{(2s; \, \epsilon)})       
\cdot \prod_{k=1}^{2s} D^{(2s; \, \epsilon)}(w_{2s(i-1)+k}^{(2s; \, \epsilon)}) \non \\ 
& & \qquad  \qquad \qquad \cdot 
\prod_{\alpha=1}^{2sN_s} 
 \left(A^{(2s; \, \epsilon)}+D^{(2s; \, \epsilon)} \right)
(w_{\alpha}^{(2s; \, \epsilon)}) \Bigg)      
P^{(2s)}_{1 2 \cdots L}  |\psi_g^{(2s)} \rangle
\non \\ 
& & \quad = 
 \prod_{j=1}^{m} \prod_{\alpha=1}^{M} 
b_{2s}(\lambda_{\alpha}, \xi_j) 
\,  \lim_{\epsilon \rightarrow 0} 
\langle \psi_g^{(2s; \, \epsilon)} | \, D^{(2s; \, \epsilon)}
(w_1^{(2s; \, \epsilon)}) \cdots 
D^{(2s; \, \epsilon)}(w_{2s m}^{(2s; \, \epsilon)}) \, 
| \psi_g^{(2s; \, \epsilon)} \rangle \, . \label{eq:GDDDG}
\eea
Let us set $\lambda_{M+j}(\epsilon)= w_j^{(2s; \, \epsilon)}$ 
for $j=1, 2, \ldots, 2s m$.  
Applying formula (\ref{eq:CCD}) to (\ref{eq:GDDDG})  
we have 
\bea 
& &  
\langle 0 | \prod_{\alpha=1}^{M} 
C^{(2s; \, \epsilon)}(\lambda_{\alpha}(\epsilon)) 
\prod_{j=1}^{2s m} 
D^{(2s;  \, \epsilon)}(\lambda_{M+j}(\epsilon)) \, 
\prod_{\beta=1}^{M} 
B^{(2s; \, \epsilon)}(\lambda_{\beta}(\epsilon)) \, 
| 0 \rangle 
\non \\ 
& = & \sum_{c_1=1}^{M} \sum_{c_2=1; c_2 \ne c_1}^{M} \cdots 
 \sum_{c_{2sm}=1; c_{2sm} \ne c_1, \ldots, c_{2sm-1}}^{M}
G_{c_1 \cdots c_{2s m}}
(\lambda_1(\epsilon), \cdots, \lambda_{M+ 2s m}(\epsilon); 
\{ w_j^{(2s; \, \epsilon)} \}_{L}) \non \\ 
& & \times 
\langle 0 | 
\prod_{k=1; k \ne c_1, \ldots, c_{2sm}}^{M+2sm} 
C^{(2s; \, \epsilon)}(\lambda_k(\epsilon)) 
\, \prod_{\alpha=1}^{M} B^{(2s; \, \epsilon)}
(w_j^{(2s; \, \epsilon)}) | 0 \rangle \, ,    
\label{eq:CDB}
\eea
where 
\be 
G_{c_1 \cdots c_{2sm}}
(\lambda_1, \cdots, \lambda_{M+2sm}; \{ w_j \}_{L} )  
= \prod_{j=1}^{2sm} \left( d(\lambda_{c_j}; 
\{ w_j \}_{L} ) 
{\frac {\prod_{t=1; t \ne c_1, \ldots, c_{j-1}}^{M+j-1} 
\sinh(\lambda_{c_j}-\lambda_t + \eta) }  
 {\prod_{t=1; t \ne c_1, \ldots, c_{j}}^{M+j} 
  \sinh(\lambda_{c_j}-\lambda_t) }}  
\right) \, . \label{eq:G}
\ee

We remark that from (\ref{eq:G}) 
the set of integers $c_1, \ldots, c_{2sm}$ 
of the most dominant terms in (\ref{eq:CDB}) are given by 
$m$ sets of $2s$-strings.  
If they are not, the numerator of (\ref{eq:G}) 
and hence the right-hand-side of (\ref{eq:CDB}) 
becomes smaller at least by  the order of $1/N_s$ in the large $N_s$ limit. 
However, each of the most dominant terms diverges with respect to $N_s$ 
in the large-$N_s$ limit,   
and they should cancel each other so that the final result becomes finite.   
We therefore calculate all possible contributions 
with respect to the set of integers, $c_1, c_2, \ldots, c_{2sm}$.

Let us take  a  sequence of distinct integers $c_j$ 
satisfying $1 \le c_j \le M$ for $j=1, 2, \ldots, 2sm$.  
We denote it  by $(c_j)_{2sm}$, i.e. 
$(c_j)_{2sm}=(c_1, c_2, \ldots, c_{2sm})$.    
Let us  denote by $\Sigma_M$  the set of integers, 
$1, 2, \ldots, M$: $\Sigma_M= \{ 1, 2, \ldots, M \}$. 
We then consider the complementary set of integers 
$\Sigma_M \setminus \{ c_1, \ldots, c_{2sm} \}$, and  
 put the elements in increasing order such as 
$z_1 < z_2 < \cdots < z_{M-2sm}$.   
We then extend the sequence $z_n$ of $M-2sm$ integers into 
that of $M$ integers by setting $z_{j+ M-2sm} = c_j$ 
for $j=1, 2, \ldots, 2sm$. 
We shall denote $z_n$ 
also by $z(n)$ for $n=1, 2, \ldots, M$.

In terms of sequence $(z_n)_M$ 
we express the scalar product in the last line of 
(\ref{eq:CDB}) as follows.    
\bea 
& & \langle 0 | 
\prod_{k=1; k \ne c_1, \ldots, c_{2sm}}^{M+2sm} 
C^{(2s; \, \epsilon)}(\lambda_k(\epsilon)) 
\, 
\prod_{\alpha=1}^{M} B^{(2s; \, \epsilon)}(\lambda_{\alpha}(\epsilon)) 
| 0 \rangle 
 \non \\    
&  & \quad = \langle 0 | 
\prod_{k=1}^{M-2sm} C^{(2s; \, \epsilon)}(\lambda_{z(k)}(\epsilon)) 
\prod_{j=1}^{2sm} C^{(2s; \, \epsilon)}(w_j^{(2s; \, \epsilon)}) \, 
\prod_{i=1}^{M-2sm} B^{(2s; \, \epsilon)}(\lambda_{z(i)}(\epsilon)) 
\prod_{j=1}^{2sm} B^{(2s; \, \epsilon)}
(\lambda_{c_j}(\epsilon)) | 0 \rangle . \non \\  
\label{eq:scalar}
\eea

We evaluate  scalar product (\ref{eq:scalar}), 
sending $\nu_j$ to $\lambda_{z(j)}(\epsilon)$  
for $j \le M-2sm$ and to $w_{j-M+2sm}^{(2s; \, \epsilon)}$ for $j > M-2sm$
in the following matrix:   
\be 
H^{(1)}((\lambda_{z(k)}(\epsilon))_M, 
(\nu_{z(1)}, \ldots, \nu_{z(M-2sm)}, 
\nu_{M-2sm+1}, \ldots , \nu_{M}); (w_j^{(2s; \, \epsilon)})_{L})  \, . 
\ee
Here we define the matrix elements 
$H_{ab}^{(2s)}( \{ \lambda_{\alpha} \}_n, \, \{ \mu_j \}_n ; \, 
\{\xi_k \}_{N_s} )$ 
for $a, b=1, 2, \ldots, n$, by 
\bea 
& & H_{ab}^{(2s)}( \{ \lambda_{\alpha} \}_n, \, \{ \mu_j \}_n 
; \, \{\xi_k \}_{N_s} ) 
\non \\ 
& & = {\frac {\sinh \eta} {\sinh(\lambda_{a} - \mu_b)} } 
\left( {\frac {a(\mu_b)} {d^{(2s)}(\mu_b;  \{\xi_k \})} } 
\prod_{k=1; k \ne a}^{n}  \sinh(\lambda_k - \mu_b + \eta) 
- \prod_{K=1; k \ne a}^{n} 
\sinh(\lambda_k - \mu_b - \eta) \right) \, . \non \\ 
\label{eq:Hab}
\eea

Let us denote $M-2sm$ by $M^{'}$. 
We write the composite of two sequences
$(a(i))_M$ and $(b(j))_N$ as $(a(i))_M \# (b(j))_N$. 
Explicitly we have 
\be 
(a(i))_M \# (b(j))_N= (a(1), \ldots, a(M), b(1), \ldots, b(N)) \, . 
\ee
 For $j > M^{'} = M-2sm$,  we have   
\bea 
& & 
 \lim_{\nu_{j} \rightarrow 
w_{j-M^{'}}^{(2s; \, \epsilon)}} 
d(\nu_j ; \{ w_j^{(2s; \, \epsilon)} \}_L) \,  
H^{(1)}_{i, \, j}(( \lambda_{z(k)}(\epsilon) )_M,  
 (\nu_k)_{M^{'}} \# (\nu_{k+M^{'}})_{2sm} ; 
(w_j^{(2s; \, \epsilon)} )_{L})   \non \\  
& = &  
\prod_{\alpha=1}^{M} 
\sinh(\lambda_{\alpha}(\epsilon)- w_{j-M^{'}}^{(2s; \, \epsilon)} + \eta) 
\Bigg(  
\displaystyle{\frac {\sinh \eta} 
{\sinh(\lambda_{z(i)}(\epsilon) - w_{j-M^{'}}^{(2s; \, \epsilon)}) 
\sinh(\lambda_{z(i)}(\epsilon) - w_{j-M^{'}}^{(2s; \, \epsilon)}+ \eta)}}  
\non \\  
& & \quad - d(w_{j-M^{'}}^{(2s; \, \epsilon)}; 
\{ w_j^{(2s; \, \epsilon)} \}_L) 
\prod_{t=1}^{M} 
{\frac {\sinh(\lambda_t(\epsilon) - w_{j-M^{'}}^{(2s; \, \epsilon)} - \eta)} 
      {\sinh(\lambda_t(\epsilon) - w_{j-M^{'}}^{(2s; \, \epsilon)} + \eta)}} 
\non \\ 
& &  \qquad \times \, \displaystyle{\frac {\sinh \eta} 
{\sinh(\lambda_{z(i)}(\epsilon) - w_{j-M^{'}-1}^{(2s; \, \epsilon)}) 
\sinh(\lambda_{z(i)}(\epsilon) - w_{j-M^{'}-1}^{(2s; \, \epsilon)}+ \eta)}} \Bigg) \, . 
\label{eq:HAB}
\eea 
The second term of (\ref{eq:HAB}) for matrix element $(i, j)$ 
vanishes since we have 
$d(w_{j-M^{'}}^{(2s; \, \epsilon)}; \{ w_k^{(2s; \, \epsilon)} \}_L)= 0$. 
Here we remark that if we directly evaluate matrix $H^{(2s)}$ 
at $\epsilon=0$,   
the  second term of (\ref{eq:HAB}) for matrix element $(i, j)$ 
for $j \ne 2s(n-1) +1 +M^{'}$ with $n=1, 2, \ldots, m$, 
does not vanish, 
although it is deleted by subtracting column $j$ by column $j-1$, 
as discussed for the XXX case in Ref. \cite{Kitanine2001}. 
We thus have  
\bea 
& & 
\lim_{\epsilon \rightarrow 0} 
{\rm det}H^{(1)}((\lambda_{z(k)}(\epsilon))_M,  
(\lambda_{z(1)}(\epsilon), \ldots, \lambda_{z(M-2sm)}(\epsilon), 
w_1^{(2s; \, \epsilon)}, \ldots, w_{2sm}^{(2s; \, \epsilon)} ); 
(w_{j}^{(2s; \, \epsilon)})_{2s N_s}) \non \\  
& = & (-1)^{M-2sm} \, 
\prod_{b=1}^{M-2sm} \prod_{k=1}^{M} \sinh(\lambda_k-\lambda_{z_b} - \eta) 
\prod_{j=1}^{2sm} \prod_{k=1}^{M} \sinh(\lambda_k- w_j^{(2s)} + \eta) 
\non \\ 
& \times & 
{\rm det} \Psi^{(2s) '}((\lambda_{z(i}))_{M}, 
(\lambda_{z(i}))_{M^{'}} \# (w_j^{(2s)})_{2sm}; \{ \xi_p \}_{N_s}) 
\eea
where  $(i,j)$ element of 
$\Psi^{(2s)'}((\lambda_{z(k)})_{M}, 
(\lambda_{z(k)})_{M^{'}} \# (w_k^{(2s)})_{2sm}; \{ \xi_p \}_{N_s})$ 
for $i= 1, 2, \ldots, M$,  
are given by 
\bea 
& &  \Psi^{(2s) '}_{i, \, j} ((\lambda_{z(1)}, \ldots, \lambda_{z({M-2sm})}, 
\lambda_{c_1}, \ldots, \lambda_{c_{2sm}}), 
(\lambda_{z(1)}, \ldots, \lambda_{z({M-2sm})}, 
w_1^{(2s)}, \ldots, w_{2sm}^{(2s)}); ( \xi_p )_{N_s}) 
\non \\ 
&  & \quad = 
\left\{ 
\begin{array}{cc} 
\Phi^{(2s)'}_{z(i) , \, z(j)}( \left( \lambda_{k} \right)_M; \{\xi_p \}) 
& \mbox{for} \quad j \le M-2sm , \\  
 & \\
\displaystyle{\frac {\sinh \eta} 
{\sinh(\lambda_{z(i)} - w_{j-M^{'}}^{(2s)}) \sinh(\lambda_{z(i)} 
- w_{j-M^{'}}^{(2s)}+ \eta)}}  
& \mbox{for} \quad j > M-2sm \, .  \\ 
\end{array}  
\right.
\eea  
Therefore, for $i,j= 1, 2, \ldots, M-2sm$, we have 
\be 
\left(\left(\Phi^{(2s)'}
((\lambda_{z(k)})_M; \{ \xi_p \}) \right)^{-1} \, 
\Psi^{(2s)'}((\lambda_{z(k)})_M, (\lambda_{z(k)})_{M^{'}} \# 
(w_j^{(2s)})_{2sm}; \{ \xi_p \} ) \right)_{i, \, j} = \delta_{i, \, j} \, . 
\label{eq:diag}
\ee

In terms of sequence $( c_j )_{2sm}$, 
we express the dependence of matrix 
$(\Phi^{(2s)'})^{-1} \Psi^{(2s)'}$
on the sequence of Bethe roots $(\lambda_{z(i)})_M$ etc., 
briefly, as follows. 
\be 
(\Phi^{(2s)'})^{-1} \Psi^{(2s)'}(( c_j )_{2sm}, \{ \xi_p \})   
= (\Phi^{(2s)'})^{-1} \Psi^{(2s)'} 
\left( (\lambda_{z(j)} )_M, 
(\lambda_{z(j)})_{M^{'}} \# (w_j^{(2s)})_{2sm}; \{ \xi_p \} \right) \, .  
\ee
Recall $M^{'}=M-2sm$. 
Similarly, we define ${\Phi}^{(2s) \, '}((c_j)_{2sm}, \{ \xi_p \})$ 
and $\Psi^{(2s)'}((c_j)_{2sm}, \{ \xi_p \})$ by 
\bea 
{\Phi}^{(2s) \, '}_{i , \, j}((c_l)_{2sm}, \{ \xi_p \})  
& = & 
\Phi^{(2s) \, '}_{i , \, j}(\{ \lambda_{z(k)} \}_M; ( \xi_p )_{N_s}) 
= \Phi^{(2s) \, '}_{z(i) , \, z(j)}(\{ \lambda_{k} \}_M; 
( \xi_p )_{N_s}) \, , 
\non \\  
\Psi^{(2s)'}_{i, \, j}((c_k)_{2sm}, \{ \xi_p \}) & = & 
 \Psi^{(2s)'}_{i, \, j}((\lambda_{z(k}))_{M}, 
 (\lambda_{z(k}))_{M^{'}} \# (w_k^{(2s)})_{2sm}; ( \xi_p )_{N_s}) \, , 
\eea
for $i,j = 1, 2, \ldots, M$. 
Here we remark again that sequence  
$(z(i))_M$ is determined by sequence 
$(c_1, \ldots, c_{2sm})$ by the definition that  
$\{ z(1), z(2), \ldots, z(M^{'}) \} = \{1, 2, \ldots, M \} \setminus 
\{c_1, \ldots, \, c_{2sm}  \}$, 
and $z(1) < \cdots < z(M^{'})$ while $z(j+M^{'}) = c_{j}$ for 
$j=1, 2, \ldots, 2sm$.  

From property (\ref{eq:diag}) 
we define a $2sm$-by-$2sm$ matrix 
$\phi^{(2s; \, m)}((c_j)_{2sm}; \{ \xi_p \})$ by 
\be 
\phi^{(2s; \, m)}((c_j)_{2sm}; \{ \xi_p \})_{j, \, k} 
= 
\left( \left(\Phi^{(2s)'} \right)^{-1} 
 \Psi^{(2s)'} ((c_j)_{2sm}; \{ \xi_p \})
\right)_{j+M^{'}, \, k+M^{'}} \quad 
\mbox{for} \, \, j,k = 1, 2, \ldots, 2sm . 
\ee
Making use of (\ref{eq:ratio}), 
we  obtain the spin-$s$ EFP for the finite-size chain as follows. 
\begin{eqnarray}
& & \tau^{(2s)}_{N_s}(m)  =  
\frac{1}{\prod_{1 \leq j<r\leq 2s} \sinh^{m}(r-j)\eta}  
\times \frac{1}{\prod_{1\leq k<l\leq m}\prod^{2s}_{j=1}\prod^{2s}_{r=1}
\sinh(\xi_k-\xi_l+(r-j)\eta)}   \non  \\
& & \times  \sum_{c_1=1}^{M} \sum_{c_2=1; c_2 \ne c_1}^{M} 
\cdots \sum_{c_{2sm}=1; c_{2sm} \ne c_1, \ldots, c_{2sm-1}}^{M}
H^{(2s)}(\lambda_{c_1}, \cdots, \lambda_{c_{2sm}}; \{ \xi_p \}) 
 \,  {\rm det} \left( \phi^{(2s; \, m)}((c_j)_{2sm}; \{ \xi_p \}) \right)
\non \\ 
%
\label{eq:EFP-finite}
\end{eqnarray}
where $H^{(2s)}(\lambda_{c_1}, \ldots, \lambda_{c_{2sm}} ; \{ \xi_p \} )$
is given by 
\begin{eqnarray}
& & H^{(2s)}(\lambda_{c_1}, \ldots, \lambda_{c_{2sm}} ; \{ \xi_p \} ) 
\non \\ 
& & \quad = \frac{1}{\prod_{1 \leq l<k \leq 2sm}
\sinh(\lambda_{c_k}-\lambda_{c_l}+\eta)}
\times
\prod^{2sm}_{j=1}
\prod^{m}_{b=1}
\prod^{2s-1}_{\beta=1}
\sinh(\lambda_{c_j}-\xi_b+ \beta \eta)
\non \\
& & \quad \times  
\prod^{m}_{l=1} \prod^{2s}_{r_l=1}
\left(\prod^{l-1}_{b=1}
\sinh(\lambda_{c_{2s(l-1)+r_l}}-\xi_b+ 2s\eta)
\prod^m_{k=l+1}\sinh(\lambda_{c_{2s(l-1)+r_l}}-\xi_b)
\right) \, . 
\end{eqnarray}

%
\subsection{Diagonal elements of the spin-$s$  Gaudin matrix}

Let us define $K_n(\lambda)$ for $\eta= i \zeta$ 
with $0 < \zeta < \pi$ by 
\be 
K_n(\lambda) = {\frac 1 {2 \pi i}} \, {\frac {\sinh( n \eta)} 
{\sinh(\lambda - n \eta/2) \sinh(\lambda + n \eta/2)}} \, .  
\ee

\begin{lemma} 
For $0 < \zeta < \pi/2s$, we have 
\be 
K_1(\lambda + n \eta) =  
\int_{-\infty}^{\infty} K_2(\lambda - \mu + n \eta + i \epsilon) 
\rho(\mu) d \mu   \quad (n= 1,  2, \ldots, 2s-1) \, ,   
\label{eq:K1K2-p}
\ee
and 
\be 
K_1(\lambda - n \eta) =  
\int_{-\infty}^{\infty} K_2(\lambda - \mu - n \eta - i \epsilon) 
\rho(\mu) d \mu   \quad (n= 1,  2, \ldots, 2s-1).  
\label{eq:K1K2-m}
\ee
Here we recall $\eta = i \zeta$. 
\label{lem:Kint}
\end{lemma} 
\begin{proof}
We first consider the case of positive $n$. 
Let us recall  the Lieb equation 
\be 
\rho(\lambda) =  
K_1(\lambda) - \int_{-\infty}^{\infty} 
K_2(\lambda - \mu) \rho(\mu) d \mu \, .  \label{eq:Lieb}
\ee
Shifting variable $\lambda$ analytically to 
$\lambda + i \zeta - i \epsilon$ in (\ref{eq:Lieb})
we have 
\be
\rho(\lambda + i \zeta - i \epsilon) 
= K_1(\lambda + i \zeta - i \epsilon) - \int_{-\infty}^{\infty} 
K_2(\lambda- \mu + i \zeta - i \epsilon) \rho(\mu) d \mu
\ee
Using  
\be 
{\frac 1 {\sinh(\lambda - \mu - i \epsilon)}} 
= 
{\frac 1 {\sinh(\lambda - \mu + i \epsilon)}} 
+ 2\pi i \delta(\mu-\lambda) 
\label{eq:HT}
\ee
we have  
\be 
\int_{-\infty}^{\infty} 
K_2(\lambda- \mu + i \zeta - i \epsilon) \rho(\mu) d \mu
= 
\int_{-\infty}^{\infty} 
K_2(\lambda- \mu + i \zeta + i \epsilon) \rho(\mu) d \mu
+ \rho(\lambda) \, . 
\ee
Combining $\rho(\lambda + i \zeta) = - \rho(\lambda)$ we obtain eq. 
(\ref{eq:K1K2-p}) for $n=1$. Making analytic continuation 
with respect to $\lambda$ we derive eq. (\ref{eq:K1K2-p}) 
for $n=2, 3, \ldots, 2s-1$. Similarly, we can show 
(\ref{eq:K1K2-m}). 
\end{proof} 

\begin{proposition} When $0 < \zeta < \pi/2s$,  
matrix elements $(A, A)$ of the spin-$s$  Gaudin matrix 
with $A=2s(a-1)+ \alpha$ are evaluated by 
\be 
{\frac 1 { 2 \pi i \, N_s}} \, \Phi^{(2s) \, '}_{A , \, A}(\{ \lambda_j \}_M) 
=  \rho_{\rm tot}(\mu_{a}) +O(1/N_s) \, . \label{eq:diagAA}
\ee 
Relations (\ref{eq:diagAA}) are expressed 
in terms of integrals as follows.  
\be 
 \rho_{\rm tot}(\mu_{a}) = 
{\frac 1 {N_s}} \sum_{p=1}^{N_s} 
K_{2s}(\mu_a - (\alpha-1/2) \eta - \xi_p + s \eta) - 
\sum_{\gamma=1}^{2s} \int^{\infty}_{- \infty} 
 K_2(\mu_a - \mu_c - (\alpha - \gamma) \eta + \epsilon^{(\alpha, \gamma)}) 
 \rho_{\rm tot}(\mu_c) d \mu_c \, , \label{eq:PHIAA} 
\ee
where 
$
\epsilon^{(\alpha, \gamma)} 
= \epsilon^{(\alpha)}_a - \epsilon^{(\gamma)}_c$.  
We recall that $C$ corresponds to $(c, \gamma)$ with $C=2s(c-1)+\gamma.$  
\label{prop:Phiaa}
\end{proposition}
\begin{proof}
Let us first show  
\be 
K_{2s} (\mu_a - (\alpha - 1/2) \eta + s \eta) 
 - \sum_{\gamma=1}^{2s} \int^{\infty}_{- \infty} 
 K_2(\mu_a - \mu_c + (\gamma - \alpha) \eta 
  + \epsilon^{(\alpha, \gamma)}) \rho(\mu_c) d \mu_c = \rho(\mu_a) 
\label{eq:homo}
\ee
for $\alpha= 1, 2 \ldots, 2s$. 
Making use of the following relations   
\be 
K_n(\lambda) = {\frac 1 {2\pi i}} \left( 
{\frac {\cosh(\lambda- n \eta/2)} {\sinh(\lambda- n \eta/2)}} 
- {\frac {\cosh(\lambda + n \eta/2)} {\sinh(\lambda +  n \eta/2)}} 
 \right)
\ee
we have 
\be 
K_{2s}(\lambda + (2s-1)\eta/2) = 
\sum_{n=0}^{2s-1} K_1(\lambda + n \eta)  \, . 
\label{eq:Ksum}
\ee
We thus obtain (\ref{eq:homo}) as follows.  
\bea 
& & K_{2s} (\mu_a - (\alpha - 1/2) \eta + s \eta) 
 - \sum_{\gamma=1}^{2s} \int^{\infty}_{- \infty} 
 K_2(\mu_a - \mu_c + (\gamma - \alpha) \eta 
  + \epsilon^{(\alpha, \gamma)}) \rho(\mu_c) d \mu_c \non \\ 
&=& \sum_{\gamma=1}^{2s} 
\left( K_1(\mu_a + (\gamma- \alpha) \eta) 
- \int_{-\infty}^{\infty} K_2(\mu_a - \mu_c + (\gamma - \alpha) \eta 
  + \epsilon^{(\alpha, \gamma)}) \rho(\mu_c) d \mu_c
\right)  
\non \\ 
&=&  
K_1(\mu_a) 
- \int_{-\infty}^{\infty} K_2(\mu_a - \mu_c) 
\rho(\mu_c) d \mu_c \non \\ 
& = & \rho(\mu_a) \, . \label{eq:drv}
\eea
Here, in the second line of (\ref{eq:drv}),  
the summands for $\gamma \ne \alpha$ vanish due to lemma \ref{lem:Kint}.    
We then apply  the Lieb equation (\ref{eq:Lieb}) to show 
the last line of (\ref{eq:drv}). 
We obtain (\ref{eq:PHIAA}) from (\ref{eq:homo}).  
\end{proof}

\begin{corollary} Let us take a sequence of integers, 
 $c_1, c_2, \ldots, c_{2sm}$,  
which satisfy $1 \le c_j \le M$ for $j=1, 2, \ldots, 2sm$, 
and determine a sequence $(z(n))_M$  by the conditions 
that  $\{z(1), z(2), \ldots, z(M^{'}) \} = \{1, 2, \ldots, M \} \setminus 
\{ {c_1}, \ldots, {c_{2sm}}  \}$, 
with $z(1) < \cdots < z(M^{'})$ and  $z(j+M^{'}) = c_{j}$ for 
$j=1, 2, \ldots, 2sm$.  Here we recall $M^{'}=M-2sm$.   
In the region: $0 < \zeta < \pi/2s$,  
diagonal elements $(j, j)$ of the spin-$s$  Gaudin matrix 
$\Phi^{(2s)'}((c_k)_{2sm})$ are evaluated as  
\be 
{\frac 1 {2 \pi i \, N_s}} \, \Phi^{(2s) \, '}_{j , \, j}((c_k)_{2sm}) 
=  \rho_{\rm tot}(\mu_{a}) + O(1/N_s) \, ,  
\quad \mbox{for} \, \, j = 1, 2, \ldots, M.  
\label{eq:Phijj}
\ee
Here integer $a$ satisfies $z(j)=2s(a-1)+ \alpha$ for an integer $\alpha$ 
with $1 \le \alpha \le 2s$. 
\label{cor:Phiaa}
\end{corollary}

\subsection{Integral equations}

We calculate matrix elements of 
$\left( (\Phi^{(2s)'})^{-1} \, \Psi^{(2s)'} \right)((c_j)_{2sm})$ 
through the spin-$s$ Gaudin matrix. 
For $j, k= 1, 2, \ldots, M$,  
we have 
\bea 
\left( \Psi^{(2s)'}((c_j)_{2sm})  \right)_{j , \, k}
& = & \left( \Phi^{(2s)'} ( \Phi^{2s'})^{-1} \Psi^{(2s)'} \right)_{j , \, k}
\non \\ 
& = & \sum_{t=1}^{M} \left( \Phi^{(2s)'}((c_j)_{2sm}) \right)_{j ,  \, t}
\left( (\Phi^{(2s)'})^{-1}\Psi^{(2s)'} ((c_j)_{2sm})
\right)_{t , \, k} \, . \label{eq:decomp}
\eea
We remark that matrix elements $(A, B)$ of $\Psi^{(2s)'}((c_l)_{2sm})$ 
with $A=j+M^{'}$ and $B=k+M^{'}$ 
are expressed in terms of $K_1(\lambda)$ as   
\be 
\Psi^{(2s)'}_{j+M^{'}, \, k+M^{'}}((c_j)_{2sm})/ 2 \pi i  
= K_1(\lambda_{c_{j}}-w_{k}^{(2s)} + \eta/2) \quad 
\mbox{for} \quad 
j, k= 1, 2, \ldots, 2sm. 
\ee 

Suppose that we have  a sequence $(z(n))_M$ for   
 a given sequence $(c_i)_{2sm}$  
satisfying $1 \le c_i \le M$ for $i=1, 2, \ldots, 2sm$.   
Let us take a pair of integers $j, k$ with $1 \le j, k \le M$.    
We denote $z(j)$ by $A$, and we introduce $a$ and $\alpha$ by      
$A=2s(a-1)+\alpha$ with  $1 \le a \le N_s/2$ and $1 \le \alpha \le 2s$. 
Applying proposition \ref{prop:Phiaa} 
and corollary \ref{cor:Phiaa} to (\ref{eq:decomp})  
we have 
\begin{multline}
\hspace{1.0cm} \sum_{t=1}^{M} 
\Phi^{(2s)'}_{j , \, t}((c_l)_{2sm})/2\pi i \, 
\left( (\Phi^{(2s)'})^{-1} \Psi^{(2s)'} ((c_l)_{2sm}) \right)_{t , \, k} \\ 
= \sum_{t=1}^{M} \Biggl\{ \left(\sum^{N_s}_{p=1}
{\frac 1 {2 \pi i}} \,  
\frac{\sinh(2s \eta)}{\sinh(\lambda_A-\xi_p)
\sinh(\lambda_A - \xi_p + 2s \eta)}
-\sum^{M}_{D=1} K_2(\lambda_A-\lambda_D) \right) \delta_{A , \, z(t)} \\ 
\hspace{6cm}
    +K_2(\lambda_A -\lambda_{z(t)}) \Biggr\} 
\left( (\Phi^{(2s)'})^{-1}\Psi^{(2s)'}((c_l)_{2sm})  \right)_{t , \, k}        
\\ 
= N_s \rho_{\rm tot}(\mu_a) 
\left( (\Phi^{(2s)'})^{-1} \Psi^{(2s)'} ((c_j)_{2sm}) \right)_{j , \, k}  
+\sum_{t=1}^{M} K_2(\lambda_A-\lambda_{z(t)}) \, 
\left( (\Phi^{(2s)'})^{-1} \Psi^{(2s)'}((c_l)_{2sm}) \right)_{t , \, k}  
+ O(1/N_s) \, . 
\label{eq:PPP}
\end{multline}

Let us discuss the order of magnitude of 
the correction term in (\ref{eq:PPP}). 
It follows from (\ref{eq:Phijj}) that 
if the density of string centers $\rho(\mu_{a})$ is $O(1)$ 
in the large $N_s$ limit, 
then the diagonal element $(A,A)$ of $\Phi^{(2s) '}$ is 
$O(N_s)$.  
We thus suggest that the matrix elements 
of $\phi^{(2s; \, m)}((c_l)_{2sm})$ should be at most 
of the order of $1/N_s$,  
and hence the correction term in (\ref{eq:PPP}) 
should be at most $O(1/N_s)$.

Let us now define $\varphi_{A , \, B}(\{\xi_p \})$ 
by the following relations  for $j, k = 1, 2, \ldots, M$:  
\be 
\varphi_{z(j) , \, k} = N_s \rho_{\rm tot}(\mu_a) \, 
\left( ( \Phi^{(2s)'})^{-1} 
\Psi^{(2s)'}((c_l)_{2sm}) \right)_{j , \, k } \, ,   
\ee
where integer $a$  is given by $a= [(z(j)-1)/2s]+1$. 
In terms of function $a(z)=[(z-1)/2s]+1$ we have    
\bea 
& & \sum_{t=1}^{M} K_2(\lambda_A-\lambda_{z(t)}) \, 
\left( (\Phi^{(2s)'})^{-1} \Psi^{(2s)'}((c_l)_{2sm}) \right)_{t , \, k}  
\non \\ 
& = & \sum_{t=1}^{M} K_2(\lambda_A-\lambda_{z(t)}) 
{\frac {\varphi_{z(t), \, k}} {N_s \, \rho_{\rm tot}(\mu_{a(z(t))})} } 
\non \\ 
& = & \sum_{\gamma=1}^{2s} {\frac 1 {N_s}} \sum_{c=1}^{N_s/2} 
(\rho_{\rm tot}(\mu_{c}))^{-1}  \, 
K_2(\lambda_A-\lambda_{(c, \gamma)}) \, 
\varphi_{2s(c-1)+ \gamma, \, k} \, . 
\eea
Here we have replaced the sum over $t$ by the sum over 
$c$ and $\gamma$ where $z(t)=C=2s(c-1)+\gamma$ with $1 \le \gamma \le 2s$. 
Expressing $z(j)$ and $k$ by $A$ and $B$, respectively, 
we have the following equations:  
\be 
\varphi_{A, \,  B}(\{ \xi_p \}) + 
\sum_{\gamma=1}^{2s} {\frac 1 {N_s}} \sum_{c=1}^{N_s/2} 
(\rho_{\rm tot}(\mu_{c}))^{-1}  \, 
K_2(\lambda_A-\lambda_{(c, \gamma)}) 
\varphi_{C, \, B}(\{ \xi_p \} )  \,    
= {\frac 1 {2 \pi i}} \,  \Psi^{(2s)'}_{A , \, B}((c_l)_{2sm})  
+ O(1/N_s) \, . 
\label{eq:sumAB}
\ee
Let us introduce $b$ and $\beta$ by $k=2s(b-1) + \beta + M^{'}$ with  
$1 \le \beta \le 2s$ and $1 \le b \le N_s/2$. 
%
In terms of string center $\mu_a$ we express (or approximate) 
$\varphi_{z(j) , \, k}(\{\xi_p \})$ by a continuous function 
of $\mu_a$ and $\xi_b$, as follows.   
\be 
\varphi_{z(j), \, k}(\{ \xi_p \}) 
= \varphi_{\alpha}^{(\beta)}(\mu_a, \xi_b) + O(1/N_s) \, .   
\ee   
By taking the large-$N_s$ limit,  
the discrete equations (\ref{eq:sumAB}) are now expressed  as follows.
\be 
\varphi_{\alpha}^{(\beta)}(\mu_a, \xi_b) 
+ \sum_{\gamma=1}^{2s} \int_{-\infty}^{\infty} 
K_2(\mu_a- \mu_c + (\gamma - \alpha) \eta + \epsilon_{AC} )
 \varphi_{\gamma}^{(\beta)}(\mu_c, \xi_b) d \mu_c 
= K_1(\lambda_{c_{j-M^{'}}}- w_k^{(2s)} + \eta/2) \, . 
\label{eq:int}
\ee
Here $\epsilon_{AC} = 
\epsilon_{a}^{(\alpha)} - \epsilon_{c}^{(\gamma)}$. 
We recall that for $j > M^{'}$ we have set $z(j)=c_{j-M^{'}}$.

\begin{lemma}
In the region $0 < \zeta < \pi/2s$, 
 a solution to the integral equations (\ref{eq:int}) 
for integers $A=2s(a-1)+ \alpha$ (i.e. $(a, \alpha)$)  
and $B=2s(b-1)+ \beta+M^{'}$ with 
$1 \le \alpha, \beta \le 2s$ and $1 \le b \le m$ is given by  
\be 
\varphi_{A , \, B} = \varphi_{\alpha}^{(\beta)}(\mu_a, \xi_b)  = 
\, \, \rho(\mu_a - \xi_b) \, \delta_{\alpha, \, \beta} . 
\ee
\label{lem:IntEq}
\end{lemma}
\begin{proof} 
(i) In $(\alpha, \alpha)$ case, i.e. when integers 
 $A=2s(a-1)+ \alpha$  and $B=2s(b-1) + \alpha$ 
 correspond to indices $(a, \alpha)$ and $(b, \alpha)$, 
respectively,    assuming that  
$\varphi_{\gamma}^{(\alpha)}(\mu_c, \xi_b) =0$ for $\gamma \ne \alpha$, 
we reduce integral equations (\ref{eq:int}) 
to the Lieb equation (\ref{eq:Lieb}). 
Therefore, we have 
$\varphi_{\alpha}^{(\alpha)}(\mu_a, \xi_b)  
= \rho(\mu_a - \xi_b)$. 
\par \noindent 
(ii) In $(\alpha, \beta)$ case, i.e. when $A=(a, \alpha)$ 
and $B=(b, \beta)$ with $\beta \ne \alpha$,   
assuming $\varphi_{\gamma}^{(\beta)}(\mu_c, \xi_b) =0$ 
for $\gamma \ne \beta$, we have from (\ref{eq:int})
\be 
\int_{-\infty}^{\infty} 
K_2(\mu_a- \mu_c + (\beta - \alpha) \eta + \epsilon_{AB} )
 \varphi_{\beta}^{(\beta)}(\mu_c, \xi_b) d \mu_c 
= {\frac 1 {2 \pi i}} \,  \Psi^{(2s)}_{A , \, B} \, . 
\ee
For $\alpha < \beta$, we have $\epsilon_{AB} = i \epsilon$. 
Shifting $\mu_a$ analytically such as 
$\mu_a \rightarrow \mu_a - (\beta - \alpha -1) \eta$,  we have 
\be 
\int_{-\infty}^{\infty} 
K_2(\mu_a- \mu_c +  \eta + i \epsilon )
 \varphi_{\beta}^{(\beta)}(\mu_c, \xi_b) d \mu_c 
= {\frac 1 {2 \pi i}} \,  
{\frac {\sinh \eta} {\sinh(\mu_a + \eta - \xi_b - \eta/2) 
\sinh(\mu_a + \eta - \xi_b + \eta/2)}}  \, . 
\ee
Making use of (\ref{eq:HT}) we reduce it essentially 
to the Lieb equation. We thus obtain  
$\varphi_{\beta}^{(\beta)}(\mu_a, \xi_b) = \rho(\mu_a - \xi_b)$.  
For $\alpha > \beta$, we have $\epsilon_{AB} = - i \epsilon$, and 
show it similarly, shifting $\mu_a$ analytically as 
$\mu_a \rightarrow \mu_a - (\alpha - \beta +1) \eta$.
\end{proof}

\begin{proposition}
Let us take a set of integers, $c_1, \ldots, c_{2sm}$, 
satisfying $0 < c_j \le M$ for $j=1, 2, \ldots, 2sm$.  
Suppose that the number of $c_j$ which satisfy 
$c_j - 2s [(c_j-1)/2s] = \alpha$ 
is given by $m$ for each integer $\alpha$ satisfying 
$1 \le \alpha \le 2s$. 
Then, when $0 < \zeta < \pi/2s$, 
the solution to integral equations (\ref{eq:int}) 
for $A= c_j$ with $j=1, 2, \ldots, 2sm$ and  
for $B = B^{'} + M^{'}$ where $B^{'} = 1, 2, \ldots, 2sm$,  
is given by  
\be 
\varphi_{\alpha}^{(\beta)}(\mu_{a_j}, \xi_b)  = 
\, \, \rho(\mu_{a_j} - \xi_b) \, \delta_{\alpha, \, \beta} \, ,   
\label{eq:sol-cjB}
\ee
where $a_j = [(c_j-1)/2s] +1$, $\alpha=c_j - 2s [(c_j-1)/2s]$   
 and $B^{'}=2s(b-1) + \beta$ with 
$1 \le \beta \le 2s$. 
\end{proposition}
\begin{proof}
It follows from lemma \ref{lem:IntEq} that   
$\varphi_{c_j , \, B}$ of (\ref{eq:sol-cjB}) gives a solution to the 
integral equations. Taking the Fourier transform of (\ref{eq:int}),    
we show in \S 5.5 that the set of integral equations 
  (\ref{eq:int}) for $A= c_j$ for $j=1, 2, \ldots, 2sm$ and  
$B^{'}=1, 2, \ldots, 2sm$ has a unique solution. 
Thus we obtain the unique solution (\ref{eq:sol-cjB}). 
\end{proof}

Let us recall the assumption 
that function $\varphi_{\alpha}^{(\beta)}(\mu_a, \xi_b)$ 
is continuous with respect to $\mu_a$ and $\xi_b$. 
Then, for any given set of integers, $c_1, \ldots, c_{2sm}$ 
satisfying $0 < c_j \le M$ for $j=1, 2, \ldots, 2sm$, we may 
approximate the matrix elements of 
 $(\Phi^{(2s)'})^{-1} \Psi^{(2s)'}$  as follows. 
 For integers $j$ and $k$ with $1 \le j, k \le 2sm$, 
we define $a_j$, $\alpha_j$, $b_k$ and $\beta_k$ as follows. 
\bea 
a_j & = & [(c_j-1)/2s]+1 \, , \quad  
\alpha_j =c_j - 2s [(c_j-1)/2s] \, ,  \non \\
b_k & = & [(k-1)/2s]+1 \, , \quad \beta_k = k - 2s [(k-1)/2s] \, . 
\label{eq:aabb}
\eea
Then, we have 
\be 
 \left( (\Phi^{(2s)'})^{-1} \Psi^{(2s)'} 
\left(   ({c_j} )_{2sm} \right) \right)_{j+M^{'} , \, k+M^{'}}    
 =  {\frac 1 {N_s}} \, {\frac {\rho(\mu_{a_j} - \xi_{b_k})} 
{\rho_{\rm tot}(\mu_{a_j})}}   
\delta_{\alpha_j, \, \beta_k} + O({1/{N_s^2}}) \, .   
\ee
Here we recall $M^{'}=M-2sm$. 

For a given $2s$-string, $\lambda_{(a, \alpha)}$,  
with $\alpha=1, 2, \ldots, 2s$, we define 
$\lambda_{(a, \alpha)}^{'}$ by the `regular part' 
of $\lambda_{(a, \alpha)}$: 
\be 
\lambda_{(a, \alpha)}^{'} = \mu_a - (\alpha -1/2) \eta . 
\ee
Let us introduce a $2sm$-by-$2sm$ matrix $S$ by 
\be 
S_{j, \, k}(c_1, \ldots, c_{2sm}; (\xi_p )_{N_s}) 
= \rho(\lambda^{'}_{c_j} - w_k^{(2s)} + \eta/2) \, 
\delta_{\alpha_j , \, \beta_k} \quad \mbox{for} 
\, \, \, j, k= 1, 2, \ldots, 2sm.   
\ee
Here $a_j$, $\alpha_j$, $b_k$ and $\beta_k$ are given by (\ref{eq:aabb}). 
Then, we obtain  
\be 
\phi^{(2s; \, m)}_{j, \, k} ((c_k)_{2sm}; \{ \xi_p \}) = {\frac 1 N_s} 
{\frac 1 {\rho_{\rm tot}(\mu_a)}} \, 
S_{j, \, k}((c_k)_{2sm}; (\xi_p )_{N_s}) + O({1/{N_s^2}}) \, . 
\label{eq:phi-S}
\ee
and we have 
\be 
{\rm det} \left(
\phi^{(2s; \, m)}((c_k)_{2sm}; \{ \xi_p \}) \right)
=  
\prod_{j=1}^{2sm} \left( 
{\frac 1 {N_s}} \, {\frac 1 {\rho_{\rm tot}(\mu_{a_j})}} \right) 
\, \cdot \, \Big( 
{\rm det} S(( c_j )_{2sm}; (\xi_p )_{N_s} ) \, + O(1/N_s) \Big). 
\label{eq:detPhi=detS}
\ee

%
%
\subsection{Fourier transform in the cases of spin-$1$ and general spin-$s$}

The integral equations (\ref{eq:int}) 
for the spin-1 case are given by 
\begin{equation}
\left\{
\begin{array}{l}
\displaystyle
    \varphi^{(1)}_1(\mu,\xi_p)+ \int^{\infty}_{-\infty} K_2(\mu-\lambda) 
  \varphi^{(1)}_1(\lambda,\xi_p) \, d\lambda
+\int^{\infty}_{-\infty} K_2(\mu-\lambda+\eta+ \epsilon^{(1,2)}) 
            \varphi^{(1)}_{2}(\lambda, \xi_p)  d\lambda      
\\ \hspace{8cm}
\displaystyle= {\frac 1 {2 \pi i}} \,  \frac{\sinh(\eta)}
{\sinh(\mu-\xi_p + {\frac{\eta}{2}})\sinh(\mu-\xi_p - {\frac{\eta}{2}})} \\
\displaystyle
    \varphi^{(2)}_1(\mu, \xi_p) + \int^{\infty}_{-\infty} 
    K_2(\mu-\lambda) \varphi^{(2)}_1(\lambda,\xi_p) d\lambda 
   + \int^{\infty}_{-\infty} K_2(\mu-\lambda+\eta+ \epsilon^{(1,2)}) 
    \varphi^{(2)}_2(\lambda,\xi_p)  d\lambda    
\\ \hspace{8cm}
\displaystyle={\frac 1 {2 \pi i}} \, 
\frac{\sinh(\eta)}{\sinh(\mu-\xi_p+\frac{\eta}{2}) 
                                 \sinh(\mu-\xi_p+{\frac{3 \eta}{2}})} \\    
\displaystyle
    \varphi^{(1)}_2(\mu,\xi_p)+\int^{\infty}_{-\infty}
     K_2(\mu-\lambda-\eta + \epsilon^{(2,1)}) 
      \varphi^{(1)}_1(\lambda,\xi_p) d\lambda 
    + \int^{\infty}_{-\infty}  K_2(\mu-\lambda) 
      \varphi^{(1)}_2(\lambda, \xi_p) d\lambda           
\\ \hspace{8cm}
\displaystyle= {\frac 1 {2 \pi i}} \, 
\frac{\sinh(\eta)}
{\sinh(\mu-\xi_p-\frac{3 \eta}{2})\sinh(\mu-\xi_p - \frac{\eta}{2})}\\
\displaystyle
    \varphi^{(2)}_2(\mu, \xi_p) + \int^{\infty}_{-\infty}
     K_2( \mu-\lambda-\eta+ \epsilon^{(2,1)}) 
     \varphi^{(2)}_1(\lambda, \xi_p) d\lambda 
+\int^{\infty}_{-\infty} K_2(\mu-\lambda) 
   \varphi^{(2)}_2(\lambda, \xi_p) d\lambda        
\\ \hspace{8cm}
\displaystyle= {\frac 1 {2 \pi i}} \, 
\frac{\sinh(\eta)}
{\sinh(\mu-\xi_p - \frac{\eta}{2})\sinh(\mu-\xi_p+\frac{\eta}{2})}     
\end{array}
\right..   \label{eq:integrals}
\end{equation}

We solve integral equations (\ref{eq:int})  
via the Fourier transform. 
Let us express the Fourier transform of function 
$\varphi^{(\beta)}_{\alpha}(\mu,\xi)$ by  
\be 
\widehat{\varphi}^{(\beta)}_{\alpha}(\omega, \xi) 
=\int^{\infty}_{-\infty}
e^{i\mu \omega}\varphi_{\alpha}^{(\beta)}(\mu, \xi) d\mu, \quad  
\mbox{for} \, \,  \alpha, \beta= 1, 2, \ldots, 2s \,  . 
\ee
We denote by $\widehat{K}_n(\omega)$ the Fourier transform of 
kernel $K_n(\lambda)$. 
We define matrix ${\cal M}^{(2s)}_{\hat \varphi}$ by 
\be 
\left( {\cal M}^{(2s)}_{\hat \varphi} \right)_{\alpha \beta} = 
\widehat{\varphi}^{(\beta)}_{\alpha}(\omega, \xi) \quad 
\mbox{for} \quad \alpha, \beta= 1, 2, \ldots, 2s \,  . 
\ee
We introduce a $2s$-by-$2s$ matrix ${\cal M}^{(2s)}_{K_2}$. We define  
 matrix element $(j, k)$ for $j, k = 1, 2, \ldots, 2sm$, by  
\be 
\left( {\cal M}^{(2s)}_{K_2} \right)_{j, k} = \left\{ 
 \begin{array}{cc} 
\displaystyle{1+ \int^{\infty}_{-\infty} e^{i\mu \omega} K_2(\mu) \, d\mu}
 & \mbox{for} \, \, j=k \, , \\ 
& \\
\displaystyle{\int^{\infty}_{-\infty} 
e^{i\mu \omega} K_2(\mu + (k-j) \eta + i 0) \, d\mu} & 
\mbox{for} \, \, j < k \, , \\
& \\ 
\displaystyle{\int^{\infty}_{-\infty} 
e^{i\mu \omega} K_2(\mu - (j-k) \eta - i 0) \, d\mu} & 
\mbox{for} \, \, j > k \, . \\
\end{array} 
\right. 
\ee
When $0 < \zeta < \pi/2s$,   
we calculate the matrix elements of ${\cal M}^{(2s)}_{K_2}$ as follows.  
\be 
\left( {\cal M}^{(2s)}_{K_2} \right)_{j,k} 
= \delta_{j,k} (1 + {\widehat K}_2(\omega)) 
+ (1 - \delta_{j,k})  e^{(k-j)\zeta \omega}  
\left({\widehat K}_2(\omega)- e^{{\rm sgn}(j-k) \, \zeta \omega} 
\right) \quad \mbox{for} \, \,  
 j, k = 1, 2, \ldots, 2s. 
\ee
Here we define ${\rm sgn}(j-k)$ by the following:   
${\rm sgn}(j-k)= -1$ for $j-k < 0$, and ${\rm sgn}(j-k)=+1$ 
for $j - k > 0$. Here, ${\widehat K}_2(\omega)$, is given by   
\be
{\widehat K}_2(\omega) = 
\int_{-\infty}^{\infty} e^{i \omega \mu} K_2(\mu) d \mu 
= {\frac {\sinh(\displaystyle{\frac {\pi} 2} - \zeta) \omega} 
       {\sinh( \displaystyle{\frac {\pi  \omega} 2})}} \, . 
\ee
Similarly, we define a $2s$-by-$2s$ matrix ${\cal M}^{(2s)}_{K_1}$ by 
\be 
\left( {\cal M}^{(2s)}_{K_1} \right)_{j, k} 
= \int_{-\infty}^{\infty} e^{i \omega \mu } 
K_1(\mu- \xi_b + (k-j) \eta) d\mu   
\quad \mbox{for} \, \, j,k = 1, 2, \ldots, 2s. 
\ee
When $0 < \zeta < \pi/2s$, we can show 
\be 
\left( {\cal M}^{(2s)}_{K_1} \right)_{j, k} = 
e^{i \xi_b \omega} \left\{ \delta_{j, k} {\widehat K}_1(\omega) 
+ ( 1 - \delta_{j, k} ) \, e^{(k-j)\zeta \omega} 
\left({\widehat K}_1(\omega) - e^{{\rm sgn}(j-k)\zeta \omega/2} \right) 
\right\}
\ee
for $j, k=1, 2, \ldots, 2s$. Here 
${\widehat K}_1(\omega)$ is given by 
\be
{\widehat K}_1(\omega) = 
\int_{-\infty}^{\infty} e^{i \omega \mu} K_1(\mu) d \mu 
= {\frac {\sinh \left(
\displaystyle{{\frac {\pi} 2} - {\frac {\zeta} 2}} \right) \omega} 
       {\sinh( \displaystyle{\frac {\pi  \omega} 2})}} \, . 
\ee

Taking the Fourier transform of integral equations (\ref{eq:int}) 
we have the following matrix equation. 
\be 
{\cal M}^{(2s)}_{K_2} \, 
{\cal M}^{(2s)}_{\hat \varphi} = {\cal M}^{(2s)}_{K_1} \, . 
\label{eq:MatrixEq}
\ee
For the spin-1 case, from (\ref{eq:integrals}) we have 
\bea
& & \left( 
\begin{array}{cc} 
1 + \widehat{K}_2(\omega) &  e^{\zeta \omega} \widehat{K}_2(\omega) - 1 \\ 
e^{-\zeta \omega} \widehat{K}_2(\omega) -1 & 1 + \widehat{K}_2(\omega)  \\
\end{array} \right) 
\left( 
\begin{array}{cc} 
\widehat{\varphi}^{(1)}_{1}(\omega) & 
\widehat{\varphi}^{(2)}_{1}(\omega) \\ 
\widehat{\varphi}^{(1)}_{2}(\omega) & 
\widehat{\varphi}^{(2)}_{2}(\omega) \\ 
\end{array} \right)
\non \\ 
& & \quad = e^{i \xi_b \omega} \, 
\left( 
\begin{array}{cc} 
\widehat{K}_1(\omega) & 
e^{\zeta \omega} \widehat{K}_1(\omega) - e^{\zeta \omega/2} \\ 
e^{- \zeta \omega} \widehat{K}_1(\omega) - e^{- \zeta \omega/2} 
& \widehat{K}_1(\omega) \\  
\end{array} \right)  \, . 
\label{eq:Fourier-int}
\eea
It is easy to show that matrix $M^{(2)}({\widehat{\varphi}})$ 
is given by the following:   
\be 
\left( 
\begin{array}{cc} 
\widehat{\varphi}^{(1)}_{1}(\omega) & 
\widehat{\varphi}^{(2)}_{1}(\omega) \\ 
\widehat{\varphi}^{(1)}_{2}(\omega) & 
\widehat{\varphi}^{(2)}_{2}(\omega) \\ 
\end{array} \right)
= 
\left( 
\begin{array}{cc} 
\displaystyle{\frac {e^{i \xi_b \omega}} {2 \cosh(\zeta \omega/2)} } & 0 \\ 
0  &  \displaystyle{\frac {e^{i \xi_b \omega} } {2 \cosh(\zeta \omega/2)} } \\
\end{array} \right) \, . 
\ee

We calculate the determinant of ${\cal M}^{(2)}_{K_2}$ (the spin-1 case) 
as follows.  
\bea 
{\rm det} 
\left( 
\begin{array}{cc} 
1 + \widehat{K}_2(\omega) &  e^{\zeta \omega} \widehat{K}_2(\omega) - 1 \\ 
e^{-\zeta \omega} \widehat{K}_2(\omega) -1 & 1 + \widehat{K}_2(\omega)  
\end{array} \right)
& = &
{\rm det} 
\left( 
\begin{array}{cc} 
1 + \widehat{K}_2(\omega) &  e^{\zeta \omega} \widehat{K}_2(\omega) - 1 \\ 
 - (1 + e^{-\zeta \omega}) & 1 +  e^{-\zeta \omega}  
\end{array} 
\right) 
\non \\ 
& = &
{\rm det} 
\left( 
\begin{array}{cc} 
 \widehat{K}_2(\omega) (1 +  e^{\zeta \omega}) 
&  e^{\zeta \omega} \widehat{K}_2(\omega) - 1 \\ 
 0 & 1 +  e^{-\zeta \omega}  
\end{array} 
\right) \non \\ 
& = &  \widehat{K}_2(\omega) (1 +  e^{-\zeta \omega})
(1 +  e^{\zeta \omega}) \, . 
\eea
Here,  we first subtract the 2nd row by the 1st row multiplied by 
$e^{-\zeta \omega}$. We next add the 2nd column to the 1st column.   
Finally, the determinant is factorized and we have the result.   

By the same method we can calculate the determinant of 
${\cal M}^{(2s)}_{K_2}$ for 
the spin-$s$ case. We have    
\be 
{\rm det} {\cal M}^{(2s)}_{K_2} = 
 \left( 2 \cosh(\zeta \omega/2) \right)^{2s} \, 
\widehat{K}_{2s}(\omega) \, . 
\ee
The determinant is  nonzero generically, 
 and hence the solution to matrix equation (\ref{eq:MatrixEq}) is unique. 
Therefore, we obtain 
the solution of integral equation (\ref{eq:int}).

%
%
\setcounter{equation}{0} 
 \renewcommand{\theequation}{5.\arabic{equation}}
\section{The EFP of the spin-$s$ XXZ spin chain near AF point}

\subsection{Multiple-integral representations of the spin-$s$ EFP }

Let us derive multiple-integral representations 
for the emptiness formation probability of the 
spin-$s$ XXZ spin chain. 
We shall take the large $N_s$ limit of 
the EFP (\ref{eq:EFP-finite}) for a finite-size system,  
and we replace rapidity $\lambda_{c_j}$ with 
complex  variable $\lambda_j$ for $j=1, 2, \ldots, 2sm$, as follows. 
For a given rapidity of $2s$-string, 
 $\lambda_{A}= \mu_a - (\alpha -1/2) \eta + \epsilon_A$,   
we define its regular part $\lambda_A^{'}$ by 
$\lambda_A^{'} = \mu_a - (\alpha -1/2) \eta$.  
In the large-$N_s$ limit, we first replace 
$\lambda_{c_k}=\lambda_{c_k}^{'} + \epsilon_{c_k}$ 
by $\lambda_k^{'} + \epsilon_{c_k}$ 
where $\lambda_k^{'}$ are complex integral variables  
corresponding to complete strings such as $\lambda_k^{'}= 
\mu_k - (\beta -1/2) \eta$ for some integer $\beta$ 
with $1 \le \beta \le 2s$ where $\eta= i \zeta$ with $0 < \zeta < \pi$ 
 and $\mu_k$ is real.    
We express $\lambda_k^{'}$ and $\epsilon_{c_k}$ 
simply by $\lambda_k$ and $\epsilon_{k}$, 
respectively, and then we obtain multiple-integral representations.

Applying (\ref{eq:detPhi=detS}) 
we derive the emptiness formation probability for arbitrary spin-$s$ 
in the thermodynamic limit $N_s \rightarrow \infty$, as follows. 
\begin{equation}
\begin{split}
\tau^{(2s)}_{\infty}(m; \{\xi_p\})= 
& \frac{1}{\prod_{1\leq j<r\leq 2s} (\sinh(r-j)\eta)^m}  
\times\frac{1}{\prod_{1\leq k<l\leq m}\prod^{2s}_{j=1}\prod^{2s}_{r=1}
\sinh(\xi_k-\xi_l+(r-j)\eta)}     \\
&\times\prod^{2sm}_{l=1}\left(\sum^{2s}_{k=1}
\int^{\infty+(-k+\frac{1}{2})\eta}_{-\infty+(-k+\frac{1}{2})\eta}
d\lambda_l \right)
H^{(2s)}(\lambda_1,\cdots,\lambda_{2sm}) \, 
\textstyle{\det} S(\lambda_1, \ldots, \lambda_{2sm})
\label{eq:EFP-integral1}
\end{split}
\end{equation}
where $H^{(2s)}((\lambda_l)_{2sm})$ is given by 
\begin{equation}
\begin{split}
H^{(2s)}((\lambda_l)_{2sm})=&\frac{1}{\prod_{1 \leq l< k \leq 
2sm}\sinh(\lambda_{k} -\lambda_{l}+\eta + \epsilon_{k, l})}
\\
&\times
\prod^{2sm}_{l=1}
\prod^{m}_{k=1}
\prod^{2s-1}_{p=1} 
\sinh(\lambda_{l}-\xi_k+(2s-p)\eta)
\\
&\times
\prod^{m}_{l=1}\prod^{2s}_{r_l=1}
\left( \prod^{l-1}_{k=1}\sinh(\lambda_{2s(l-1)+r_l} - \xi_k+2s\eta)
\prod^m_{k=l+1} \sinh(\lambda_{2s(l-1)+r_l}-\xi_k)
\right)
\end{split}
\end{equation}
and  matrix elements of $S(\lambda_1, \ldots, \lambda_{2sm})$ are given by 
\begin{equation}
    S_{j, \, 2s(l-1)+k}=
    \left\{
    \begin{array}{ll}
        \rho(\lambda_j-\xi_l+(k- {\frac {1}{2}})\eta)    
       & \mbox{if} \quad \lambda_j-\mu_{j}=(\frac{1}{2}-k)\eta \\
        0     & otherwise.
    \end{array}
    \right. \label{eq:matrix-elements-S}
\end{equation}
Here $\mu_{j}$ denotes 
the center of the $2s$-string in which   
$\lambda_{j}$ is the $k$th rapidity.  
Explicitly, we have $\lambda_j= \mu_{j} - (k -1/2) \eta$. 
In the denominator, we have set  $\epsilon_{k, l}$ associated 
with $\lambda_k$ and $\lambda_l$ as follows.  
\be 
\epsilon_{k, \, l} = 
\left\{ 
\begin{array}{cc} 
i \epsilon & \mbox{for} \quad {\cal I}m(\lambda_k - \lambda_l) >0 \, ,  \\ 
- i \epsilon & \mbox{for} \quad {\cal I}m(\lambda_k - \lambda_l) < 0. 
\end{array} 
\right. 
\label{eq:epsilon}
\ee

In the homogeneous case we have $\epsilon_p=0$  for $p=1, 2, \ldots, N_s$. 
We have thus defined inhomogeneous parameters $\xi_p$.  
We recall that in the homogeneous case, the spin-$s$ Hamiltonian is derived 
from the logarithmic derivative of the spin-$s$ transfer matrix. 

Here we should remark that an expression of the matrix elements 
of $S$ similar to (\ref{eq:matrix-elements-S}) 
has been given in eq. (6.14) of Ref. \cite{Kitanine2001} 
for the correlation functions of the integrable spin-$s$ XXX spin chain.

\subsection{Symmetric expression of the spin-$s$ EFP}

We shall express the spin-$s$ EFP (\ref{eq:EFP-integral1}) 
in a simpler way, making use of permutations 
of $2sm$ integers, $1, 2, \cdots, 2sm$, 
and the formula of the Cauchy determinant. 

Let us take  a set of integers $a_{(j, k)}=a_{2s(j-1)+k}$ 
satisfying $1 \le a_{(j, k)} \le N_s/2 $ 
for $j=1, 2, \ldots, m$ and $k=1, 2, \ldots, 2s$.
Here we remark that indices $a_{(j,k)}$ correspond 
to  the string centers $\mu_{(j,k)}= \mu_{2s(j-1)+k}$.  
In order to reformulate the sum  
over integers, $c_1, \ldots, c_{2sm}$, in eq. (\ref{eq:EFP-finite}) 
in terms of indices $a_{(j, k)}$,  
let us introduce ${\hat c}_1, \ldots, {\hat c}_{2sm}$ by  
\be 
{\hat c}_{2s (j-1)+k} = 2s (a_{(j, k)} -1) + k \, , \quad \mbox{for} 
\quad j=1, 2, \ldots, m \quad \mbox{and} \, k=1, 2, \ldots, 2s \, .   
\ee
We also define $\beta(z)$ by $\beta(z) = z -2 s [(z-1)/2s]$. 
Then, ${\hat c}_j$ are expressed as follows. 
\be 
{\hat c}_j = 2s (a_j -1) + \beta(j) \, , \quad \mbox{for} \, \, 
j=1 , 2, \ldots, 2sm .      
\ee
We decompose the sum over $c_j$ into $2s$ sums over $a_j$ as follows.   
\be 
\sum_{c_j=1}^{M} g(c_j) = \sum_{k=1}^{2s} 
\sum_{a_j=1}^{N_s/2} g(2s(a_j-1) + k)
\, . 
\ee

Let us consider such a function $f(c_1, c_2, \ldots, c_{2sm})$ 
of sequence of integers $(c_j)_{2sm}$ 
that vanishes unless $c_j$'s are distinct.  
We also assume that $f(c_1, c_2,  \ldots, c_{2sm})$ vanishes 
unless the number of $c_j$'s satisfying  
$\beta(c_j)=\alpha$ is given by $m$ for each integer $\alpha$ 
satisfying $1 \le \alpha \le 2s$. Here we recall that the two properties 
are in common with the summand of (\ref{eq:EFP-finite}), in particular, 
with ${\rm det}(\phi^{(2s; \, m)}((c_j)_{2sm}; \{ \xi_p \})$.    
Then, we have 
\bea 
& & \sum_{c_1=1}^{M} \sum_{c_2=1}^{M}
\cdots \sum_{c_{2sm}=1}^{M} f(c_1, \cdots, c_{2sm}) 
 = 
{\frac 1 {(m!)^{2s}}}
\sum_{a_1=1}^{{N_s}/2} 
\sum_{a_2=1}^{{N_s}/2} 
\cdots 
\sum_{a_{2sm}=1}^{{N_s}/2}
\sum_{P \in {\cal S}_{2sm}} 
f({\hat c}_{P1}, \cdots, {\hat c}_{P(2sm)} ) \non \\ 
&  & \qquad = 
\sum_{a_1=1}^{{N_s}/2} 
\sum_{a_2=1}^{{N_s}/2} 
\cdots 
\sum_{a_{2sm}=1}^{{N_s}/2}
\sum_{\pi \in {\cal S}_{2sm}/(S_m)^{2s} } 
f({\hat c}_{\pi 1}, \cdots, {\hat c}_{\pi(2sm)} ) 
\, .  
 \label{eq:csum-to-asum}
\eea
Here an element $\pi$ of ${\cal S}_{2sm}/({\cal S}_m)^{2s}$ 
gives a permutation of integers $1, 2, \ldots, 2sm$,   
where $\pi j$'s such that $\pi j \equiv k$ (mod $2s$) 
are put in increasing order in the sequence 
($\pi 1, \pi 2, \ldots, \pi (2sm))$ for $k=0, 1, \ldots, 2s-1$. 

Reformulating the sum over $c_j$ s in (\ref{eq:EFP-finite}) 
in terms of $a_{j}$'s,  in the large $N_s$ limit we have 
\bea 
\tau^{(2s)}_{\infty}(m) & = & 
\frac{1}{\prod_{1\leq \alpha < \beta \leq 2s}(\sinh(\beta -\alpha)\eta)^m}  
\times\frac{1}{\prod_{1\leq k<l\leq m} 
\prod^{2s}_{\alpha=1}\prod^{2s}_{\beta=1}
\sinh(\xi_k-\xi_l+(\alpha-\beta)\eta)}   \non \\ 
& & \times  \prod_{k=1}^{2s} \prod_{j=1}^{m} 
\int_{-\infty}^{\infty} d\mu_{(j, k)} 
\, \sum_{P \in {\cal S}_{2sm}}  {\frac 1 {(m!)^{2s}}} 
{\rm det} 
S(\lambda_{P1}, \ldots, \lambda_{P(2sm)})  \, 
 H^{(2s)} (\lambda_{P1}, \ldots, \lambda_{P(2sm)}) 
\non \\ 
& = & 
\frac{1}{\prod_{1\leq \alpha < \beta \leq 2s}(\sinh(\beta -\alpha)\eta)^m}  
\times\frac{1}{\prod_{1\leq k<l\leq m} 
\prod^{2s}_{\alpha=1}\prod^{2s}_{\beta=1}
\sinh(\xi_k-\xi_l+(\alpha-\beta)\eta)}   \non \\ 
& \times & \prod_{j=1}^{2sm} 
\int_{-\infty}^{\infty} d\mu_{j} 
\, \sum_{\pi \in {\cal S}_{2sm}/({\cal S}_m)^{2s} } 
{\rm det} 
S(\lambda_{\pi 1}, \ldots, \lambda_{\pi (2sm)})  \, 
 H^{(2s)} (\lambda_{\pi 1}, \ldots, \lambda_{\pi (2sm)}) \, ,  
\label{eq:EFP-integral2}
\eea
where symbols $\lambda_j$ denote the following 
\be 
\lambda_{j} = \mu_j - (\beta(j) - {\frac 1 2}) \eta \, 
\, \quad \mbox{for} \, \, j= 1, 2, \ldots, 2sm.  
\label{eq:def-lambda_j}
\ee

We calculate ${\rm det} S$ 
applying the Cauchy determinant formula 
\begin{equation}
 \det\left(\frac{1}{\sinh(\lambda_a-\xi_k)}\right)
  =\frac{\prod^m_{k<l}\sinh(\xi_k-\xi_l)\prod^m_{a>b} 
                     \sinh(\lambda_a-\lambda_b)}
{\prod^m_{a=1}\prod^m_{k=1}\sinh(\lambda_a-\xi_k)} \, ,   
\end{equation}
and we obtain the symmetric expression of the spin-$s$ EFP as follows. 
\begin{equation}
\begin{split}
\tau^{(2s)}_{\infty}(m; \{\xi_p\}) & = 
\frac{1}{\prod_{1 \leq \alpha < \beta \leq 2s}
\sinh^{m}(\beta-\alpha )\eta} \,  
\prod_{1\leq k < l \leq m}
\frac{\sinh^{2s}(\pi(\xi_k-\xi_l)/\zeta)}
{\prod^{2s}_{j=1}\prod^{2s}_{r=1}\sinh(\xi_k-\xi_l+(r-j)\eta)}   
 \\
&\times \, {\frac {i^{2sm^2}} { (2 i \zeta)^{2sm} }} \, 
%
%
\left( \prod^{2sm}_{j=1} \int^{\infty}_{-\infty} d \mu_{j} \right) \, 
\prod^{2s}_{\gamma=1} 
\prod_{1 \le b < a \le m}
\sinh(\pi(\mu_{2s(a-1)+\gamma}-\mu_{2s(b-1)+\gamma})/\zeta)
\\
& 
\times \left( \prod^{2sm}_{j=1} 
{\frac {\prod^{m}_{b=1} \prod^{2s-1}_{\beta=1}
\sinh(\lambda_{j}-\xi_b+ \beta \eta)}
{\prod_{b=1}^{m} \cosh(\pi(\mu_{j}-\xi_b)/\zeta)}} \right) 
%
\,  
\sum_{\sigma \in {\cal S}_{2sm}/({\cal S}_m)^{2s} } 
({\rm sgn} \, \sigma) \, 
\\
&
\times \frac{\prod^{m}_{l=1}\prod^{2s}_{r_l=1}
\left(\prod^{l-1}_{k=1}\sinh(\lambda_{\sigma(2s(l-1)+r_l)}-\xi_k+2s\eta)
\prod^m_{k=l+1}\sinh(\lambda_{\sigma(2s(l-1)+r_l)}-\xi_k)\right)}
{\prod_{1\leq l<k\leq 2sm} 
\sinh(\lambda_{\sigma(k)}-\lambda_{\sigma(l)}+\eta 
+ \epsilon_{\sigma(k), \sigma(l)} )} \, . 
\label{eq:EFP2}
\end{split}
\end{equation}
Here  (${\rm sgn} \, \sigma$) 
denotes the sign of permutation 
$\sigma \in {\cal S}_{2sm}/({\cal S}_m)^{2s}$. 

%
%
\subsection{The spin-$s$ EFP for the homogeneous chain}

Sending $\xi_b$ to zero for $b=1, 2, \ldots, m$, we derive   
the spin-$s$ EFP in the homogeneous limit.  
Here we remark that the spin-$s$ Hamiltonian is derived from the logarithmic 
derivative of the spin-$s$ transfer matrix 
$t^{(2s, \, 2s)}(\lambda; \{\xi_b \}_{N_s})$ 
in the homogeneous case where $\xi_b=0$ for $b=1, 2, \ldots, N_s$. 

Sending $\xi_b$ to zero for $b=1, 2, \ldots, m$, 
we have the following: 
\begin{equation}
\begin{split}
& \lim_{\xi_1 \rightarrow 0}  \lim_{\xi_2 \rightarrow 0}  \cdots 
\lim_{\xi_{m} \rightarrow 0}  \,  
\left( \tau^{(2s)}_{\infty}
(m; \{ \xi_p \}_{m}) \right) 
\\ 
& = 
{\frac {\left( \pi/\zeta \right)^{ s m (m-1)} } 
       {\prod_{1 \leq \alpha < \beta \leq 2s} 
        \sinh^{m^2}((\beta - \alpha) \eta)}}  
\\
& \times \, {\frac {i^{2sm^2}} { (2 i \zeta)^{2sm} }} \, \, 
%
%
\prod^{2sm}_{j=1} \int^{\infty}_{-\infty} d \mu_{j}  
\, 
\prod^{2s}_{\beta=1} 
\prod_{1 \le b < a \le m}
\sinh(\pi(\mu_{2s(a-1)+\beta}-\mu_{2s(b-1)+\beta})/\zeta)
\\
& 
\times \left( \prod^{2sm}_{j=1} 
{\frac {\prod^{2s-1}_{\beta=1} \sinh^{m}(\lambda_{j} + \beta \eta)}
{ \cosh^{m}(\pi \mu_{j}/\zeta)}} \right)
\times  
\sum_{\sigma \in {\cal S}_{2sm}/({\cal S}_m)^{2s}} 
({\rm sgn} \, \sigma) \, 
\\
&
\times {\frac {\prod^{m}_{l=1} \prod^{2s}_{r_l=1}
\left(
%
%
\sinh^{l-1}(\lambda_{\sigma(2s(l-1)+r_l)} + 2s\eta)
%
%
\sinh^{m-l}(\lambda_{\sigma(2s(l-1)+r_l)}) \right)}
{\prod_{1 \leq l < k \leq 2sm} 
\sinh(\lambda_{\sigma(k)}-\lambda_{\sigma(l)}+\eta 
+ \epsilon_{\sigma(k), \sigma(l)} )}} \, . 
\label{eq:EFP3}
\end{split}
\end{equation}
Here we recall definition (\ref{eq:def-lambda_j}) of $\lambda_j$. 
 
Let us discuss that  expression (\ref{eq:EFP3}) gives 
 the spin-$s$ EFP for the homogeneous chain. 
First, we remark that  
$\tau^{(2s)}_{N_s}(m; \{\xi_p \}_{m})$ does not depend on 
$\xi_p$ with $p > m$. Hence we may consider that 
 inhomogeneous parameters $\xi_p$ with $p > m$ are all set to be zero,   
after computing the EFP: $\tau^{(2s)}_{N_s}(m; \{\xi_p \}_{m})$.

We now show that the order of the homogeneous limit: $\xi_p \rightarrow 0$ 
and the thermodynamic limit $N_s \rightarrow \infty$ can be reversed.   
We can show the following relation: 
\be 
 \prod_{p=1}^{m} \lim_{\xi_p \rightarrow 0}  \, 
 \left(   
\lim_{N_s \rightarrow \infty}  \,  \left( 
\tau^{(2s)}_{N_s}(m; \{\xi_p\}_{m}) \right) \right)
=
\lim_{N_s \rightarrow \infty}
\left( 
\prod_{p=1}^{m} \lim_{\xi_p \rightarrow 0} 
  \,  \left( 
\tau^{(2s)}_{N_s}(m; \{\xi_p\}_{m}) \right) \right) \, . 
\ee
In fact, when  $N_s$ is very large, it follows from 
(\ref{eq:detPhi=detS}) that we have 
\be 
\tau^{(2s)}_{\infty}(m; \{ \xi_p \}_m) 
= \tau^{(2s)}_{N_s}(m; \{ \xi_p \}_m) + O(1/N_s).  
\ee
Furthermore, we can explicitly show that  
$\tau^{(2s)}_{N_s}(m; \{\xi_p \}_m)$ is continuous  
with respect to $\xi_p$ at $\xi_p=0$ for $p=1, 2, \ldots, m$. 
We first reformulate the sum over $c_j$ in (\ref{eq:EFP-finite}) 
into the sum over $a_{(j,k)}$ by relation (\ref{eq:csum-to-asum}).  
\bea 
& & \sum_{c_1=1}^{M} \cdots \sum_{c_{2sm}=1}^{M} 
{\rm det}S(({c}_{j})_{2sm}) H^{(2s)}((\lambda_{c_j})_{2sm}) 
\non \\ 
& = & \prod_{j=1}^{m} \prod_{k=1}^{2s} \left( {\frac 1 {m!}} 
\sum_{a_{(j,k)}=1}^{N_s/2} \right) \sum_{P \in {\cal S}_{2sm}} 
{\rm det}S(({\hat c}_{Pj})_{2sm}) 
H^{(2s)}((\lambda_{{\hat c}_j})_{2sm})  \non \\ 
& = & \prod_{j=1}^{m} \prod_{k=1}^{2s} \left( {\frac 1 {m!}} 
\sum_{a_{(j,k)}=1}^{N_s/2} \right)  
{\rm det}S(({\hat c}_{j})_{2sm}) \sum_{P \in {\cal S}_{2sm}} 
({\rm sgn}P) \,  
H^{(2s)}((\lambda_{{\hat c}_j})_{2sm}) \, . 
\eea
We then apply the Cauchy determinant formula to 
evaluate ${\rm det}S({\hat c}_{j})$ as follows. 
\bea
& & {\rm det}S(({\hat c}_{j})_{2sm}) 
 = {\rm det}S((\lambda_{{\widehat c}_j })_{2sm})    
 = \prod_{\alpha=1}^{2s} 
{\rm det} S^{(\alpha)}((\lambda_{2s(a_{(j, \alpha)}-1) + \alpha})_{m} ) 
\non \\ 
& = & 
{\frac 
{\left(  \prod_{j < k}  \sinh \pi(\xi_j- \xi_{k})/\zeta 
\right)^{2s}
\cdot 
\prod_{\alpha=1}^{2s} 
\prod_{j<k} \sinh \left\{ \pi(\mu_{2s(a_{(j, \alpha)}-1) + \alpha} 
- \mu_{2s(a_{(k, \alpha)}-1) + \alpha} )/\zeta \right\} }  
{(2 i \zeta) ^{2sm} 
\prod_{j=1}^{2sm} \prod_{b=1}^{m} 
\sinh \pi(\mu_j - \xi_b - \eta/2)/\zeta }} \, . 
\non \\ 
\label{eq:detS}
\eea
Making use of (\ref{eq:detS}) we show that such factors in 
the denominator of (\ref{eq:EFP-finite}) 
that vanish in the limit of sending $\xi_p$ 
to zero are canceled by the factors in the numerator of 
(\ref{eq:detS}). 
We thus have shown that the EFP for the finite system,   
$\tau^{(2s)}_{N_s}(m; \{\xi_p \}_m)$, 
is continuous with respect to $\xi_p$ at $\xi_p=0$.    

Therefore, expression (\ref{eq:EFP3}) gives 
 the spin-$s$ EFP for the homogeneous chain. 
That is, we have the following equality: 
\be 
\lim_{\xi_1 \rightarrow 0}  \lim_{\xi_2 \rightarrow 0}  
\cdots 
 \lim_{\xi_{m} \rightarrow 0}  \,  
\left( \tau^{(2s)}_{\infty}(m; \{ \xi_p \}_{m}) \right)   
= \lim_{N_s \rightarrow \infty} 
 \tau^{(2s)}_{N_s}(m; \{ \xi_p=0 \}_{N_s}) \, .  
\ee

\subsection{The spin-1 EFP with $m=1$}

Let us calculate $\tau^{(2s)}(m)$ for $s=1$ and $m=1$.   
From formula (\ref{eq:EFP-integral1}) we have 
\bea 
& &  \tau^{(2)}(1) 
 = {\frac 1 {\sinh \eta}} 
\left(\int^{\infty - \eta/2}_{-\infty - \eta/2} d \lambda_1
+ \int^{\infty - 3\eta/2}_{-\infty - 3\eta/2} d \lambda_1 \right) 
\left(\int^{\infty - \eta/2}_{-\infty - \eta/2} d \lambda_2 
+ \int^{\infty - 3\eta/2}_{-\infty - 3\eta/2} d \lambda_2 \right) 
\non \\ 
& & \qquad \times 
H^{(2)}(\lambda_1, \lambda_2) \, {\rm det}S(\lambda_1, \lambda_2)  
\non \\ 
& = &  - {\frac 1 {i \sin \zeta}} \, {\frac 1 {4 \zeta^2}} 
\int^{\infty}_{-\infty} d \mu_1 \int^{\infty}_{-\infty} d \mu_2 
\frac {\sinh(\mu_1- \xi_1 + \eta/2) \sinh(\mu_2- \xi_1 - \eta/2)}   
{\sinh(\mu_1 - \mu_2 + i \epsilon) \, 
\cosh\left( \pi( \mu_1- \xi_1)/\zeta\right)  
\cosh \left( \pi( \mu_2- \xi_1)/\zeta \right) } \non 
\\ 
& - &{\frac 1 {i \sin \zeta}} \, \frac {(-1)} {4 \zeta^2} 
\int^{\infty}_{-\infty} d \mu_1 \int^{\infty}_{-\infty} d \mu_2 
\frac {\sinh(\mu_1- \xi_1 - \eta/2) \sinh(\mu_2- \xi_1 + \eta/2)}   
{\sinh(\mu_1 - \mu_2 - 2 \eta) \, 
\cosh \left( \pi( \mu_1- \xi_1)/\zeta \right) 
\cosh \left( \pi( \mu_2- \xi_1)/\zeta \right) } \, .\non \\  
\label{eq:EFPs=1}
\eea
Here we note that ${\rm det} S(\lambda_1, \lambda_2) = 0$ for 
 $\lambda_1 = \mu_1 - \eta/2$ and $\lambda_2 = \mu_2 - 3 \eta/2$, 
or for  $\lambda_1 = \mu_1 - 3 \eta/2$ and $\lambda_2 = \mu_2 - \eta/2$. 
Showing the following relations of integrals 
\bea 
 \int_{-\infty}^{\infty} d \mu {\frac 1 {\cosh\left(\pi \mu/\zeta \right)}} 
{\frac {\sinh(\mu- \eta/2)} {\sinh(\lambda- \mu + i \epsilon)}}  
& = & - \int_{-\infty}^{\infty} d \mu 
{\frac 1 {\cosh\left(\pi \mu/\zeta \right)}} 
{\frac {\sinh(\mu + \eta/2)} {\sinh(\lambda- \mu - \eta)}} 
- 2 \pi i {\frac {\sinh(\lambda-\eta/2)} {\cosh(\pi \lambda/\zeta)}} \, ,   
\non \\ 
\int_{-\infty}^{\infty} d \lambda 
{\frac 1 {\cosh\left(\pi \lambda/\zeta \right)}} 
{\frac {\sinh(\lambda + \eta/2)} {\sinh(\lambda- \mu - \eta)}}  
& = & - \int_{-\infty}^{\infty} d \lambda 
{\frac 1 {\cosh\left(\pi \lambda/\zeta \right)}} 
{\frac {\sinh(\lambda - \eta/2)} {\sinh(\lambda- \mu - 2 \eta)}} \, ,  
\eea
we thus have  
\be 
\tau^{(2)}(1)= {\frac {\pi} 4}  
{\frac 1 {\zeta \sin \zeta}}
\left(
\int_{-\infty}^{\infty} d x \, 
{\frac {\cosh 2 \zeta x - \cosh \eta} {\cosh^2 \pi x}}
 \right) \, . 
\ee
Evaluating the integral we obtain the spin-1 EFP with $m=1$ as follows. 
\be 
\tau^{(2)}(1) 
= \frac {\zeta - \sin \zeta \cos \zeta} {2 \zeta \sin^2 \zeta}  \, . 
\ee
Let us denote by $\la E_1^{a, \, b} \ra$ the ground-state 
expectation value of operator $E_1^{a, \, b}$. For the spin-1 case, 
we have  $\widetilde{E}_1^{0, \, 0 \, (2 \, +)} 
+\widetilde{E}_1^{1, \, 1 \, (2 \, +)} 
+ \widetilde{E}_1^{2, \, 2 \, (2 \, +)} = \widetilde{P}_1^{(2)}$, 
and hence we have 
\be 
\la \widetilde{ E}_1^{0, \, 0 \, (2 \, +)} \ra 
+ \la \widetilde{E}_1^{1, \, 1 \, (2 \, +)} \ra 
+ \la \widetilde{E}_1^{2, \, 2 \, (2 \, +)} \ra 
= 1 \, .  
\ee
Due to the uniaxial symmetry we have 
$\la \widetilde{ E}_1^{0, \, 0 \, (2 \, +)} \ra 
= \la \widetilde{ E}_1^{2, \, 2 \, (2 \, +)} \ra$. 
Thus, we obtain   
\be  
\la \widetilde{ E}_1^{1, \, 1 \, (2 \, +)} \ra =  
 {\frac {\cos \zeta \, (\sin \zeta - \zeta \cos \zeta)} 
{\zeta \sin^2 \zeta}}  
\, . \label{eq:EPofE11}
\ee

In the XXX limit, we have  
\be 
\lim_{\zeta \rightarrow 0} 
\frac {\zeta - \sin \zeta \cos \zeta} {2 \zeta \sin^2 \zeta} 
= {\frac 1 3} \, . 
\ee
The limiting value 1/3 coincides with the spin-1 XXX result 
obtained  by Kitanine \cite{Kitanine2001}. 
As pointed out in Ref. \cite{Kitanine2001}, $\la E_1^{22} \ra= 
\la E_1^{11} \ra = \la E_1^{00} \ra = 1/3$ for the XXX case since it has 
the rotational symmetry.

In the symmetric expression of the spin-$1$ EFP with $m=1$, 
putting $\lambda_1= \mu_1 - \eta/2$ and $\lambda_2= \mu_2 - 3\eta/2$ 
in (\ref{eq:EFP3}), 
we directly obtain the following:  
\bea    
\tau^{(2s)}(1) & = & 
{\frac 1 {i \sin \zeta}} \, {\frac 1 {4 \zeta^2}} 
\int^{\infty}_{-\infty} d \mu_1 \int^{\infty}_{-\infty} d \mu_2 
\frac {\sinh(\mu_1 + \eta/2) \sinh(\mu_2- \eta/2)} 
{\cosh\left( \pi \mu_1/\zeta \right)  
\cosh \left( \pi \mu_2/\zeta \right)} \non \\ 
& & \quad 
\times \left( {\frac 1 
{\sinh(\mu_2 - \mu_1 - i \epsilon)}} 
- {\frac 1 {\sinh(\mu_1 - \mu_2 + 2 \eta)}} \right) \, . 
\label{eq:syms=1}
\eea 
In the second line of (\ref{eq:syms=1}) the first term 
corresponds to the first term of (\ref{eq:EFPs=1}), 
while the second term to 
the second term of (\ref{eq:EFPs=1}) with 
$\mu_1$ and $\mu_2$ being exchanged.

%
%
%
%
 \setcounter{equation}{0} 
 \renewcommand{\theequation}{6.\arabic{equation}}
%

\section{Spin-$s$ XXZ correlation functions near AF point}

\subsection{Finite-size correlation functions of 
the integrable spin-$s$ XXZ spin chain}

We now calculate  correlation functions other than EFP 
for the spin-$s$ XXZ spin chain by the method of \S 3.3, 
making use of the formulas of Hermitian
 elementary matrices such as (\ref{eq:Em=n}).

We define the correlation function of the spin-$2s$ XXZ spin chain 
for a given product of $(2s+1) \times (2s+1)$ elementary matrices 
such as $\widetilde{E}_1^{i_1 , \, j_1 \, (2s+)} \cdots 
\widetilde{E}_m^{i_m, \, j_m \, (2s+)}$ as follows. 
\be
F^{(2s \, +)}(\{i_k, j_k \}) =  \langle \psi_g^{(2s \, +)} | 
\prod_{k=1}^{m} \widetilde{E}_k^{i_k , \, j_k \, (2s \, +)} 
|\psi_g^{(2s \, +)}  \rangle / \langle \psi_g^{(2s \, +)} 
| \psi_g^{(2s \, +)} \rangle \, . 
\label{eq:def-CF}
\ee
By the method of expressing 
spin-$s$ local operators in terms of spin-1/2 global operators 
\cite{DM1}, 
we express the $m$th product of $(2s+1) \times (2s+1)$ elementary matrices  
in terms of a $2sm$th product of $2 \times 2$ elementary matrices  
with entries $\{ \epsilon_j, \epsilon_j^{'} \}$ as follows. 
\be 
\prod_{k=1}^{m} \widetilde{E}_k^{i_k , \, j_k \, (2s \, +)} = 
C(\{i_k, j_k \}; \{ \epsilon_j^{'}, \, \epsilon_j \} ) 
\, 
\widetilde{P}_{1 2 \ldots L}^{(2s)} \, \cdot \,  
\prod_{k=1}^{2sm} 
 e_k^{\epsilon_k^{'}, \,  \epsilon_k} \, \cdot \, 
\widetilde{P}_{1 2 \ldots L}^{(2s)} . 
\ee
We evaluate the spin-$2s$ XXZ correlation function 
 $F^{(2s \, +)}(\{i_k, j_k \}) $ by 
\be
F^{(2s \, +)}(\{i_k, j_k \}) = C(\{i_k, j_k\}; 
\{ \epsilon_j^{'}, \, \epsilon_j \} ) \, 
\langle \psi_g^{(2s+)} | \,  \widetilde{P}_{1 2 \ldots L}^{(2s)} 
\, \cdot \, 
\prod_{j=1}^{2sm} e_j^{\epsilon_j^{'}, \, \epsilon_j} 
 \, \cdot \, \widetilde{P}_{1 2 \ldots L}^{(2s)} 
|\psi_g^{(2s+)}  \rangle / \langle \psi_g^{(2s+)} | \psi_g^{(2s+)}  
\rangle \, . 
\label{eq:CFe}
\ee
We denote the right-hand side of (\ref{eq:CFe}) by 
$F^{(2s \, +)}( \{ \epsilon_j, \, \epsilon_j^{'} \} )$.

Let us introduce some symbols. 
We denote by ${\bm \alpha}^{+}$ 
the set of suffices $j$ such that $\epsilon_j=0$, 
and by ${\bm \alpha}^{-}$ 
the set of suffices $j$ such that $\epsilon_j^{'}=1$:  
\be 
{\bm \alpha}^{+}= \{ j ; \, \epsilon_j=0 \}  \, , \quad 
{\bm \alpha}^{-}= \{ j ; \, \epsilon_j^{'}=1 \} \,. 
\ee 
We denote by $\alpha_+$ and $\alpha_-$ the number 
of elements of the set ${\bm \alpha}^{+}$ 
and ${\bm \alpha}^{-}$, respectively. Due to charge conservation, 
we have 
\be 
\alpha_{+} + \alpha_{-} = 2sm .  
\ee
We denote by $j_{\rm min}$ and $j_{\rm max}$ the smallest element and 
the largest element of ${\bm \alpha}^{-}$, respectively.    
We also denote by $j^{'}_{\rm min}$ and $j^{'}_{\rm max}$ 
the smallest element and the largest element 
of ${\bm \alpha}^{+}$, respectively.  

Let  $c_j$ ($j \in {\bm \alpha}^{-}$) and 
 $c_j^{'}$ ($j \in {\bm \alpha}^{+}$)
be integers such that 
 $1 \le c_j \le M$ for $j \in {\bm \alpha}^{-}$ 
and $1 \le c_j^{'} \le M+j$ for $j \in {\bm \alpha}^{+}$.  
We define sequence $(b_{\ell})_{2sm}$ by 
\be 
(b_1 , b_2, \ldots, b_{2sm}) = 
(c_{j_{\rm max}^{'}}^{'}, \ldots, c^{'}_{j_{\rm min}^{'}}, 
c_{j_{\rm min}}, \ldots, 
c_{j_{\rm max}}) \, .  \label{seq-b}  
\ee
Here sequence 
$(c_{j_{\rm max}^{'}}^{'}, \ldots, c^{'}_{j_{\rm min}^{'}}, 
c_{j_{\rm min}}, \ldots, c_{j_{\rm max}})$ is given by the 
composite of 
sequence of $c^{'}_{j}$'s in decreasing order with respect to suffix $j$, 
and sequence of $c_{j}$'s in increasing order with respect to suffix $j$.

Let us introduce the following symbols: 
\be 
\prod_{j \in {\bm \alpha}^{-}} 
\left( \sum_{c_j=1}^{M} \right)   
\prod_{j \in {\bm \alpha}^{+}} 
\left( \sum_{c_j^{'}=1}^{M+j} \right) 
=
 \sum_{c_{j_{\rm min}}=1}^{M} \cdots \sum_{c_{j_{\rm max}}=1}^{M} 
\sum_{c^{'}_{j^{'}_{\rm min}}=1}^{M+j^{'}_{\rm min}} \cdots 
\sum_{c^{'}_{j^{'}_{\rm max}}=1}^{M+j^{'}_{\rm max}} \,. 
\ee 
Extending the derivation of the spin-$s$ EFP  
we can rigorously derive the following expression of the spin-$s$ XXZ 
correlation functions in the massless regime with $0 \le \zeta < \pi/2s$.  
\bea 
& & F^{(2s \, +)}(\{ \epsilon_j, \epsilon_j^{'} \}) 
 =   \prod_{j \in {\bm \alpha}^{-}} 
\left( \sum_{c_j=1}^{M} \right)   
\prod_{j \in {\bm \alpha}^{+}} 
\left( \sum_{c_j^{'}=1}^{M+j} \right) 
\quad {\rm det} M^{(2sm)}((b_{\ell})_{2sm})  
\non \\ 
& & \, \, \times  
(-1)^{\alpha_{+}} 
{\frac { \prod_{j \in {\bm \alpha}^{-}} \left( \prod_{k=1}^{j-1} 
\varphi(\lambda_{c_j} - w_k^{(2s)} + \eta) 
\prod_{k=j+1}^{2sm} \varphi(\lambda_{c_j} - w_k^{(2s)} ) \right)}
{\prod_{1 \le k < \ell \le 2sm} 
\varphi(\lambda_{b_{\ell}} - \lambda_{b_k} + \eta)} } \non \\ 
& & \, \, \times  
{\frac { \prod_{j \in {\bm \alpha}^{+}} \left( \prod_{k=1}^{j-1} 
\varphi(\lambda_{c_j^{'}} - w_k^{(2s)} - \eta) 
\prod_{k=j+1}^{2sm} \varphi(\lambda_{c_j^{'}} - w_k^{(2s)} ) \right)}
{\prod_{1 \le k < \ell \le 2sm} 
\varphi(w_{k}^{(2s)} - w_{\ell}^{(2s)})} } + O(1/N_s) \, . 
\label{finiteCF}
\eea
Here $\varphi(\lambda)= \sinh \lambda$.   
We have defined 
the $2sm \times 2sm$ matrix $M^{(2sm)}((b_{j})_{2sm})$ as follows. 
For $\ell, k = 1, 2, \ldots, 2sm$, the matrix element of $(\ell, k)$  
is given by 
\be 
\left( M^{(2sm)}((b_{j})_{2sm}) \right)_{\ell, \,  k} 
= \left\{ 
\begin{array}{cc}  
- \delta_{b_{\ell}-M, \, k} &  {\rm for} \, \, b_{\ell} > M  \\ 
\delta_{\beta(b_{\ell}), \, \beta(k)} \, \cdot \, 
\rho(\lambda_{b_{\ell}}- w_k^{(2s)} + \eta/2 ) 
/{N_s \rho_{\rm tot}(\mu_{a(b_{\ell})})} 
 &  {\rm for} \, \, b_{\ell} \le M  
\end{array}
\right.   
\ee
where $\mu_j$ denote the centers of $\lambda_j$ as follows. 
\be 
\lambda_j = \mu_j - (\beta(j)- \frac 1 2 ) \eta 
\qquad j= 1, 2, \ldots, 2sm. 
\ee
We recall that $a(j)$ and  $\beta(j)$ 
have been defined in terms of the Gauss' symbol $[\cdot]$ by 
$a(j)= [(j-1)/2s] +1$ and 
$\beta(j) = j - 2s[(j-1)/2s]$, respectively.

We remark that under the limit of sending $\epsilon$ to zero, 
the sum over variable $c_j$ is restricted up to $M$. 

%
%
\subsection{Multiple-integral representations of the spin-$s$ XXZ correlation 
function for an arbitrary product of elementary matrices}

Let us formulate matrix $S$ for the correlation function 
of an arbitrary product of elementary matrices. 
We define the $(j,k)$ element of matrix 
$S=S\left( (\lambda_j)_{2sm}; (w_j^{(2s)})_{2sm} \right)$ by   
\be 
S_{j,k} = \rho(\lambda_j - w_k^{(2s)} + \eta/2) \, 
\delta(\alpha(\lambda_j), \beta(k)) \, , \quad {\rm for} \quad 
j, k= 1, 2, \ldots, 2sm \, .  
\ee 
Here $\delta(\alpha, \beta)$ denotes the Kronecker delta. We define 
 $\alpha(\lambda_j)$ by $\alpha(\lambda_j)= \gamma$ 
if $\lambda_j= \mu_j - (\gamma - 1/2) \eta$ or $\lambda_j= w_{k}^{(2s)}$ 
where $\beta(k)=\gamma$ ($1 \le \gamma \le 2s$).  
We remark that $\mu_j$ correspond to the centers 
of complete $2s$-strings $\lambda_j$. 
We also remark that 
the above definition of matrix $S$  generalizes that of 
 (\ref{eq:matrix-elements-S}) since 
  $\alpha(\lambda_j)$ is now also defined 
also for $\lambda_j= w_k^{(2s)}$.

Let  $\Gamma_j$ be a small contour rotating counterclockwise   
around $\lambda=w_j^{(2s)}$.   
Since the ${\rm det} S$ has simple poles at $\lambda=w_j^{(2s)}$ 
with residue $1 / 2 \pi i$,  
we therefore have 
\be \int_{-\infty+ i \epsilon}^{\infty+ i \epsilon} 
{\rm det} S((\lambda_k)_{2sm}) \, 
 d \lambda_1 = 
\int_{-\infty - i \epsilon}^{\infty - i \epsilon} 
{\rm det} S((\lambda_k)_{2sm}) \, d \lambda_1 
- \oint_{\Gamma_1} {\rm det} S((\lambda_k)_{2sm}) 
\, d \lambda_1 \, .  
\ee
For sets  ${\bm \alpha}^{-}$ and ${\bm \alpha}^{+}$ 
we define ${\tilde \lambda}_j$ for $j \in {\bm \alpha}^{-}$ 
and ${\tilde \lambda}^{'}_{j}$ for $j \in {\bm \alpha}^{+}$, 
respectively,  
by the following sequence: 
\be 
({\tilde \lambda}^{'}_{j^{'}_{max}}, \ldots, 
{\tilde \lambda}^{'}_{j^{'}_{min}},  {\tilde \lambda}_{j_{min}}, 
{\tilde \lambda}_{j_{max}})
=(\lambda_1, \ldots, \lambda_{2sm}) \, . 
\ee
Thus, from the expression of  
the correlation function in terms of a finite sum (\ref{finiteCF})  
we obtain the multiple-integral representation as follows.     
\bea 
& & F^{(2s+)}(\{ \epsilon_j, \epsilon_j^{'} \}) \non \\ 
& = & 
\left( \int_{-\infty+ i \epsilon}^{\infty+ i \epsilon}
+ \cdots 
+ \int_{-\infty - i(2s-1) \zeta + i \epsilon}
^{\infty - i(2s-1) \zeta  + i \epsilon} \right)  d \lambda_1 
\cdots 
\left( \int_{-\infty+ i \epsilon}^{\infty+ i \epsilon}
+ \cdots 
+ \int_{-\infty - i(2s-1) \zeta + i \epsilon}^{\infty - i(2s-1) \zeta  + i \epsilon} 
\right)  d \lambda_{\alpha_+} 
\non \\ & & 
\left( \int_{-\infty - i \epsilon}^{\infty - i \epsilon}
+ \cdots 
+ \int_{-\infty - i(2s-1) \zeta - i \epsilon}
^{\infty - i(2s-1) \zeta  - i \epsilon} \right)  d \lambda_{\alpha_+ + 1} 
\cdots 
\left( \int_{-\infty - i \epsilon}^{\infty - i \epsilon}
+ \cdots 
+ \int_{-\infty - i(2s-1) \zeta - i \epsilon}
^{\infty - i(2s-1) \zeta  - i \epsilon} 
\right)  d \lambda_{m}
\non \\ 
& & \quad \times Q(\{ \epsilon_j, \epsilon_j^{'} \}; \lambda_1, \ldots, \lambda_{2sm}) \, {\rm det}
S(\lambda_1, \ldots, \lambda_{2sm}) \, . 
\eea
Here we have defined 
$Q(\{ \epsilon_j, \epsilon_j^{'} \}; \lambda_1, \ldots, \lambda_{2sm})$ 
in terms of small numbers 
$\epsilon_{\ell, k}$ of (\ref{eq:epsilon}) by   
\bea 
Q(\{ \epsilon_j, \epsilon_j^{'} \}; \lambda_1, \ldots, \lambda_{2sm})) & = & 
(-1)^{\alpha_{+}} 
{\frac { \prod_{j \in {\bm \alpha}^{-}} \left( \prod_{k=1}^{j-1} 
\varphi({\tilde \lambda}_{j} - w_k^{(2s)} + \eta) 
\prod_{k=j+1}^{2sm} \varphi({\tilde \lambda}_{j} - w_k^{(2s)} ) \right)}
{\prod_{1 \le k < \ell \le 2sm} 
\varphi(\lambda_{\ell} - \lambda_{k} + \eta + \epsilon_{\ell, k})} } \non \\ 
& \times &  
{\frac { \prod_{j \in {\bm \alpha}^{+}} \left( \prod_{k=1}^{j-1} 
\varphi({\tilde \lambda}^{'}_{j} - w_k^{(2s)} - \eta) 
\prod_{k=j+1}^{2sm} \varphi({\tilde \lambda}^{'}_{j} - w_k^{(2s)} ) \right)}
{\prod_{1 \le k < \ell \le 2sm} 
\varphi(w_{k}^{(2s)} - w_{\ell}^{(2s)})} } \, . 
\eea
Thus, correlation functions (\ref{eq:def-CF}) are expressed 
in the form of a single term of multiple integrals. 

Similarly as the symmetric spin-$s$ EFP, 
we can show the symmetric expression of the multiple-integral 
representations of the spin-$s$ correlation function 
as follows. 
\begin{equation}
\begin{split}
& F^{(2s \, +)}( \{ \epsilon_j, \, \epsilon_j^{'} \} )
= 
\frac{1}{\prod_{1 \leq \alpha < \beta \leq 2s}
\sinh^{m}(\beta-\alpha )\eta} \,  
\prod_{1\leq k < l \leq m}
\frac{\sinh^{2s}(\pi(\xi_k-\xi_l)/\zeta)}
{\prod^{2s}_{j=1}\prod^{2s}_{r=1}\sinh(\xi_k-\xi_l+(r-j)\eta)}   
 \\
& \sum_{\sigma \in {\cal S}_{2sm}/({\cal S}_m)^{2s} } 
 \prod^{\alpha_+}_{j=1} 
\left( 
\int^{\infty+ i \epsilon }_{-\infty + i \epsilon} + 
\cdots + 
\int^{\infty - i(2s-1) \zeta + i \epsilon }_{-\infty  - i(2s-1) \zeta  
+ i  \epsilon} 
\right)
d \mu_{\sigma j} 
 \, 
\prod^{2sm}_{j=\alpha_+ +1}
\left( 
\int^{\infty - i \epsilon }_{-\infty - i \epsilon} + 
\cdots 
+ \int^{\infty - i(2s-1) \zeta - i \epsilon }_{-\infty  - i(2s-1) \zeta  
- i  \epsilon} 
\right) d \mu_{\sigma j} 
\non \\  
&
\times ({\rm sgn} \, \sigma) \, 
Q(\{ \epsilon_j, \epsilon_j^{'} \}; \lambda_{\sigma 1}, \ldots, 
\lambda_{\sigma(2sm)})) \, 
\left( \prod^{2sm}_{j=1} 
{\frac {\prod^{m}_{b=1} \prod^{2s-1}_{\beta=1}
\sinh(\lambda_{j}-\xi_b+ \beta \eta)}
{\prod_{b=1}^{m} \cosh(\pi(\mu_{j}-\xi_b)/\zeta)}} \right) \\ 
&\times \, {\frac {i^{2sm^2}} { (2 i \zeta)^{2sm} }} \, 
\prod^{2s}_{\gamma=1} 
\prod_{1 \le b < a \le m}
\sinh(\pi(\mu_{2s(a-1)+\gamma}-\mu_{2s(b-1)+\gamma})/\zeta) \, . 
\label{eq:CFF2}
\end{split}
\end{equation}
It is straightforward to take the homogeneous limit: 
$\xi_k \rightarrow 0$. 
Here we recall that  (${\rm sgn} \, \sigma$) 
denotes the sign of permutation 
$\sigma \in {\cal S}_{2sm}/({\cal S}_m)^{2s}$.

\subsection{An example of the spin-1 correlation function}

 Applying formula (\ref{eq:Em=n}) to the spin-1 case with $m=n=1$, 
we have 
\bea  
\widetilde{ E}_{1}^{1, \, 1 \, (2 +)} 
& = & 2 \, 
\widetilde{P}^{(2s)}_{1 \cdots L} 
\,  
D^{(1+)}(w_1) A^{(1+)}(w_2)  
\prod_{\alpha=3}^{2 N_s} 
(A^{(1+)}+D^{(1+)})(w_{\alpha}) \, \,   
\widetilde{P}^{(2s)}_{1 \cdots L}
\, .  \label{eq:E11}
\eea
Therefore, we evaluate it sending 
$\epsilon$ to zero, as follows. 
\bea 
& & \la {\psi}_g^{(2 \, +)} | \widetilde{E}_1^{1 , \, 1 \, (2 +)} 
| {\psi}_g^{(2 \, +)} \ra / 
\la {\psi}_g^{(2 \, +)} | {\psi}_g^{(2 \, +)} \ra  
\non \\ 
& = & 2 \, \lim_{\epsilon \rightarrow 0} 
\la \psi_g^{(2 \, +; \, \epsilon)} | D^{(2+; \epsilon)}(w_1^{(2; \epsilon)}) 
A^{(2+; \epsilon)}(w_2^{(2; \epsilon)})  
\prod_{\alpha=3}^{2 N_s} 
(A^{(2+; \epsilon)}+D^{(2+; \epsilon)})(w_{\alpha}^{(2; \epsilon)}) 
|\psi_g^{(2+; \epsilon)} \ra / 
\la \psi_g^{(2+; \epsilon)} | \psi_g^{(2+; \epsilon)} \ra  
\non \\ 
& = &  2 \, 
\left( \int_{-\infty+ i \epsilon}^{\infty+ i \epsilon}
+ \int_{-\infty - i \zeta + i \epsilon}^{\infty - i \zeta  + i \epsilon}
\right) d \lambda_1 
 \, 
\left( \int_{-\infty - i \epsilon}^{\infty - i \epsilon}
+ \int_{-\infty - i \zeta - i \epsilon}^{\infty - i \zeta - i \epsilon}
\right) d \lambda_2 \, 
Q(\lambda_1, \lambda_2) \,  {\rm det} S(\lambda_1, \lambda_2) 
\label{eq:E11-int}
\eea
where $Q(\lambda_1, \lambda_2)$ is given by 
\be 
Q(\lambda_1, \lambda_2) = 
(-1) {\frac 
{\varphi(\lambda_2 - \xi_1 + \eta) \varphi(\lambda_1 - \xi_1 - \eta) }  
{\varphi(\lambda_2- \lambda_1 + \eta + \epsilon_{2,1}) \varphi(\eta)}}    
\ee
and matrix $S(\lambda_1, \lambda_2)$ is given by  
\be 
S(\lambda_1, \lambda_2) = 
\left( 
\begin{array}{cc} 
\rho(\lambda_1-w_1^{(2)} + \eta/2) \delta(\lambda_1, 1) & 
\rho(\lambda_1-w_2^{(2)} + \eta/2) \delta(\lambda_1, 2) \\ 
\rho(\lambda_2-w_1^{(2)} + \eta/2) \delta(\lambda_2, 1) & 
\rho(\lambda_2-w_2^{(2)} + \eta/2) \delta(\lambda_2, 2)
\end{array} 
\right) \, . 
\ee
We thus note note that 
the correlation function is now expressed 
in terms of a single product of the multiple-integral   
representation. 

Let us now evaluate the double integral (\ref{eq:E11-int}), explicitly.   
The integral over $\lambda_1$ is decomposed into the following:
\bea 
& & \left( \int_{-\infty+ i \epsilon}^{\infty+ i \epsilon}
+ \int_{-\infty - i \zeta + i \epsilon}^{\infty - i \zeta  + i \epsilon}
\right) d \lambda_1 \, Q(\lambda_1, \lambda_2) \,  {\rm det} S(\lambda_1, \lambda_2)  \non \\ 
& = &
\left( \int_{-\infty - i \zeta/2 }^{\infty  - i \zeta/2} 
+ \int_{-\infty - i 3 \zeta/2}^{\infty - i 3 \zeta/2}
\right) d \lambda_1 \, Q(\lambda_1, \lambda_2) \,  {\rm det} S(\lambda_1, \lambda_2) \non \\ 
&  &  
- \oint_{\Gamma_1}  d \lambda_1 \, Q(\lambda_1, \lambda_2) \,  {\rm det} S(\lambda_1, \lambda_2) - \oint_{\Gamma_2}  d \lambda_1 \, Q(\lambda_1, \lambda_2) \,  {\rm det} S(\lambda_1, \lambda_2) \, . 
\eea
Thus, the integral (\ref{eq:E11-int}) is calculated as 
\bea 
& & \la \widetilde{ \psi}_g^{(2 \, +)} | 
\widetilde{ E}_1^{11 \, (2 \, +)} | \widetilde{ \psi}_g^{(2 \, +)} \ra / 
\left( 2 \la \widetilde{ \psi}_g^{(2 \, +)} 
| \widetilde{ \psi}_g^{(2 \, +)} \ra  \right) 
\non \\ 
& = &  - 2 \pi i 
\int_{-\infty}^{\infty} 
{\frac 
{\sinh(x - \eta/2) \sinh(x - 3\eta/2)} {\sinh \eta} } \rho^2(x) \, d x   
+ 2 \cosh \eta \, \int_{-\infty}^{\infty} 
{\frac {\sinh(x - \eta/2)} {\sinh(x + \eta/2)} } \rho(x) \, d x 
\non \\ 
& & - \int_{-\infty}^{\infty} \rho(x) d x  
+ (-1) 2 \cosh \eta \, \int_{-\infty}^{\infty} 
{\frac {\sinh(x - \eta/2) }
{\sinh(x + \eta/2) }} \rho(x) \, d x  \non \\ 
& = & 
{\frac {\cos \zeta (\sin \zeta - \zeta \cos \zeta)} 
{ 2 \zeta \sin^2 \zeta}} \, . 
\eea
We have thus confirmed (\ref{eq:EPofE11})  directly evaluating  
the integrals. 

%
%
%
%
 \setcounter{equation}{0} 
 \renewcommand{\theequation}{7.\arabic{equation}}
%
%
\section{Concluding remarks}

In the paper we have explicitly shown  
the multiple-integral representation of the emptiness formation 
probability for the integrable spin-$s$ XXZ spin chain 
in a region of the massless regime 
of $\eta = i \zeta$ with $0 \le \zeta < \pi/2s$. 
We have also calculated the emptiness formation 
probability for the homogeneous case 
of the integrable spin-$s$ XXZ spin chain. 

In the XXX limit where we send $\zeta$ to zero,  
the expression of EFP for the spin-$s$ XXZ case reduces to 
that of the spin-$s$ XXX case. 

Moreover, we have presented a formula 
for the multiple-integral representation 
of the spin-$s$ XXZ correlation function 
of an arbitrary product of elementary 
matrices  in the massless regime 
where $\eta = i \zeta$ with $0 \le \zeta < \pi/2s$. 
We have also presented the symmetric expression 
of the multiple-integral representations of the 
spin-$s$ XXZ correlation functions. 

Finally, we have introduced conjugate vectors 
$\widetilde{|| 2s, n \rangle}$ in order to formulate 
Hermitian elementary matrices $\widetilde{E}^{m, \, n \, (2s \, +)}$ 
and Hermitian projection operators $\widetilde{P}^{(\ell)}$ 
in the massless regime. 
We have also defined the massless fusion $R$-matrices 
$\widetilde{R}^{(\ell, \, 2s \, w)}$ for $w=+$ and $p$.

\section*{Acknowledgment} 

One of the authors (T.D.) would like to thank B.M. McCoy 
and other participants of the workshop 
on the integrable chiral Potts model, organized by M.T. Batchelor, 
Kioloa, NSW, Australia,  Dec. 7-11, 2008, for many useful comments.  
He is grateful to K. Motegi, J. Sato, and  Y. Takeyama 
for bringing him useful references.  
Furthermore, the authors would like to thank 
S. Miyashita for encouragement and keen interest in this work.  
This work is partially supported by 
Grant-in-Aid for Scientific Research (C) No. 20540365.


\appendix

%
%
\setcounter{equation}{0} 
 \renewcommand{\theequation}{A.\arabic{equation}}
\section{Affine quantum group 
with homogeneous grading}

The affine quantum algebra $U_q(\widehat{sl_2})$ 
is an associative algebra over ${\bf C}$ generated by  
$X_i^{\pm}, K_i^{\pm}$ for $i=0,1$ with the following relations: 
\bea 
K_i K_i^{-1} & = & K^{-1}_i K_i = 1 \, , \quad 
K_i X_i^{\pm} K_i^{-1}  =  q^{\pm 2} X_i^{\pm} \, ,  \quad 
K_i X_j^{\pm} K_i^{-1}  =  q^{\mp 2} X_j^{\pm}  
\quad (i \ne j) \, , \non \\
{[} X_i^{+}, X_j^{-} {]} & = & \delta_{i,j} \, 
{\frac   {K_i - K_i^{-1}}  {q- q^{-1}} } \, ,     \non \\
(X_i^{\pm})^{3} X_j^{\pm} & - & [3]_q \, (X_i^{\pm})^{2} X_j^{\pm} X_i^{\pm}   
+ [3]_q \, X_i^{\pm} X_j^{\pm} (X_i^{\pm})^2 -   
 X_j^{\pm} (X_i^{\pm})^3 = 0 \quad (i \ne j) \, . 
\label{eq:defrl}
\eea
Here the symbol $[n]_q$ denotes the $q$-integer of an integer $n$: 
\be 
[n]_q = {\frac {q^n - q^{-n}} {q - q^{-1}}} \, . \label{eq:q-integer} 
\ee
The algebra $U_q(\widehat{sl_2})$ is also a Hopf algebra over ${\bf C}$ 
with comultiplication
\bea 
\Delta (X_i^{+}) & = & X_i^{+} \otimes 1 + K_i \otimes X_i^{+}  \, , 
 \quad 
\Delta (X_i^{-})  =  X_i^{-} \otimes K_i^{-1} + 1 \otimes X_i^{-} \, ,  
\non \\
\Delta(K_i) & = & K_i \otimes K_i  \, , 
\eea 
and antipode:  
$S(K_i)=K_i^{-1} \, , S(X_i)= - K_i^{-1} X_i^{+} \, , 
S(X_i^{-}) = - X_i^{-} K_i $, and counit: 
$\varepsilon(X_i^{\pm})=0$ and $\varepsilon(K_i)=1$ for $i=0, 1$ .

%
%

The algebra $U_q(sl_2)$ is given by the Hopf subalgebra of 
$U_q(\widehat{sl_2})$ generated by 
$X_i^{\pm}$, $K_i$ with either $i=0$ or $i=1$. Hereafter 
we denote by $X^{\pm}$ and $K$ the generators of $U_q(sl_2)$. 


For a given complex number $\lambda$ we define 
a homomorphism of algebras  $\varphi_{\lambda}$: 
$U_q(\widehat{sl_2}) \rightarrow U_q({sl_2})$.     
\be 
\varphi_{\lambda}(X_0^{\pm})  =  
e^{\pm 2 \lambda} \, X^{\mp } \, , \quad 
\varphi_{\lambda}(X_1^{\pm})  =  X^{\pm} \, , \quad  
\varphi_{\lambda}(K_0) = K^{-1} \, ,    
\quad \varphi_{\lambda}(K_1)  = K \, \, .    
\label{eq:eval-h}  
\ee
Map (\ref{eq:eval-h}) is associated 
with  homogeneous grading \cite{Jimbo-Miwa}. 
For a  representation $(\pi, V^{(\ell)})$ of $U_q(sl_2)$ 
we  have a representation of $U_q(\widehat{sl_2})$ by 
$\pi(\varphi_{\lambda}(a))$ for $a \in U_q(\widehat{sl_2})$.   
We call it the spin-$\ell/2$ evaluation representation 
with evaluation parameter $\lambda$,  
and denote it 
by $(\pi_{\lambda}, V^{(\ell)}(\lambda))$ or $V^{(\ell)}(\lambda)$. 

We define opposite coproduct 
$\Delta^{op}$ by 
\be 
\Delta^{op}(a)= \tau \circ \Delta(a) \quad \mbox{\rm for} \, \, 
a \in U_q(sl_2) \, , 
\ee 
where $\tau$ denotes the permutation 
operator: $\tau( a \otimes b) = b \otimes a$ for $a, b \in U_q(sl_2)$.

%
%
\setcounter{equation}{0} 
 \renewcommand{\theequation}{B.\arabic{equation}}
\section{Fusion projection operators being idempotent}

We give the derivation \cite{V-DWA} of 
$(P_{1 2 \cdots \ell}^{(\ell)})^2=P_{1 2 \cdots \ell}^{(\ell)}$ 
making use of the Yang-Baxter equations. 

\begin{lemma} Operators $P_{1 2 \cdots \ell}^{(\ell)}$ 
defined by (\ref{eq:def-projector}) have  
 the following two expressions:    
\be 
P_{1 \, 2 \cdots \ell-1}^{(\ell-1)} 
{\check R}^{+}_{\ell-1, \, \ell}((\ell-1) \eta) 
P_{1 \, 2 \cdots \ell-1}^{(\ell-1)} = 
P_{2 \, 3 \cdots \ell}^{(\ell-1)} 
{\check R}_{1, \, 2}^{+}((\ell-1) \eta) 
P_{2 \, 3 \cdots \ell}^{(\ell-1)} \, . 
\label{eq:two}
\ee
\end{lemma} 
\begin{proof} 
Applying notation (\ref{defAjk}) to permutation operator $\Pi_{1, 2}$ 
we define permutation operators $\Pi_{j, \, k}$ for integers $j$ and $k$ 
satisfying $0 \le j < k \le L$. 
The form of the left-hand side of (\ref{eq:two}) 
is expressed in terms of $R$-matrices as follows 
(see also eq. (3.7) of \cite{DM1}).  
\be 
P_{1 \cdots \ell}^{(\ell)} =   
\prod_{j=1}^{[\ell/2]} \Pi_{j, \, \ell-j+1} \, \cdot \, 
R^{+}_{\ell-1 \, \ell} \cdots R^{+}_{2, \, 3 \cdots \ell} 
R^{+}_{1, \, 2 \cdots \ell} \, . 
\label{eq:prod-R}
\ee
Making use of the Yang-Baxter equations (\ref{eq:YBE}) 
we reformulate (\ref{eq:prod-R}) as follows.  
\be 
P_{1 \cdots \ell}^{(\ell)} =   
\prod_{j=1}^{[\ell/2]} \Pi_{j, \, \ell-j+1} \, \cdot \, 
R^{+}_{1 \, 2} R^{+}_{1 2, \, 3} \cdots R^{+}_{1, \, 2 \cdots \ell} \, 
\ee
which gives the expression of the right-hand side of (\ref{eq:two}).  
\end{proof} 

From (\ref{eq:two}) we show that  
$P_{j+1 \, j+2 \cdots \, j+\ell}^{(\ell)}$ is expressed as follows.   
\be 
P_{j+1 \, j+2 \cdots j+\ell-1}^{(\ell-1)} 
{\check R}^{+}_{j+\ell-1, \, j+\ell}((\ell-1) \eta) 
P_{j+1 \, j+2 \cdots j+\ell-1}^{(\ell-1)} = 
P_{j+2 \, j+3 \cdots j+\ell}^{(\ell-1)} 
{\check R}^{+}_{j+1, \, j+2}((\ell-1) \eta) 
P_{j+2 \, j+3 \cdots j+\ell}^{(\ell-1)} \, . 
\label{eq:two-j}
\ee

\begin{lemma} Operator $P_{1 2 \cdots \ell}^{(\ell)}$ projects   
operator ${\check R}_{\ell-1, \, \ell}(u)$ to 1 as follows. 
\be 
P_{1 2 \cdots \ell}^{(\ell)} 
{\check R}^{+}_{\ell-1, \, \ell}(u) =  P_{1 2 \cdots \ell}^{(\ell)} \, . 
\label{eq:eigen1}
\ee
\end{lemma} 
\begin{proof} 
Due to the spectral decomposition of the $R$-matrix,  we have 
$ P^{(2)}_{\ell-1, \ell} {\check R}^{+}_{\ell-1, \, \ell}(u) 
= P^{(2)}_{\ell-1, \ell} $. 
Applying (\ref{eq:two}) and (\ref{eq:two-j}),
 we thus obtain (\ref{eq:eigen1}). 
\end{proof} 

\begin{proposition} Operators $P_{1 2 \cdots \ell}^{(\ell)}$ 
defined by (\ref{eq:def-projector}) 
are idempotent:  
$\left( P_{1 2 \cdots \ell}^{(\ell)} \right)^2 
= P_{1 2 \cdots \ell}^{(\ell)}$ . 
\end{proposition}
\begin{proof}
We show it from (\ref{eq:eigen1}) by induction on 
$\ell$. Suppose that it is idempotent 
for $\ell$.  We have  
\bea 
\left( P_{1 2 \cdots \ell+1}^{(\ell+1)} \right)^2  
& = & 
P_{1 2 \cdots \ell}^{(\ell)} \, 
R^{+}_{\ell \, \ell+1}(\ell \eta) \, \cdot 
(P_{1 2 \cdots \ell}^{(\ell)})^2 \, \cdot  
R^{+}_{\ell \, \ell+1}(\ell \eta) \, 
P_{1 2 \cdots \ell}^{(\ell)} \non \\ 
& = & 
P_{1 2 \cdots \ell}^{(\ell)} \, 
R^{+}_{\ell \, \ell+1}(\ell \eta) \, \cdot  \, 
P_{1 2 \cdots \ell}^{(\ell)} \, 
R^{+}_{\ell \, \ell+1}(\ell \eta) \, \cdot \, 
P_{1 2 \cdots \ell}^{(\ell)} \non \\ 
& = & 
P_{1 2 \cdots \ell}^{(\ell)} \, 
R^{+}_{\ell \, \ell+1}(\ell \eta) \, 
P_{1 2 \cdots \ell}^{(\ell)} \cdot  \,  
P_{1 2 \cdots \ell}^{(\ell)} 
 =  
P_{1 2 \cdots \ell}^{(\ell)} \, 
R^{+}_{\ell \, \ell+1}(\ell \eta) \, 
P_{1 2 \cdots \ell}^{(\ell)} \, .  
\eea
\end{proof}

%
%
\setcounter{equation}{0} 
 \renewcommand{\theequation}{C.\arabic{equation}}
\section{Basis vectors of spin-$\ell/2$ representation of $U_q(sl_2)$}

In terms of  
the $q$-integer $[n]_q$ defined in (\ref{eq:q-integer}),  
we define the $q$-factorial  $[n]_q!$ for integers $n$ by 
\be 
[n]_q ! = [n]_q \, [n-1]_q \, \cdots \, [1]_q \, .  
\ee
For integers $m$ and $n$ satisfying $m \ge n \ge 0$ 
we define the $q$-binomial coefficients as follows
\be 
\left[ 
\begin{array}{c} 
m \\ 
n 
\end{array}  
 \right]_q 
= {\frac {[m]_q !} {[m-n]_q ! \, [n]_q !}}  \, . 
\ee 

We now define the basis vectors of the $(\ell+1)$-dimensional 
irreducible representation of $U_q(sl_2)$,  
$|| \ell, n  \rangle$ for $n=0, 1, \ldots, \ell$ as follows. 
We define $||\ell, 0 \rangle$ by 
\be 
||\ell , 0 \rangle = |0 \rangle_1 \otimes |0 \rangle_2 \otimes 
\cdots \otimes |0 \rangle_\ell \, . 
\ee   
Here  $|\alpha \rangle_j$ for $\alpha=0, 1$ 
denote the basis vectors of the spin-1/2 representation defined 
on the $j$th position in the tensor product.  We define 
$|| \ell, n \rangle$ for $n \ge 1$ and evaluate them as follows \cite{DM1} .   
\bea 
|| \ell, n \rangle  & =  &  
\left( \Delta^{(\ell-1)} (X^{-}) \right)^n ||\ell, 0 \rangle \,  
{\frac 1 {[n]_q!}} \non \\ 
%
& = &
\sum_{1 \le i_1 < \cdots < i_n \le \ell} 
\sigma_{i_1}^{-} \cdots \sigma_{i_n}^{-} | 0 \rangle \, 
q^{i_1+ i_2 + \cdots + i_n  - n \ell + n(n-1)/2} \, . 
\label{eq:|ell,n>} 
\eea
We define the conjugate vectors explicitly by the following: 
\be 
\langle \ell, n || =  
\left[ 
\begin{array}{c} 
\ell \\ 
n 
\end{array}  
 \right]_q^{-1} \, q^{n(\ell-n)} \, 
\sum_{1 \le i_1 < \cdots < i_n \le \ell} 
\langle 0 | \sigma_{i_1}^{+} \cdots \sigma_{i_n}^{+} \, 
q^{i_1 + \cdots + i_n - n \ell + n(n-1)/2}   \, . 
\label{eq:<ell,n|}
\ee
It is easy to show the normalization conditions \cite{DM1}: 
$\langle \ell, n || \, || \ell, n \rangle = 1$.  
In the massive regime where $q = \exp \eta$ with real $\eta$, 
conjugate vectors $\langle \ell, n || $ are Hermitian conjugate to vectors 
$|| \ell, n \rangle$.

Through the recursive construction (\ref{eq:def-projector}) 
of $P^{(\ell)}$s, 
it is easy to show the following \cite{DM1}: 
\bea 
P^{(\ell)}_{1 2 \cdots \ell} || \ell, n \rangle 
& = & || \ell, n \rangle \, , \non \\ 
\langle \ell, n || P^{(\ell)}_{1 2 \cdots \ell}  & = & \langle \ell, n ||  
\, . \label{eq:consistency-P}
\eea
Thus, the fusion projector $P^{(\ell)}$ 
is consistent with the spin-$\ell/2$ representation 
of $U_q(sl_2)$. 

In order to define Hermitian elementary matrices, 
we now introduce another set of dual basis vectors.   
For a given nonzero integer $\ell$ we define  
$\widetilde{\langle \ell, n ||}$ for $n=0, 1, \ldots, n$, by 
\be 
\widetilde{\langle \ell, n ||} =  
\left( 
\begin{array}{c} 
\ell \\ 
n 
\end{array}  
 \right)^{-1} \,  
\sum_{1 \le i_1 < \cdots < i_n \le \ell} 
\langle 0 | \sigma_{i_1}^{+} \cdots \sigma_{i_n}^{+} \, 
q^{-(i_1 + \cdots + i_n) + n \ell - n(n-1)/2}   \, . 
\label{eq:tilde<ell,n|}
\ee
They are conjugate to $|| \ell, n \rangle$:   
$\widetilde{\langle \ell, m ||} \, || \ell, n \rangle = 
\delta_{m, n} $ . 

In the massless regime where $|q|=1$,  
matrix $|| \ell, n \rangle  \widetilde{\langle \ell, n ||}$ is 
Hermitian: $(|| \ell, n \rangle  \widetilde{\langle \ell, n ||})^{\dagger} 
= || \ell, n \rangle  \widetilde{\langle \ell, n ||}$. 
However, in order to define projection operators $\tilde{P}$ such that 
$P \tilde{P} = P$, we define another set of vectors 
 $\widetilde{|| \ell, n \rangle}$ in section 2.4. They are conjugate to 
the dual vectors $\langle \ell, n ||$.

%
%
\setcounter{equation}{0} 
 \renewcommand{\theequation}{D.\arabic{equation}}
\section{The massless fusion $R$-matrices of the spin-1 case}

Let us evaluate the matrix elements of the massless monodromy matrix 
$\widetilde{ T}^{(1, \, 2 \, +)}_{0, \, 1}(\lambda_0; \xi_1)$,  
i.e. the massless $L$ operator of the spin-1 representation.  
\be 
\widetilde{ T}^{(1, \, 2 \, +)}_{0, \, 1}(\lambda_0; \xi_1)
= \widetilde{P}^{(2)}_{12} 
R^{(1, \, 1 \, +)}_{0, \, 1 \, 2}(\lambda_0; \{ w_j^{(2)} \}_2)  
\widetilde{P}^{(2)}_{12} 
\, . 
\ee
Here we have set inhomogeneous parameters  $w_1^{(2)}= \xi_1$ and 
$w_2^{(2)}= \xi_1 - \eta$. 
Let us recall $R_{0, 12}^{+}= R_{0, 2}^{+}R_{0,1}^{+}$.  
For instance, we have $A_{12}^{+} = A_2^{+} A_1^{+} + B_2^{+} C_1^{+}$. 
In terms of $b_{0j}=b(\lambda_0- w_j^{(2)})$ and 
$c_{0j}=c(\lambda_0- w_j^{(2)})$ for $j=1, 2$, 
the (1,1) element of $\widetilde{A}_{1}^{(2+)}$ is given by 
\be 
\langle 2, 1 || A^{(2+)}(\lambda_0) \widetilde{||2, 1 \rangle}
= (b_{01} + b_{02} + q^{-2} c_{01} c_{02})/2 \, .  
\ee
Thus, setting $u=\lambda_0 - \xi_1$,   
all the non-zero matrix elements of 
$\widetilde{T}^{(1 , \, 2 \, +)}(\lambda_0)$
are given by  
\bea 
\langle 2, 0 || A^{(2+)}(\lambda_0) \widetilde{||2, 0 \rangle} & = & 
\langle 2, 2 || D^{(2+)}(\lambda_0) \widetilde{||2, 2 \rangle} 
= 1 \, , 
\non \\   
\langle 2, 1 || A^{(2+)}(\lambda_0) \widetilde{||2, 1 \rangle} 
& = &  \langle 2, 1 || D^{(2+)}(\lambda_0) \widetilde{||2, 1 \rangle} 
 =  \sinh(u+\eta)/\sinh(u+2 \eta) \, ,  \non \\     
\langle 2, 2 || A^{(2+)}(\lambda_0) \widetilde{||2, 2 \rangle} 
& = & \langle 2, 0 || D^{(2+)}(\lambda_0) \widetilde{||2, 0 \rangle}  
= \sinh u/\sinh(u+2 \eta) \, ,  \non \\  
\langle 2, 1 || B^{(2+)}(\lambda_0) \widetilde{||2, 0 \rangle} 
& = & e^{-u} \sinh \eta /\sinh(u+2 \eta) \, ,  \non \\     
\langle 2, 2 || B^{(2+)}(\lambda_0) \widetilde{||2, 1 \rangle} 
  & = & [2]_q \, q^{-1} e^{-u} \sinh \eta/\sinh(u+2 \eta) \, ,  \non \\  
\langle 2, 0 || C^{(2+)}(\lambda_0) \widetilde{||2, 1 \rangle} 
& = & [2]_q \,  e^{u} \sinh \eta /\sinh(u+2 \eta) \, ,  \non \\     
\langle 2, 1 || C^{(2+)}(\lambda_0) \widetilde{||2, 2 \rangle} 
  & = & q e^{u} \sinh \eta/\sinh(u+2 \eta) \, .  \label{eq:MEofL} 
\eea

We should remark that the massive monodromy matrix 
${T}^{(1, \, 2 \, +)}_{0, \, 1}(\lambda_0; \xi_1)$ has   
the same matrix elements as the massless monodromy matrix 
$\widetilde{ T}^{(1, \, 2 \, +)}_{0, \, 1}(\lambda_0; \xi_1)$. 
For instance, we calculate the (1,1) element of operator 
$A^{(2+)}_1$ as follows.  
\be   
\langle 2, 1 || A^{(2+)}(\lambda_0) ||2, 1 \rangle 
 =  b_{01} q^{-2} + b_{02} + q^{-2} c_{01}c_{02}  
  =   \sinh(u+\eta)/\sinh(u+2 \eta) \, . 
\ee 

Let us define the matrix elements of the massless 
fusion $R$ matrix of type (2, 2) 
as follows.  
\be 
\widetilde{R}^{(2, \, 2 \, +)}_{0 , \, 1}(\lambda_0-\xi_1)
^{b_0 \, b_1}_{c_0 \, c_1}
= \, _1 \langle 2, b_1 || _a \langle 2, b_0 || 
R^{(2 , \, 2 \, +)}_{0 , \, 1}(\lambda_0-\xi_1) 
\widetilde{|| 2 , c_0 \rangle_a} \widetilde{|| 2, c_1 \rangle_1}  \, .   
\ee 
Here $|| 2 , c_0 \rangle_a$ and  $|| 2 , c_1 \rangle_1$ 
denote vectors in the auxiliary space $V_0^{(2)}$ and 
the quantum space $V_1^{(2)}$, respectively. 
Making use of matrix elements of 
the monodromy matrix of type (1, 2) 
we derive the fusion $R$ matrix of type (2,2).  
For an illustration, let us calculate  
$\widetilde{R}^{(2 , \, 2 \, +)}_{0, 1}(u)^{1 \, 0}_{0 \, 1}$.   
\be 
_a \langle 2, 1 || 
R^{(2 \, 2 \, +)}_{0 , \, 1}(\lambda_0-\xi_1) 
\widetilde{|| 2 , 0 \rangle_a}
= \left( A_{1}^{(2+)}(\lambda)C^{(2+)}_{1}(\lambda-\eta) + q^{-1}  
C_{1}^{(2+)}(\lambda)A^{(2+)}_{1}(\lambda-\eta) \right) q/[2]_q \, . 
\ee 
Evaluating operators $A^{(2+)}_1$ and $C^{(2+)}_1$ 
in the quantum space $V_1^{(2)}$, we have 
\be 
R^{(2 \, 2 \, +)}_{0, 1}(u)^{1 \, 0}_{0 \, 1}
= {\frac {[2]_q e^u \sinh \eta} {\sinh(u + 2 \eta)}} \, . 
\ee
The fusion $R$-matrix becomes permutation $\Pi_{0, \, 1}$ at $u=0$. 
In fact, we have  
$R^{(2 , \, 2 \, +)}_{0, 1}(0)^{1 \, 0}_{0 \, 1}=1$.

%
%
\setcounter{equation}{0} 
 \renewcommand{\theequation}{E.\arabic{equation}}
\section{Spin-$s$ elementary matrices  
in global operators}
%

For integers $i_k$  and $j_k$ with 
$1 \le i_1 < \cdots < i_m \le \ell$ and 
$1 \le j_1 < \cdots < j_n \le \ell$,  
we have 
\bea  
\widetilde{ || \ell, m \rangle} \langle \ell, n || 
& = & 
\left(
\begin{array}{c} 
\ell \\
n 
\end{array} 
 \right) \, 
\left[
\begin{array}{c} 
\ell \\
m 
\end{array} 
 \right]_q \, 
\left[
\begin{array}{c} 
\ell \\
n 
\end{array} 
 \right]_q^{-1} \, 
q^{-(i_1 + \cdots + i_m) + ( j_1 + \cdots + j_n)
+ m(m+1)/2 - n(n+1)/2}  \non \\ 
& & \times \, 
\widetilde{P}_{1 \cdots \ell}^{(\ell)} 
\left( \prod_{k=1}^{m} e_{i_k}^{1, \, 0}
 \cdot \prod_{p=1; p \ne i_1, \ldots, i_m, j_1, \ldots, j_n}^{\ell} 
 e_{p}^{0, \, 0} \cdot 
\prod_{q=1}^{n} e_{j_q}^{0, \, 1} \right) 
\widetilde{P}_{1 \cdots \ell}^{(\ell)}  \, . 
\label{eq:ellmn} 
\eea  
Applying the spin-1/2 formulas of 
QISP \cite{KMT1999} to (\ref{eq:ellmn}), we can express 
any given spin-$s$ local operator  
in terms of the spin-1/2 global operators.   
It is parallel to the massive case \cite{DM1}. 
Let us set $i_1=1, i_2=2$, \ldots, $i_m=m$ and 
$j_1=1, j_2=2$, \ldots, $j_n=n$ in (\ref{eq:ellmn}). For $m > n$ 
we have 
\bea  
& & \widetilde{E}_{i}^{m, \, n \, (\ell \, + )} 
 = \left(
\begin{array}{c} 
\ell \\
n 
\end{array} 
 \right) \, 
\left[
\begin{array}{c} 
\ell \\
m 
\end{array} 
 \right]_q \, 
\left[
\begin{array}{c} 
\ell \\
n 
\end{array} 
 \right]_q^{-1} \, 
\widetilde{P}^{(\ell)}_{1 \cdots L} 
\,  
\prod_{\alpha=1}^{(i-1) \ell} (A^{(1+)}+D^{(1+)})(w_{\alpha}) 
\prod_{k=1}^{n} D^{(1+)}(w_{(i-1)\ell+k} ) 
\non \\ 
& &  \times \prod_{k=n+1}^{m} B^{(1+)}(w_{(i-1)2s+k})  
\, \prod_{k=m+1}^{\ell} A^{(1+)}(w_{(i-1) \ell+k})  
\prod_{\alpha=i \ell +1}^{\ell N_s} 
(A^{(1+)}+D^{(1+)})(w_{\alpha}) \, \,   
P^{(\ell)}_{1 \cdots L}
\, .  \label{eq:Em>n}
\eea
For $m < n$ we have 
\bea  
& & \widetilde{E}_{i}^{m, \, n \, (\ell \, + )} 
 = \left(
\begin{array}{c} 
\ell \\
n 
\end{array} 
 \right) \, 
\left[
\begin{array}{c} 
\ell \\
m 
\end{array} 
 \right]_q \, 
\left[
\begin{array}{c} 
\ell \\
n 
\end{array} 
 \right]_q^{-1} \, 
\widetilde{P}^{(\ell)}_{1 \cdots L} 
\,  
\prod_{\alpha=1}^{(i-1) \ell} (A^{(1+)}+D^{(1+)}(w_{\alpha}) 
\prod_{k=1}^{m} D^{(1+)}(w_{(i-1)\ell+k} ) 
\non \\ 
& &  \times \prod_{k=m+1}^{n} C^{(1+)}(w_{(i-1)2s+k})  
\, \prod_{k=m+1}^{\ell} A^{(1+)}(w_{(i-1) \ell+k})  
\prod_{\alpha=i \ell +1}^{\ell N_s} 
(A^{(1+)}+D^{(1+)})(w_{\alpha}) \, \,   
P^{(\ell)}_{1 \cdots L}
\, .  \label{eq:Em>n}
\eea

%
%
\setcounter{equation}{0} 
 \renewcommand{\theequation}{F.\arabic{equation}}
\section{Derivation of the density of string centers}

In terms of shifted rapidities ${\tilde \lambda}_{A}$ with 
$A=2s(a-1) + \alpha$ for $a=1, 2, \ldots, N_{s}/2$ 
and $\alpha = 1, 2, \ldots, 2s$, 
the Bethe ansatz equations for the homogeneous chain 
are given by 
\be 
\left( {\frac 
{\sinh(\tilde{\lambda}_{A}  + s \eta)} 
{\sinh(\tilde{\lambda}_{A} - s \eta)} } \right)^{N_s} 
= \prod_{B=1; B \ne A}^{M}  
{\frac {\sinh(\tilde{\lambda}_{A} - \tilde{\lambda}_{B} + \eta)} 
       {\sinh(\tilde{\lambda}_{A} - \tilde{\lambda}_{B} - \eta)}} 
\, , \quad \mbox{for} \, \, A=1, 2, \ldots, M \,. \label{eq:BAE3}
\ee
Putting $\lambda_A = \mu_a - (2s + 1 - 2 \alpha)$ 
and taking the product over $\alpha$ for $\alpha=1, 2, \ldots, 2s$, 
for the left-hand side of (\ref{eq:BAE3})   
and for the right-hand side of (\ref{eq:BAE3}),  we have 
\bea 
& & (-1)^{2s} \left\{ \prod_{k=1}^{2s} 
  \left( {\frac {\sinh((k-1/2) \eta - \mu_a)}
                {\sinh((k-1/2) \eta + \mu_a)}} \right)^{N_s} \right\}^{-1}
\non \\ 
& = & (-1)^{2s+N_s/2} 
\prod_{b=1}^{{N_s}/2}
\left\{ 
{\frac {\sinh(2s \eta - (\mu_a-\mu_b))} 
       {\sinh(2s \eta + (\mu_a-\mu_b))}} 
\prod_{k=1}^{2s-1} 
\left(
{\frac {\sinh(k \eta - (\mu_a-\mu_b))} 
       {\sinh(k \eta + (\mu_a-\mu_b))}} 
\right)^2
\right\}^{-1} \, . 
\label{eq:R=L} 
\eea
Taking the logarithm of (\ref{eq:R=L}) and 
making use of the following relation  
\be 
K_{2k}(\lambda) = {\frac d {d\lambda}} \frac 1 {2\pi i} \log\left( 
{\frac {\sinh(k \eta - \lambda)}  {\sinh(k \eta + \lambda)}} \right)
\ee
we have the integral equation 
for the density of string centers, $\rho(\lambda)$, as follows.   
\be 
\rho(\lambda) = 
\sum_{j=1}^{\ell} K_{2 j -1}(\lambda) 
- \int_{-\infty}^{\infty}  \left(K_{4s}(\lambda - \mu_b) 
+ \sum_{k=1}^{2s-1} 2 K_{2k}(\lambda - \mu_b) \right) \rho(\lambda) 
d \lambda \, . 
\label{eq:rho}
\ee
For $0 < \zeta \le \pi/m$ we have the following Fourier transform: 
\be 
\int_{-\infty}^{\infty} e^{i \mu \omega} K_m(\mu) d \mu = 
{\frac {\sinh((\pi - m \zeta)\omega/2)} {\sinh(\pi \omega/2)}} \, . 
\ee
Taking the Fourier transform of (\ref{eq:rho})  we have 
the Fourier transform $\widehat{\rho}(\omega)$ of $\rho(\lambda)$ as  
follows. 
\be 
\widehat{\rho}(\omega)   =  
\left( \sum_{k=1}^{2s} \widehat{K}_{2k-1}(\omega) \right)/
\left( 1+ \widehat{K}_{4s}(\omega) + 2 \sum_{k=1}^{2s-1} 
\widehat{K}_{2k}(\omega) \right) \non \\ 
 =  \frac 1 {2 \cosh (\zeta \omega/2)} \, . 
\ee
Taking the inverse Fourier transform  
we obtain $\rho(\lambda)= 1/2 \zeta \cosh(\pi \lambda/\zeta)$.   

%
%
\setcounter{equation}{0} 
 \renewcommand{\theequation}{G.\arabic{equation}}
\section{Some formulas of the algebraic Bethe ansatz}
Applying the commutation relations between $C$ and $D$ operators 
we have 
\be 
\langle 0 | \prod_{\alpha=1}^{M} C(\lambda_{\alpha}) \prod_{j=1}^{m} 
D(\lambda_{M+j}) 
= \sum_{a_1=1}^{M+1} \sum_{a_2=1; a_1 \ne a_1}^{M+2} \cdots 
 \sum_{a_m=1; a_1 \ne a_1, \ldots, a_m }^{M+m }
G_{a_1 \cdots a_m}(\lambda_1, \cdots, \lambda_{M+m}) \non     
\ee
where 
\be 
G_{a_1 \cdots a_m}(\lambda_1, \cdots, \lambda_{M+m})  
= \prod_{j=1}^{m} \left( d(\lambda_{a_j}; \{ w_j \}_L ) 
{\frac {\prod_{b=1; b \ne a_1, \ldots, a_{j-1}}^{M+j-1} 
\sinh(\lambda_{a_j}-\lambda_b + \eta) }  
 {\prod_{b=1; b \ne a_1, \ldots, a_{j}}^{M+j} 
  \sinh(\lambda_{a_j}-\lambda_b) }}  
\right) \, . 
\label{eq:CCD}
\ee

Let $\{\lambda_k \}_M$ be a set of Bethe roots. 
We have \cite{KMT1999,KMT2000} 
\bea 
& & \langle 0 | \prod_{k=1; k \ne a_1, \ldots, a_m}^{M} 
C(\lambda_k) \prod_{j=1}^{m} C(w_{j}) 
\prod_{\gamma=1}^{M} B(\lambda_{\gamma}) \, | 0 \rangle
/\langle 0 | \, \prod_{k=1}^{n} C(\lambda_k)  
\prod_{\gamma=1}^{M} B(\lambda_{\gamma}) \, | 0 \rangle \non \\ 
& = & 
{\rm det}\left(\left(\Phi^{'}(\{ \lambda_{\alpha} \}) \right)^{-1} \, 
\Psi^{'}(\{ \lambda_{\alpha} \} \setminus 
\{ \lambda_{a_1}, \ldots, \lambda_{a_m}  \} \cup 
\{ w_1, \ldots, w_m \} ) \right) \times \prod_{j=1}^{m}  \non \\ 
& \times &  \prod_{\alpha=1; \alpha \ne a_1, \ldots, a_m}^{M} 
{\frac {\sinh(\lambda_{\alpha} - w_j + \eta)} 
{\sinh(\lambda_{\alpha} - \lambda_{a_j} + \eta)}} 
\prod_{j=1}^{m} \prod_{\alpha=1}^{M} 
{\frac {\sinh(\lambda_{\alpha}- \lambda_{a_j})}   
       {\sinh(\lambda_{\alpha}- w_{j})} } 
\prod_{1 \le j < k \le m} 
{\frac {\sinh(\lambda_{a_j}- \lambda_{a_k})}
       {\sinh(w_{j}- w_{k})}} \, .  
\label{eq:ratio}
\eea

\end{document}